\documentclass[12pt]{article}
\usepackage{amsmath,XoohmE}
\usepackage{graphicx,color,multirow,array}
\usepackage{longtable}
\usepackage{placeins}
\newcommand{\wt}{\mathop{\hbox{\rm wt}}\nolimits}
\newcommand{\img}{\mathop{\hbox{\rm Im}}\nolimits}

\newcommand{\cA}{\mathcal{A}}

\newcommand{\ddt}{\partial_\tau}
\newcommand{\Cl}{\mathop{\hbox{\rm Cl}}\nolimits}

\newcommand{\cM}{{\mathcal{M}}}
\newcommand{\sM}{{\mathscr{M}}}

\newcommand{\cY}{\mathcal{Y}}
\newcommand{\gen}{\mathbin{\,:\,}}
\newtheorem{theorem}{Theorem}[section]
\newtheorem{definition}{Definition}[section]
\newtheorem{proposition}[theorem]{Proposition}

\newtheorem{lemma}[theorem]{Lemma}
\long\def\oMit#1{}
\def\Here.{\smash{\colorbox{yellow}{\color{red}\Large$\boldsymbol*$} }}
\def\Cd#1{\left[\begin{smallmatrix}#1\end{smallmatrix}\right]}
\def\4{\hphantom{\hbox{$-$}}}
\def\vC#1{\vcenter{\hbox{\hss#1\hss}}}
\definecolor{Hey}{rgb}{1,0,.4}
\definecolor{orange}{rgb}{.8,.4,0}
\definecolor{plum}{rgb}{.4,0,.6}
\definecolor{gray}{rgb}{.6,.6,.6}
\definecolor{Green}{rgb}{0,.6,0}  \def\cG#1{{\color{Green}#1}}
\definecolor{Red}{rgb}{.8,0,0}    \def\cR#1{{\color{Red}#1}}
\definecolor{Blue}{rgb}{0,0,.8}   \def\cB#1{{\color{Blue}#1}}

\def\Ic{\hbox{\sf c\kern-3pt\rule[.2pt]{.5pt}{4pt}\kern2.5pt}}
\catcode`@=11    
\newcounter{xmpl}\resetby{section}{xmpl}
\newenvironment{example}{\par\noindent\addtocounter{xmpl}{1}%
                          \def\@currentlabel{{\thesection.\arabic{xmpl}}}%
                          {\bfseries Example~\arabic{section}.\arabic{xmpl}}%
                           ~\ignorespaces\small}
                        {\hrulefill\rule{.45pt}{1ex}\par}
\newcounter{cnst}\resetby{section}{cnst}
\newenvironment{construction}{\par\noindent\addtocounter{cnst}{1}%
                          \def\@currentlabel{{\thesection.\arabic{cnst}}}%
                          {\bfseries Construction~\thesection.\arabic{cnst}}%
                           ~\ignorespaces\slshape}
                        {\hrulefill\rule{.45pt}{1ex}\par}
\newcounter{crit}\resetby{section}{crit}

\def\Ft#1{\footnote{#1}}
\newdimen\parshift\parshift=\parindent
 \long\def\@footnotetext#1{\insert\footins{\reset@font\footnotesize\interlinepenalty%
  \interfootnotelinepenalty\splittopskip\footnotesep\splitmaxdepth\dp\strutbox%
   \floatingpenalty\@MM\hsize\columnwidth\addtolength{\hsize}{-2\parshift}
    \@parboxrestore\protected@edef\@currentlabel{\csname p@footnote\endcsname\@thefnmark}
      \color@begingroup
       \@makefntext{\rule\z@\footnotesep\ignorespaces#1\@finalstrut\strutbox\vglue1mm}
        \color@endgroup}}
 \long\def\@makefntext#1{\hglue\parshift
                         \vbox{\noindent\hb@xt@0em{\hss\@makefnmark\,}#1}}
\catcode`@=12			
 \font\rOpe=cmsy10                        
 \def\ktl{{\hbox{\rOpe\char'170}}}        
 \def\kbl{{\hbox{\rOpe\char'170}}}        
 \def\kcr{{\reflectbox{\rOpe\char'170}}}        
 \def\ktr{{\reflectbox{\rOpe\char'170}}}        
 \def\kbr{{\reflectbox{\rOpe\char'170}}}        
 \def\Border{\vbox{\hsize0pt
        \setlength{\unitlength}{1mm}
        \newcount\xco
        \newcount\yco
        \xco=-21
        \yco=12
        \begin{picture}(0,0)(-7.5,0)
        \put(\xco,\yco){$\ktl$}
        \advance\yco by-1
        {\loop
        \put(\xco,\yco){$\kcr$}
        \advance\yco by-2
        \ifnum\yco>-240
        \repeat
        \put(\xco,\yco){$\kbl$}}
        \xco=170
        \yco=12
        \put(\xco,\yco){$\ktr$}
        \advance\yco by-1
        {\loop
        \put(\xco,\yco){$\kcr$}
        \advance\yco by-2
        \ifnum\yco>-240
        \repeat
        \put(\xco,\yco){$\kbr$}}
        \put(-19.5,13){\scalebox{.54}{State University of New York
            Physics Department|University of Maryland Center for
            String and Particle  Theory \&\ Physics Department|%
            Howard University Physics \&\ Astronomy Department}}
        \put(-19.5,-241.5){\scalebox{.486}{University of Alberta Mathematical
            and Statistical Sciences Department|Pepperdine University Natural
            Sciences Division|Bard College Mathematics
            Program|University of Washington
            Mathematics Department}}
        \end{picture}
        \par\vskip-8mm}}
\definecolor{UMred}{rgb}{.9,.05,.2}
 \def\UMbanner{\vbox{\hsize0pt
        \setlength{\unitlength}{.4mm}
        \thicklines
        \begin{picture}(0,0)(-30,-10)
        \put(165,16){\line(1,0){4}}
        \put(170,16){\line(1,0){4}}
        \put(180,16){\line(1,0){4}}
        \put(175,0){\line(1,0){4}}
        \put(180,0){\line(1,0){4}}
        \put(185,0){\line(1,0){4}}
        \put(169,0){\line(0,1){16}}
        \put(170,0){\line(0,1){16}}
        \put(179,0){\line(0,1){16}}
        \put(180,0){\line(0,1){16}}
        \put(184,0){\line(0,1){16}}
        \put(185,0){\line(0,1){16}}
        \put(169,16){\oval(8,32)[bl]}
        \put(170,16){\oval(8,32)[br]}
        \put(179,0){\oval(8,32)[tl]}
        \put(185,0){\oval(8,32)[tr]}
        \end{picture}
        \par\vskip-6.5mm
        \thicklines}}
\catcode`@=11
\def\LT@makecaption#1#2#3{%
  \LT@mcol\LT@cols c{\hbox to\z@{\hss\parbox[t]\LTcapwidth{%
    \sbox\@tempboxa{#1{{\bf#2}: }#3}%
    \ifdim\wd\@tempboxa>\hsize
      #1{{\bf#2}: }#3%
    \else
      \hbox to \hsize{\hfil\box\@tempboxa\hfil}%
    \fi
    \endgraf\vskip\baselineskip}%
  \hss}}}
\catcode`@=12

 \SfTitles 
 \seceq
 
 \allowdisplaybreaks
\LTcapwidth=\hsize
\begin{document}
\thispagestyle{empty}
\vbox{\Border\UMbanner}
 \noindent
 \today\hfill UMDEPP 08-020, SUNY-O/670
  \vfill
 \begin{center}
{\LARGE\sf\bfseries Adinkras for Clifford Algebras,\\[2mm]
and Worldline Supermultiplets}\\[5mm]
  \vfill
{\sf\bfseries C.F.\,Doran$^a$, M.G.\,Faux$^b$, S.J.\,Gates, Jr.$^c$, T.\,H\"{u}bsch$^d$,\\
     K.M.\,Iga$^e$, G.D.\,Landweber$^f$ {\rm and} R.L.~Miller$^g$}\\[2mm]
{\small\it
  $^a$Department of Mathematical and Statistical Sciences,\\[-1mm]
      University of Alberta, Edmonton, Alberta, T6G 2G1 Canada%
  \\[-4pt] {\tt  doran@math.ualberta.ca}
  \\
  $^b$Department of Physics,\\[-1mm]
      State University of New York, Oneonta, NY 13825%
  \\[-4pt] {\tt  fauxmg@oneonta.edu}
  \\
  $^c$Center for String and Particle Theory,\\[-1mm]
      Department of Physics, University of Maryland, College Park, MD 20472%
  \\[-4pt] {\tt  gatess@wam.umd.edu}
  \\
  $^d$Department of Physics \&\ Astronomy,\\[-1mm]
      Howard University, Washington, DC 20059
  \\[-4pt] {\tt  thubsch@howard.edu}
  \\
  $^e$Natural Science Division,\\[-1mm]
      Pepperdine University, Malibu, CA 90263%
  \\[-4pt] {\tt  Kevin.Iga@pepperdine.edu}
  \\
 $^f$Department of Mathematics, Bard College,\\[-1mm]
     Annandale-on-Hudson, NY 12504-5000%
  \\[-4pt] {\tt  gregland@bard.edu}
  \\
 $^g$Department of Mathematics,\\[-1mm]
      University of Washington, Seattle, WA 98105%
  \\[-4pt] {\tt  rlmill@math.washington.edu}
 }\\[5mm]
  \vfill
{\sf\bfseries ABSTRACT}\\[3mm]
\parbox{5.55in}{
Adinkras are a graphical depiction of representations of the $N$-extended supersymmetry algebra in one dimension, on the worldline.  These diagrams represent the component fields in a supermultiplet as vertices, and the action of the supersymmetry generators as edges.
In a previous work, we showed that the chromotopology (topology with colors) of an Adinkra must come from a doubly even binary linear code.  Herein, we relate Adinkras to Clifford algebras, and use this to construct, for every such code, a supermultiplet corresponding to that code.  In this way, we correlate the well-known classification of representations of Clifford algebras to the classification of Adinkra chromotopologies.}
\end{center}
  \vfill
\noindent PACS: 04.65.+e

\clearpage\setcounter{page}{1}\pagenumbering{roman}
\tableofcontents

\clearpage\setcounter{page}{1}\pagenumbering{arabic}
\section{Introduction}
 \label{Intro}
All physical systems involve a notion of time, and therefore admit an essentially unique time-like dimensional reduction. Analyses based on this defer including the technical details related to higher-dimensional Lorentz algebra. Such is also the case of supersymmetric theories, which end up having an $N$-extended worldline supersymmetry, where we restrict $0\leq N\leq32$ having in mind superstrings and their $M$- and $F$-theory extensions. In particular, we are exclusively interested in classifying the off-shell representations of supersymmetry, which are logically indispensable in supersymmetric quantum theories.

To this end, Ref.\cite{rA} introduced Adinkras: graphs in which vertices represent component fields in a supermultiplet (white for bosons and black for fermions), and variously decorated edges to represent the action of supersymmetry amongst the component fields. In Ref.\cite{r6-3}, we showed that every connected Adinkra has the chromotopology (1-skeleton topology with black/white equipartitioned vertices and $N$-colored edges) of a quotient of a $N$-cube, $[0,1]^N$, by a doubly even binary linear block code. The term ``doubly even'' refers to the fact that all codewords in such codes are $N$-bit binary numbers wherein the number of 1's is divisible by 4.

In this paper, we show that any doubly even code can be used to construct an Adinkra and hence, a supermultiplet of $N$-extended worldline supersymmetry.  This is done explicitly, with several examples for small values of $N$, utilizing a relationship between $N$-extended supersymmetry and Clifford algebras. However, several inequivalent Adinkras may correspond to equivalent supermultiplets. We show that this ambiguity is resolved by noting the details of the correspondence between the given Adinkras on one hand, and the supermultiplets on the other, to suitable representations of Clifford algebras in the middle. The details of this depend on the engineering dimensions of the component fields.

The paper is organized as follows: Section~\ref{s:susy} is a review of the supersymmetry algebra, introduces Clifford algebras and explains the relationship between the two.  Section~\ref{s:construct} explains how to take a code and produce a Clifford representation, and hence, a supermultiplet, with the corresponding chromotopology.  Section~\ref{s:ambiguity} relates these results to the classification of $D=1$ adinkraic supermultiplets, raises the issue that inequivalent Adinkras might represent equivalent supermultiplets, and discusses the conditions for this to happen. Section~\ref{s:C} summarizes our results, and Tables~\ref{redtable} and~\ref{ncstable} list all inequivalent Adinkras (and corresponding supermultiplets), with their node choice symmetry and up to the choice of edge dashing for $N\leq4$; Table~\ref{N_5_table} then indicates the onset of the combinatorial explosion at $N=5$.
 Appendix~\ref{s:linalg} collects some useful facts about the linear algebra of binary codes, and Appendix~\ref{s:red} contrasts reducibility {\em vs}.\ decomposability of supermultiplets.

\section{Supersymmetry and Clifford Algebras}
 \label{s:susy}
The $N$-extended supersymmetry algebra without central charges in one dimension is formed by the time-derivative, $\ddt$, and the $N$ Hermitian operators $Q_1, \cdots, Q_N$ satisfying\Ft{The convention of the positive sign on the right-hand side of Eq.\eq{eSuSy} ensures that the spectrum of the worldline Hamiltonian, $H=i\hbar\ddt$, is bounded from below as is the free energy in supersymmetric systems.}
\begin{equation}
 \big\{\,Q_I\,,\,Q_J\,\big\}=2i\,\d_{IJ}\,\ddt,\quad
  \big[\,\ddt\,,\,Q_I\,\big] =0,\quad I,J=1,\cdots,N. \label{eSuSy}
\end{equation}
Noting that the engineering dimension $[\ddt]=1$ and \Eq{eSuSy} imply that $[Q_I]=\inv2$,
in this section we determine some essential facts about the transformation rules of these operators on supermultiplets consisting of component fields each of which has a consistent engineering dimension, and to the important class of supermultiplets in which the $Q$-actions corresponds to Adinkras.

A supermultiplet $\sM$ is a real, unitary, finite-dimensional representation of the algebra\eq{eSuSy}, in the following sense: $\sM$ consists of a finite set of real bosonic fields, $(\f_1\6(\t),\dots,\f_m\6(\t))$, and real fermionic fields, $(\j_1\6(\t),\dots,\j_m\6(\t))$, jointly called {\em component fields\/} and each of which is a function of time, $\t$. The collection $(\f_1\6(\t),\dots,\f_m\6(\t)|\j_1\6(\t),\dots,\j_m\6(\t))$ is closed\ft{The term ``closed'' herein means that the application of any $Q$-monomial on any component field is a linear combination of component fields and their $\t$-derivatives.} under the linear action of the $N$ Hermitian supersymmetry generators $Q_1,\dots,Q_N$, which satisfy\eq{eSuSy} and swap bosons and fermions.

As noted, we will be interested exclusively in {\em\/off-shell supermultiplets\/}, wherein no component field satisfies any (space)time differential equation. Supersymmetry then guarantees that the number of bosonic and the number of fermionic component fields in $\sM$ is the same.

\subsection{Adinkras}
Refs.\cite{rA,r6-1,r6-2} introduced and then studied Adinkras: diagrams that encode the transformation rules of the component fields $(\f_1\6(\t),\dots,\f_m\6(\t)|\j_1\6(\t),\dots,\j_m\6(\t))$ under the action of the supersymmetry generators $Q_1,\dots,Q_N$.
 Supermultiplets that can be described by Adinkras have a collection of bosonic and fermionic component fields and a collection of supersymmetry generators $Q_1,\dots,Q_N$, so that: ({\bf1})~Given a bosonic field $\f$ and a supersymmetry generator $Q_I$, the transformation rule for $Q_I$ of $\f$ is of the form
\begin{alignat}{3}
 \text{either}&&\qquad Q_I \f &= \pm\, \j, \label{eQB1}\\
 \text{or}&&\qquad Q_I \f &= \pm\, \ddt \j, \label{eQB2}
\intertext{for some fermionic field $\j$. ({\bf2})~Given instead a fermionic field $\eta$ and a supersymmetry generator $Q_I$, the transformation rule of $Q_I$ on $\eta$ is of the form}
 \text{either}&&\qquad Q_I \eta &= \pm\, i\,B, \label{eQF1}\\
 \text{or}&&\qquad Q_I \eta &= \pm\, i\,\ddt B, \label{eQF2}
\end{alignat}
for some bosonic field $B$.  In particular, these supersymmetry generators act linearly using first-order differential operators.  Furthermore, the supersymmetry algebra requires that
\begin{alignat}{3}
 Q_I \f &= \pm\,  \j \qquad&\Longleftrightarrow\qquad
 Q_I \j &= \pm\, i\,\ddt \f, \label{eQBF1}
\intertext{and}
 Q_I \f &= \pm\, \ddt\j \qquad&\Longleftrightarrow\qquad
 Q_I \j &= \pm\, i\,\f,\label{eQBF2}
\end{alignat}
and where the $\pm$ signs are correlated to preserve Eqs.\eq{eSuSy}.

More generally, suppose we label the bosons $\f_1,\dots,\f_m$ and the fermions $\j_1,\dots,\j_m$.  Choose an integer $I$ with $1\le I\le N$, and an integer $A$ with $1\le A\le m$.  For each such pair of integers, we consider the transformation rules for $Q_I$ on the boson $\f_A$, and we might expect that these will be of the form
\begin{equation}
Q_I \, \f_A(\t)= c\,\ddt^{\l}\, \j_B(\t),\label{eQB}
\end{equation}
where $c=\pm 1$, $\l=0$ or $1$, and $B$ is an integer with $1\le B\le m$, so that $\j_B$ is some fermion; each of $c,\l,B$ will, in general, depend on $I$ and $A$.  Note that
\begin{equation}
\l=[\f_A]-[\j_B]+\inv2,
 \label{eED}
\end{equation}
for $\f_A$ and $\j_B$ to have a definite engineering dimension---provided the transformation rules had only dimensionless constants as we assume throughout. For each such transformation rule, we will get a corresponding transformation rule for the $Q_I$ on the fermion $\j_B(\t)$ that looks like this:
\begin{equation}
Q_I \, \j_B(\t) = \frac{i}{c}\,\ddt^{1-\l}\, \f_A(\t).\label{eQF}
\end{equation}
\Eqs{eQB}{eQF} constitute all of the transformation rules on the bosons and fermions, respectively.

\begin{definition}\label{dAd}
A supermultiplet, $\sM$, is {\em\/adinkraic\/} if all of its supersymmetric transformation rules are of the form\eq{eQB} and\eq{eQF}.
 For each adinkraic supermultiplet, its {\em Adinkra\/}, $\cA_\sM$, is a directed graph, consisting of a set of vertices, $V$, a set of edges, $E$, a coloring $C$ of the edges, a set of their orientations, $O$, and a labeling $D$ of each edge corresponding to whether or not it is dashed. 

Each component field of $\sM$ is represented by a vertex in $\cA_\sM$: white for bosonic fields and black for fermionic ones, thus equipartitioning the vertex set $V\to W$.  Every transformation rule of the form\eq{eQB} is depicted by an edge connecting the vertex corresponding to $\f_A$ to the vertex corresponding to $\j_B$, and color the edge with the $I^\text{th}$ color.  We use a dashed edge if $c=-1$, and oriented it from $\f_A$ to $\j_B$ if $\l=0$ and the other way around if $\l=1$.
\end{definition}

\subsection{Relationship to Clifford Algebras and Valise Supermultiplets}
 \label{s:ClifAd}
It has been understood for a long time that the supersymmetry algebra\eq{eSuSy} has a formal similarity with the Clifford algebra generated by the Dirac $\G_I$ matrices\Ft{In \eq{eSuSy}, $I,J=1,\cdots,N$ count {\em fermionic\/} dimensions. Accordingly, the quadratic form, $\d_{IJ}$, occurring in the right-hand side of \Eq{eDirac} is positive definite. This is unlike the {\em typical\/} field theory use of the Clifford/Dirac algebra, where $I,J$ would count {\em bosonic\/} dimensions of spacetime, implying the Lorentzian signature for the quadratic form on the right-hand side of \Eq{eDirac}.  Another point is that traditionally there would be a minus sign on the right hand side; but this difference is equivalent to changing the metric from positive definite to negative definite.  In the language of Ref.~\cite{rLM}, this leads to the Clifford algebra $\Cl(0,N)$ instead of $\Cl(N)=\Cl(N,0)$.}:
\begin{equation}
 \big\{\, \G_I \,,\, \G_J \,\big\} = 2\,\d_{IJ}\,\Ione. \label{eDirac}
\end{equation}
One manifestation of this was the study of the spinning particle by Gates and Rana\cite{rGR1,rGR2}, resulting in the Scalar Supermultiplet (which we herein rename {\em Isoscalar Supermultiplet\/}), defined in terms of the ${\cal GR}(d,N)$ algebra, a form of the Clifford algebra.  We will review their work first, and more explicitly describe how it relates to Clifford algebras; this will motivate the more general construction of relating supermultiplets to a representation of the Clifford algebra.

The main idea relating supermultiplets and Clifford representations is to note that the supersymmetry algebra in one dimension and the Clifford algebra differ only in that in the former, there is a factor of $i$ and a derivative $\ddt$.  So, to turn a supermultiplet into a representation of the Clifford algebra, we can simply forget the factors of $i$ and the derivatives $\ddt$.\footnote{This construction is in a sense the same construction quotienting by the ideal $(H-1)$ in Ref.\cite{r6--1}, where $H=i\ddt$. This also removes the $\ZZ$-grading afforded by the notion of engineering dimension.}  If we wish to reverse this process, we note that the factors of $i$ are consistently in the right-hand side of the transformation rules for the fermions,\footnote{This ensures the reality of all component fields; see Ref.~\cite{rA,r6-3}.} but appearance of $\ddt$'s is determined only by \Eq{eED}.  Thus, ``forgetting'' the $\ddt$ loses information that can be restored only if we know the engineering dimensions of the component fields, and which turns out to permit a combinatorial plethora of choices.

To address this ambiguity, in this section we fix the engineering dimensions of all bosonic component fields in $\sM$ to be the same, and do the same for the fermions. Then, either $[\j]=[\f]+\inv2$ and the transformation rules for all bosons are of the type\eq{eQB1}, or
$[\j]=[\f]-\inv2$ and the transformation rules for all bosons are of the type\eq{eQB2}. Then, \Eqs{eQBF1}{eQBF2} enforce the transformation rules for all fermions to be, correspondingly, either of the type\eq{eQF1} or\eq{eQF2}. These two choices result, respectively, in Isoscalar or Isospinor supermultiplets.

 In corresponding Adinkras, all edges are oriented either from the bosons to the fermions, or the other way around. If we also draw the Adinkras so that all the edges are oriented upwards, then all the bosons are at the same level and all fermions at another, adjacent level; either all bosons are above the fermions, or the other way around. Collectively, we call these {\em valise supermultiplets}, and if they are adinkraic, they correspond to {\em valise Adinkras\/}\Ft{In Ref.\cite{rA}, these were called {\em base Adinkras}.}.
 For example, both $N=3$ Adinkras shown in Figure~\ref{valise} are valise Adinkras.
\begin{figure}[htb]
\begin{center}
 \begin{picture}(130,35)(0,0)
 \put(-3,3){\includegraphics{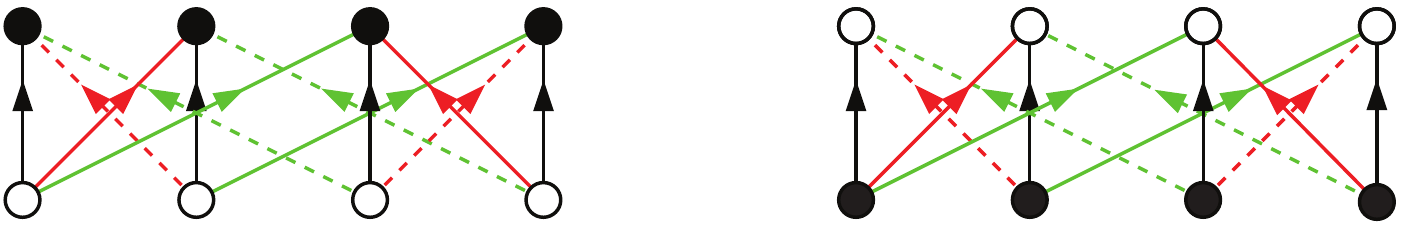}}
 \put(0,27){$\makebox[0in][c]{$\j_{(100)}$}$}
 \put(18,27){$\makebox[0in][c]{$\j_{(010)}$}$}
 \put(36,27){$\makebox[0in][c]{$\j_{(001)}$}$}
 \put(54,27){$\makebox[0in][c]{$\j_{(111)}$}$}
 \put(0,0){$\makebox[0in][c]{$\f_{(000)}$}$}
 \put(18,0){$\makebox[0in][c]{$\f_{(110)}$}$}
 \put(36,0){$\makebox[0in][c]{$\f_{(101)}$}$}
 \put(54,0){$\makebox[0in][c]{$\f_{(011)}$}$}
 \put(82,27){$\makebox[0in][c]{$B_{(100)}$}$}
 \put(100,27){$\makebox[0in][c]{$B_{(010)}$}$}
 \put(118,27){$\makebox[0in][c]{$B_{(001)}$}$}
 \put(136,27){$\makebox[0in][c]{$B_{(111)}$}$}
 \put(82,0){$\makebox[0in][c]{$\h_{(000)}$}$}
 \put(100,0){$\makebox[0in][c]{$\h_{(110)}$}$}
 \put(118,0){$\makebox[0in][c]{$\h_{(101)}$}$}
 \put(136,0){$\makebox[0in][c]{$\h_{(011)}$}$}
 \end{picture}
\end{center}
\caption{The two $N=3$ Valise Adinkras: Isoscalar and Isospinor}
\label{valise}
\end{figure}
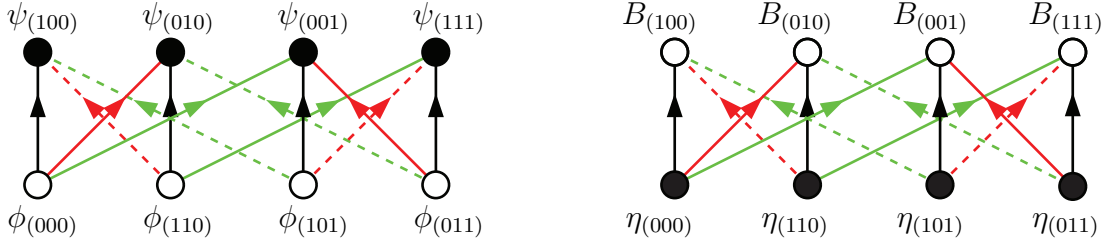
The two valise Adinkras in Figure~\ref{valise} are related by the obvious transformation, called the Klein flip, which swaps bosons and fermions. With this in mind, we restrict most of the subsequent discussion to Isoscalar supermultiplets; the Klein flip then generates the analogous results for the Isospinor supermultiplet. The Adinkra of the $N=3$ such supermultiplets is shown in Figure~\ref{valise}.

We refer to these Adinkras as valises because in these cases the engineering dimensions of the vertices are set as closely together as possible, thereby encoding the topological
class of a larger set of Adinkras in the most economically compact manner. This reminds of a suitcase which can ``unpacked'', by raising vertices---see Theorems~4.1,~5.1 and their corollaries in Ref.\cite{r6-1}---to become other Adinkras spanning more engineering dimensions.

Subsequently, in Section~\ref{s:ambiguity}, we include other kinds of supermultiplets; Figure~\ref{nonvalise} displays two such Adinkras.  Ref.\cite{r6-1} shows that given any Adinkra, we can successively perform an operation called a vertex raise to generate a valise Adinkra that has the same topology as the original Adinkra.  Thus, to find all possible Adinkra topologies, it suffices to study the topologies of valise Adinkras.
\begin{figure}[htb]
\begin{center}
 \begin{picture}(130,50)(0,0)
 \put(-3,1.5){\includegraphics{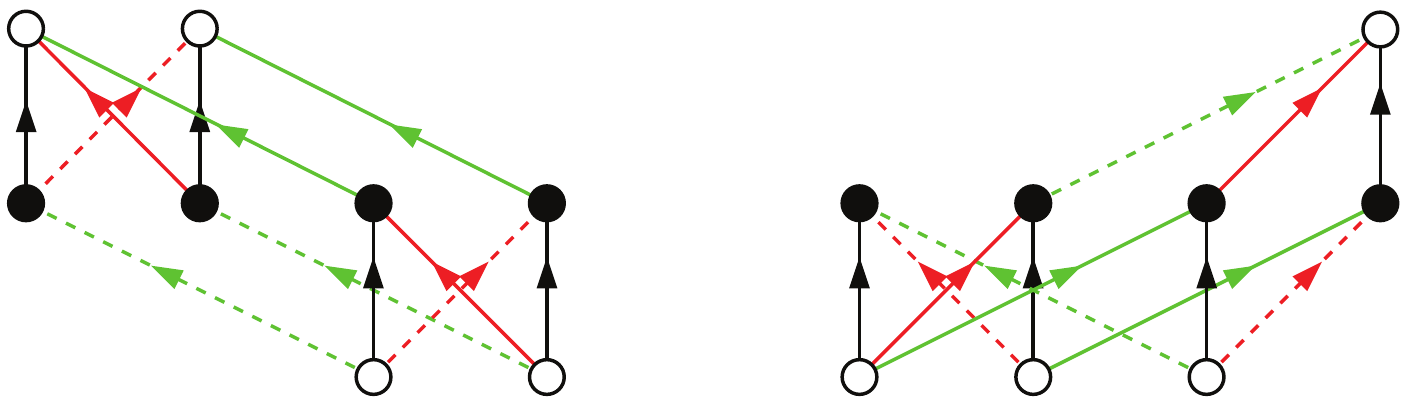}}
 \put(0,18){$\makebox[0in][c]{$\j_{(100)}$}$}
 \put(18,18){$\makebox[0in][c]{$\j_{(010)}$}$}
 \put(36,27){$\makebox[0in][c]{$\j_{(001)}$}$}
 \put(54,27){$\makebox[0in][c]{$\j_{(111)}$}$}
 \put(0,44){$\makebox[0in][c]{$F_{(000)}$}$}
 \put(18,44){$\makebox[0in][c]{$F_{(110)}$}$}
 \put(36,0){$\makebox[0in][c]{$\f_{(101)}$}$}
 \put(54,0){$\makebox[0in][c]{$\f_{(011)}$}$}
 \put(82,27){$\makebox[0in][c]{$\j_{(100)}$}$}
 \put(100,27){$\makebox[0in][c]{$\j_{(010)}$}$}
 \put(118,27){$\makebox[0in][c]{$\j_{(001)}$}$}
 \put(136,18){$\makebox[0in][c]{$\j_{(111)}$}$}
 \put(82,0){$\makebox[0in][c]{$\f_{(000)}$}$}
 \put(100,0){$\makebox[0in][c]{$\f_{(110)}$}$}
 \put(118,0){$\makebox[0in][c]{$\f_{(101)}$}$}
 \put(136,44){$\makebox[0in][c]{$F_{(011)}$}$}
 \end{picture}
\end{center}
\caption{Two $N=3$ non-valise Adinkras: a $(2|4|2)$- and a $(3|4|1)$-dimensional one}
\label{nonvalise}
\end{figure}
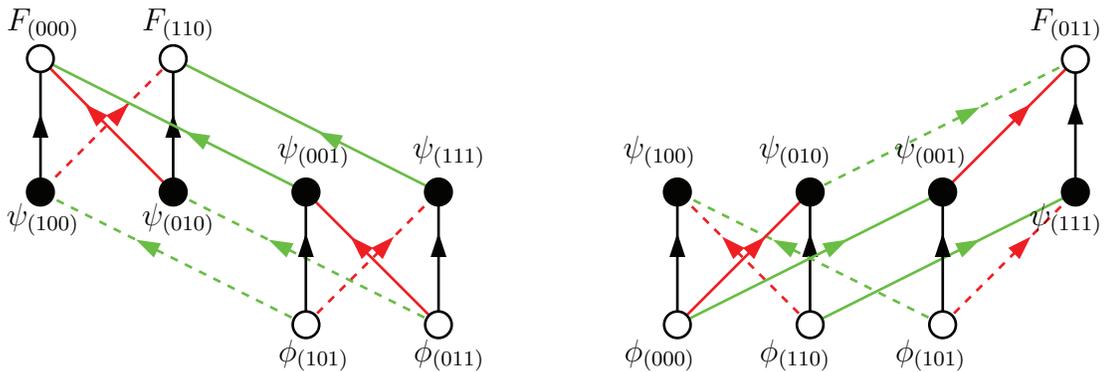

As Adinkras will play a prominent r\^ole in the presentation of our results, we pause to translate some Adinkras into the standard supersymmetry transformation relations. The left-hand side valise Adinkra in Figure~\ref{valise} encodes:
\begin{subequations}
 \label{eN3ISc}
\begin{alignat}{5}
 \makebox[0pt][c]{\hphantom{=}\bf Black edges}&&
 \makebox[0pt][c]{\hphantom{=}\bf\color{Red} Red edges}&&
 \makebox[0pt][c]{\hphantom{=}\bf\color{Green} Green edges}&\nn\\
 Q_1\,\f_{(000)}&=\j_{(100)}, \quad&\quad
 Q_2\,\f_{(000)}&=\j_{(010)}, \quad&\quad
 Q_3\,\f_{(000)}&=\j_{(001)}, \\
 Q_1\,\f_{(110)}&=\j_{(010)}, \quad&\quad
 Q_2\,\f_{(110)}&=-\j_{(100)}, \quad&\quad
 Q_3\,\f_{(110)}&=\j_{(111)}, \\
 Q_1\,\f_{(101)}&=\j_{(001)}, \quad&\quad
 Q_2\,\f_{(101)}&=-\j_{(111)}, \quad&\quad
 Q_3\,\f_{(101)}&=-\j_{(100)}, \\
 Q_1\,\f_{(011)}&=\j_{(111)}, \quad&\quad
 Q_2\,\f_{(011)}&=\j_{(001)}, \quad&\quad
 Q_3\,\f_{(011)}&=-\j_{(010)}, \\
 Q_1\,\j_{(100)}&=i\dot\f_{(000)}, \quad&\quad
 Q_2\,\j_{(100)}&=-i\dot\f_{(110)}, \quad&\quad
 Q_3\,\j_{(100)}&=-i\dot\f_{(101)}, \\
 Q_1\,\j_{(010)}&=i\dot\f_{(110)}, \quad&\quad
 Q_2\,\j_{(010)}&=i\dot\f_{(000)}, \quad&\quad
 Q_3\,\j_{(010)}&=-i\dot\f_{(011)}, \\
 Q_1\,\j_{(001)}&=i\dot\f_{(101)}, \quad&\quad
 Q_2\,\j_{(001)}&=i\dot\f_{(011)}, \quad&\quad
 Q_3\,\j_{(001)}&=i\dot\f_{(000)}, \\
 Q_1\,\j_{(111)}&=i\dot\f_{(011)}, \quad&\quad
 Q_2\,\j_{(111)}&=-i\dot\f_{(101)}, \quad&\quad
 Q_3\,\j_{(111)}&=i\dot\f_{(110)}. 
\end{alignat}
\end{subequations}

By contrast, the left-hand side non-valise Adinkra in Figure~\ref{nonvalise} encodes:
\begin{subequations}
 \label{e3cube}
\begin{alignat}{5}
 \makebox[0pt][c]{\hphantom{=}\bf Black edges}&&
 \makebox[0pt][c]{\hphantom{=}\bf\color{Red} Red edges}&&
 \makebox[0pt][c]{\hphantom{=}\bf\color{Green} Green edges}&\nn\\
 Q_1\, F_{(000)}&=\dot\j_{(100)}, \quad&\quad
 Q_2\, F_{(000)}&=\dot\j_{(010)}, \quad&\quad
 Q_3\, F_{(000)}&=\dot\j_{(001)}, \\
 Q_1\, F_{(110)}&=\dot\j_{(010)}, \quad&\quad
 Q_2\, F_{(110)}&=-\dot\j_{(100)}, \quad&\quad
 Q_3\, F_{(110)}&=\dot\j_{(111)}, \\
 Q_1\,\f_{(101)}&=\j_{(001)}, \quad&\quad
 Q_2\,\f_{(101)}&=-\j_{(111)}, \quad&\quad
 Q_3\,\f_{(101)}&=-\j_{(100)}, \\
 Q_1\,\f_{(011)}&=\j_{(111)}, \quad&\quad
 Q_2\,\f_{(011)}&=\j_{(001)}, \quad&\quad
 Q_3\,\f_{(011)}&=-\j_{(010)}, \\
 Q_1\,\j_{(100)}&=i F_{(000)}, \quad&\quad
 Q_2\,\j_{(100)}&=-i F_{(110)}, \quad&\quad
 Q_3\,\j_{(100)}&=-i\dot\f_{(101)}, \\
 Q_1\,\j_{(010)}&=i F_{(110)}, \quad&\quad
 Q_2\,\j_{(010)}&=i F_{(000)}, \quad&\quad
 Q_3\,\j_{(010)}&=-i\dot\f_{(011)}, \\
 Q_1\,\j_{(001)}&=i\dot\f_{(101)}, \quad&\quad
 Q_2\,\j_{(001)}&=i\dot\f_{(011)}, \quad&\quad
 Q_3\,\j_{(001)}&=i F_{(000)}, \\
 Q_1\,\j_{(111)}&=i\dot\f_{(011)}, \quad&\quad
 Q_2\,\j_{(111)}&=-i\dot\f_{(101)}, \quad&\quad
 Q_3\,\j_{(111)}&=i F_{(110)}.
\end{alignat}
\end{subequations}

The Reader should have no difficulty following suit with all other Adinkras presented herein.

\subsection{The Isoscalar Supermultiplet and the Clifford Algebra}
Generalizing Eqs.\eq{eN3ISc}, suppose we have an Isoscalar supermultiplet, $(\f_1,\dots,\f_m|\j_1,\dots,\j_m)$; write these as column vectors:
\begin{equation}
 \Phi=\begin{bmatrix}
               \f_1\\ \vdots\\ \f_m
            \end{bmatrix}
 \qquad\text{and}\qquad
 \Psi=\begin{bmatrix}
               \j_1\\ \vdots\\ \j_m
             \end{bmatrix}.
\end{equation}
The property of being valise corresponds to the following Ansatz:
\begin{align}
Q_I\Phi &= \IL_I\Psi,\label{scalarmult1}\\
Q_I\Psi &= i\,\IR_I\,\ddt\Phi,\label{scalarmult2}
\end{align}
where $\IL_I$ and $\IR_I$ are $m\times m$ real matrices to be determined. The supersymmetry algebra\eq{eSuSy} then implies
\begin{eqnarray}
\IL_I\IR_J+\IL_J\IR_I&=&2\,\d_{IJ}\Ione,\label{lr1}\\
\IR_I\IL_J+\IR_J\IL_I&=&2\,\d_{IJ}\Ione,\label{lr2}
\end{eqnarray}
where $\Ione$  is the $m\times m$ identity matrix. The $I=J$ cases of these equations imply that $\IR_I = \IL_I^{-1}$.
\Remk
The additional requirement\Ft{The variable sign here depends on the selection of the relative sign between the bosonic and the fermionic kinetic terms, $\inv2\dot\f^2$ and $i\j\dot\c$, respectively.} $\IL_I=\pm \IR_I^T$ is also often made\cite{rGR1}, but this is needed only when writing Lagrangians and will play no r\^ole in this paper.

We create, for each $I$, a $2m\times 2m$ real matrix $\G_I$ of the form
\begin{equation}
  \G_I\Defl \left[\begin{array}{c|c}
                      {\bf~0}&\IL_I\\\hline
                       \IR_I &{\bf0}
                  \end{array}\right].
 \label{gammalr}
\end{equation}
Equations\eq{lr1} and\eq{lr2} then imply that these $\G_1,\cdots,\G_N$ indeed satisfy \Eq{eDirac}, the algebra of the Dirac gamma matrices in dimension $N$, whence the notation.

We do not insist that $m$ be minimal here, nor do we adopt any particular convention for the type of the $\G_I$ matrices.  Rather, we consider all possible choices for $\G_I$ of all (finite) sizes, and note that for each, we have a valise supermultiplet. Refs.\cite{rGR1,rGR2} were mostly concerned with valise supermultiplets of a minimal $m$, wherein they were called `scalar multiplets' and `spinor multiplets', depending on whether the bosons or the fermions had the lower engineering dimension. In those references, the particular choices of matrices $\G_I$ was not described, beyond the fact that they were block off-diagonal.

The property of being block off-diagonal corresponds to the existence of a fermion number operator $(-1)^F$ which anti-commutes with the $\G_I$.  If we write our fields listing the bosons first, followed by the fermions, then $(-1)^F$ will be a diagonal matrix of the form
\begin{equation}
(-1)^F=\left[\begin{array}{c|c}
\Ione& {\bf0} \\\hline
{\bf0}& -\Ione\end{array}\right].\label{fermionnumber}
\end{equation}
We may well write $\G_0\Defl(-1)^F$.  Note that $\G_0$ anticommutes with the remaining $\G_I$:
\begin{equation}
 \G_0{}^2=+\Ione, \qquad \{\G_0,\G_I\}=0.
\end{equation}
The formal algebra generated by $\G_0, \G_1, \dots, \G_N$, is then defined by the anticommutation relations:
\begin{equation}
\{\G_I,\G_J\}=2\,\delta_{IJ}\Ione,\qquad \text{for }I,J=0,1,\cdots,N.\label{eDirac2}
\end{equation}
The algebra generated by the $\G_1,\dots,\G_N$ is the {\em Clifford algebra} $\Cl(0,N)$, and the inclusion of $\G_0$ extends $\Cl(0,N)$ into the Clifford algebra $\Cl(0,N{+}1)$. Any set of real matrices representing these $\G_0,\G_1,\dots,\G_N$ such that the above algebra closes is called a {\em Clifford representation}\cite{rLM}.  We will insist that $\G_0$ be diagonal and of the form given above\eq{fermionnumber}.

Thus, this construction describes a one-to-one correspondence between real valise supermultiplets and real Clifford representations.

\begin{definition}
Associated to the supermultiplet $\sM=(\f_1,\dots,\f_m|\j_1,\dots,\j_m)$,
the vector space $\IR^{2m}$ together with the $2m\times2m$ block off-diagonal matrices $\G_1,\cdots,\G_N$, defined in \Eq{gammalr}, generate the Clifford representation, $\cM$.\label{dA2C}
\end{definition}

\section{Constructing Valise Adinkras}
\label{s:construct}
The main construction in this paper is to take a description of an $N$-dimensional cube $I^N$ and a doubly even $[N,k]$-code $C$, and construct a Clifford representation.  The Isoscalar supermultiplet corresponding to this Clifford representation will then be a supermultiplet for $N$-extended worldline supersymmetry, whose Adinkra is valise, with a chromotopology given by $I^N/C$.

\subsection{Clifford Supermultiplets and Cubical Adinkras}
In Ref.\cite{r6-3}, we showed that the topology of every connected Adinkra is a quotient of an $N$-dimensional cube.  Thus, it is important to study how $N$-dimensional cubes arise as Adinkras for Clifford representations.  The representation in question is the Clifford algebra itself, $\Cl(0,N)$.

The Clifford algebra $\Cl(0,N)$ is a real $2^N$-dimensional vector space, spanned by products of the form $\G_{I_1}\cdots \G_{I_k}$, where $I_1<\dots<I_k$\cite{rLM}.  As a vector space, it splits as a direct sum of two vector spaces: the even and odd parts.  The even (resp. odd) part is the subspace spanned by products of even (resp. odd) numbers of $\G_I$ matrices.

The Clifford algebra $\Cl(0,N{+}1)$ acts on $\Cl(0,N)$ in the following way: for every $1\le I\le N$, $\G_I$ acts by multiplication on the left.  The operator $\G_0$ multiplies the even $\G$-monomials by $1$ and the odd $\G$-monomials by $-1$.  It is then easy to see that the Clifford algebra\eq{eDirac2} holds.

We take as a basis for $\Cl(0,N)$ the products $\G_{I_1}\dots \G_{I_k}$ as above.  Specifically, for every vertex of the cube $\vec{x}=(x_1,\dots,x_N)\in\{0,1\}^N$, we define a component field\Ft{Elements of a Clifford algebra can be used to realize the component fields of a supermultiplet, by allowing the element of the Clifford algebra to be a function of the time-like coordinate $\tau$.  This construction was explicitly carried out in equation (67) of \cite{rGLP}.} 
\begin{equation}
e_{\vec{x}}=\G_1{}^{x_1}\cdots \G_N{}^{x_N}.\label{cubedefine}
\end{equation}
We note that $\G_I e_{\vec{x}}=\G_I\G_1{}^{x_1}\dots\G_N{}^{x_N}$ can be transformed into the form\eq{cubedefine}, with perhaps an overall minus sign: we use the anticommutation of the $\G_I$'s and, if $x_I=1$, the fact that $\G_I{}^2=1$. This results in
\begin{equation}
\G_I\cdot e_{\vec{x}} = \pm \G_1{}^{x_1}\cdots \G_{I-1}{}^{x_{I-1}} \,\G_I{}^{1-x_I}\,
\G_{I+1}{}^{x_{I+1}}\cdots\G_N{}^{x_N}.\label{cubeedge}
\end{equation}
The sign is $+1$ if the number of $J$ with $x_J=1$ and $J < I$ is even, and is $-1$ otherwise.  If we define $(-1)^{|\vec{x} < I|}$ to be that sign, and define $\vec{x}\,\boxplus I$ to be the vector $(x_1,\ldots,x_{I-1},1-x_I,x_{I+1},\ldots,x_N)$, then this equation becomes
\begin{equation}
\G_I\cdot e_{\vec{x}} = (-1)^{|\vec{x} < I|} e_{\vec{x}\,\boxplus\,I}.
\end{equation}
We will also need the function $\wt(\vec{x})$, which equals the number of 1's in $\vec{x}$.

Thus, $\Cl(0,N)$ is a representation of the Clifford algebra $\Cl(0,N{+}1)$.  It corresponds to an Isoscalar supermultiplet, by replacing $e_{\vec{x}}\mapsto\phi_{\vec{x}}(\t)$ when the weight of $\vec{x}$ is even, and $e_{\vec{x}}\mapsto\psi_{\vec{x}}(\t)$ when the weight of $\vec{x}$ is odd.  Following the construction of the Isoscalar supermultiplet, we define $\IL_I$ and $\IR_I$ to be
\begin{alignat}{3}
  \IL_I \psi_{\vec{x}} &= (-1)^{|\vec{x} < I|}\phi_{\vec{x}\,\boxplus\,I},\qquad&\qquad
  \IR_I \phi_{\vec{x}} &= (-1)^{|\vec{x} < I|}\psi_{\vec{x}\,\boxplus\,I},
 \label{eCuBaseLR}
\intertext{and thus define}
  Q_I \psi_{\vec{x}} &= (-1)^{|\vec{x} < I|}\phi_{\vec{x}\,\boxplus\,I},\qquad&\qquad
  Q_I \phi_{\vec{x}} &= (-1)^{|\vec{x} < I|}i \ddt\psi_{\vec{x}\,\boxplus\,I}.
 \label{eCuBaseQ}
\end{alignat}

The Adinkra for this is an $N$-dimensional cube.  To see this, we note that the basis elements are labeled by $\vec{x}\in\{0,1\}^N$, the vertices of the $N$-dimensional cube.  The edges colored $I$ connect $\vec{x}$ to $\vec{x}\,\boxplus I$, which changes the $I^\text{th}$ coordinate.  The supermultiplet specified by Eqs.\eq{eCuBaseQ} was called the ``bosonic Clifford Algebra superfield'' and the ``base superfield'' in Ref.~\cite{rA}, and this cubical description matches that in Ref.~\cite{r6-3}. To emphasize its chromotopology, we will refer to this as the ``$I^N$ Clifford supermultiplet''.

\subsection{The $N=4$, $D_4$ Projection}
\label{s:d4}
In considering Adinkras that are not cubes, but rather quotients of cubes, it is useful to first consider a few examples.  First, we consider the $N=4$ example obtained by quotienting a four-dimensional cube by identifying antipodal points.  This Adinkra was first described in Ref.~\cite{rA}, where it was identified as the dimensional reduction of the $D=4$ chiral superfield.  Here, we present this example in a way that will motivate the general construction to quotient cubes.

\subsubsection{One $D_4$ Projection}
Consider the Clifford representation $\Cl(0,4)$.  Define the element
\begin{equation}
 g=\G_1\G_2\G_3\G_4.\label{g4}
\end{equation}
Define the two linear transformations $\p_{+}, \p_{-}:\Cl(0,N)\to\Cl(0,N)$ to be
\begin{equation}
 \p_{\pm}(v)=v{\cdot}\frac{1\pm g}{2},\quad v,\p_{\pm}(v)\in\Cl(0,N) . 
 \label{eqn:proj}
\end{equation}

Note that in the definition\eq{eqn:proj}, we multiply $v$ by a factor on its right.  This is important.  It implies that for every $\G_I$ and $v\in \Cl(0,N)$, we have that $\G_I\big(\p_{\pm}(v)\big)=\p_{\pm}\big(\G_I(v)\big)$.  Such a linear map is called a homomorphism of Clifford representations, and it implies that $\img(\p_+)$ and $\img(\p_-)$ are also Clifford representations.

The fact that $g$ is even means that $\p_\pm$ preserves the bosonic and fermionic statistics.  The fact that $g^2=1$ implies that $\p_+^2=\p_+$, $\p_-^2=\p_-$, $\p_+\p_-=\p_-\p_+=0$, and $\p_+ + \p_- = \Ione$.  This means that $\p_+$ and $\p_-$ are a complete set of projection operators, so that $\Cl(0,N)=\img(\p_+)\oplus\img(\p_-)$ as Clifford representations, and $\p_+$ and $\p_-$ project onto each of these components.  Neither of these components is zero, since $(1+g)/2$ and $(1-g)/2$ are themselves non-zero elements of $\Cl(0,N)$, which are in $\img(\p_+)$ and $\img(\p_-)$, respectively.

This kind of projection is nothing new: the matrix $g$ is, up to a scalar factor, the matrix known as $\g_5$ in four-dimensional field theory, and the projection $\p_\pm$ corresponds to the familiar projection to chiral spinors: the left- and right-handed halves of the Dirac spinor.  (To be meticulous, the spacetime metric being pseudo-Riemannian, one of the Clifford/Dirac matrices in four-dimensional field theory squares to $-\Ione$, and the definition of $\g_5$ accordingly includes an explicit factor of $i$.)

It is also true that $\ker(\p_+)=\img(\p_-)$ and vice-versa, so that we can also describe these representations in terms of constraints: as ``$v$ such that $v{\cdot}(1\mp g)=0$''.  Then we can realize the Clifford representation as a subspace of $\Cl(0,N)$, rather than as a quotient.  This accords with the idea that $\Cl(0,N)$, being a representation for $\Cl(0,N{+}1)$, decomposes as a direct sum into irreducibles.

We take the standard basis for $\Cl(0,4)$ mentioned above, $\{e_{\vec{x}}:\vec{x}\in\{0,1\}^N\}$, and apply $\p_+$ (resp.\ $\p_-$) on it.  The result spans $\img(\p_+)$ (resp. $\img(\p_-)$), but there are duplications (up to sign).  For instance, in $\img(\p_+)$, a vertex $\p_+(e_{\vec{x}})$ and $\p_+(e_{\vec{x}})\,g$ will be identified. Since
\begin{equation}
 \p_+(e_{(x_1,x_2,x_3,x_4)}){\cdot}g ~\propto~\p_+(e_{(1-x_1,1-x_2,1-x_3,1-x_4)}),
\end{equation}
up to an overall sign, body-diagonally opposite vertices are identified in $\img(\p_+)$.

More specifically, we start with $2\p_+(1)=1+\G_1\G_2\G_3\G_4$, and then successively apply $\G_1$, $\G_2$ and $\G_3$ on the left\footnote{The factor of $2$ is irrelevant but makes the formulas simpler.}:
\begin{subequations}
\begin{alignat}{3}
e_{(0000)}&\Defl1+\G_1\G_2\G_3\G_4,\\
e_{(1000)}&\Defl\G_1\cdot e_{(0000)}&&=\G_1+\G_2\G_3\G_4,\\
e_{(0100)}&\Defl\G_2\cdot e_{(0000)}&&=\G_2-\G_1\G_3\G_4,\\
e_{(0010)}&\Defl\G_3\cdot e_{(0000)}&&=\G_3+\G_1\G_2\G_4,\\
e_{(1100)}&\Defl\G_1\G_2\cdot e_{(0000)}&&=\G_1\G_2-\G_3\G_4,\\
e_{(1010)}&\Defl\G_1\G_3\cdot e_{(0000)}&&=\G_1\G_3+\G_2\G_4,\\
e_{(0110)}&\Defl\G_2\G_3\cdot e_{(0000)}&&=\G_2\G_3-\G_1\G_4,\\
e_{(1110)}&\Defl\G_1\G_2\G_3\cdot e_{(0000)}&&=\G_1\G_2\G_3-\G_4,
\end{alignat}
\end{subequations}
produces $2^3=8$ basis elements.
 Applying $\G_4$ within $\img(\p_+)$ produces no new vertex; for example,
\begin{alignat}{5}
 \G_4e_{(0000)}&=\G_4(1+\G_1\G_2\G_3\G_4)&&=\G_4-\G_1\G_2\G_3&&=-e_{(1110)},\\
 \G_4e_{(1000)}&=\G_4(\G_1+\G_2\G_3\G_4)&&=-\G_1\G_4+\G_2\G_3&&=e_{(0110)},
 \quad\text{\etc}
\end{alignat}
Constructing the Isoscalar supermultiplet corresponding to $\img(\p_+)$, using these vectors as a basis, turns
\begin{equation}
 e_{\vec{x}} \mapsto
 \left\{\begin{array}{ll}
         \f_{\vec{x}} &\text{if $\wt(\vec{x})$ is even,}\\[2mm]
         \j_{\vec{x}} &\text{if $\wt(\vec{x})$ is odd,}
        \end{array}\right.
 \label{eFBass}
\end{equation}
and produces a supermultiplet that may be depicted as:
\begin{equation}
 \vC{
 \begin{picture}(60,40)(5,-5)
 \put(3,3){\includegraphics[height=27mm]{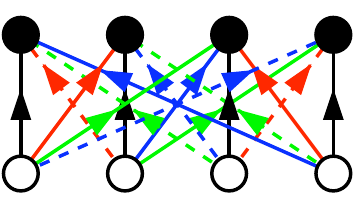}}
 \put(-14,15){$\img(\p_+):$}
 \put(0,30){$\j_{(1000)}$}
 \put(14,30){$\j_{(0100)}$}
 \put(28,30){$\j_{(0010)}$}
 \put(42,30){$\j_{(1110)}$}
 \put(0,0){$\f_{(0000)}$}
 \put(14,0){$\f_{(1100)}$}
 \put(28,0){$\f_{(1010)}$}
 \put(42,0){$\f_{(0110)}$}
 \end{picture}}
 \label{A:D4}
\end{equation}
The corresponding transformation rules are given explicitly in Table~\ref{t:D4}.
\begin{table}[htbp]
  \centering
  \begin{tabular}{@{} c|cccc @{}}
  \omit& {\bf Black Edges} & \cR{\bf Red Edges}
      & \cG{\bf Green Edges} & \cB{\bf Blue Edges} \\[1mm]
  & $Q_1$ & \cR{$Q_2$} & \cG{$Q_3$} & \cB{$Q_4$}\\
    \hline
 $\f_{(0000)}$ &$\j_{(1000)}$ &$~~\j_{(0100)}$ &$~~\j_{(0010)}$ &$-\j_{(1110)}$ \\ 
 $\f_{(1100)}$ &$\j_{(0100)}$ &$-\j_{(1000)}$  &$~~\j_{(1110)}$ &$~~\j_{(0010)}$ \\ 
 $\f_{(1010)}$ &$\j_{(0010)}$ &$-\j_{(1110)}$  &$-\j_{(1000)}$  &$-\j_{(0100)}$ \\ 
 $\f_{(0110)}$ &$\j_{(1110)}$ &$~~\j_{(0110)}$ &$-\j_{(0100)}$  &$~~\j_{(1000)}$ \\[2mm]
 $\j_{(1000)}$
  &$i\dot\f_{(0000)}$ &$-i\dot\f_{(1100)}$ &$-i\dot\f_{(1010)}$ &$~~i\dot\j_{(0110)}$ \\ 
 $\j_{(0100)}$
  &$i\dot\f_{(1100)}$ &$~~i\dot\f_{(0000)}$ &$-i\dot\f_{(0110)}$ &$-i\dot\j_{(1010)}$ \\ 
 $\j_{(0010)}$
  &$i\dot\f_{(1010)}$ &$~~i\dot\f_{(0110)}$ &$~~i\dot\f_{(0000)}$ &$~~i\dot\j_{(1100)}$ \\ 
 $\j_{(0001)}$
  &$i\dot\f_{(0110)}$ &$-i\dot\f_{(1010)}$ &$~~i\dot\f_{(1100)}$ &$-i\dot\j_{(0000)}$ \\ 
  \end{tabular}
  \caption{The action of $Q_1,Q_2,Q_3,Q_4$ on $\f_{\vec{x}}$ and $\j_{\vec{x}}$ within $\img(\p_+)$, abbreviating an explicit system of the form\eq{e3cube}: For example, $Q_1\,\f_{(0000)}=\j_{(1000)}$ and $Q_3\,\j_{(0100)}=-i\dot\f_{(0110)}$, \etc\ In the Adinkra\eq{A:D4}, black edges correspond to $Q_1$ action, red to $Q_2$, green to $Q_3$, and blue to $Q_4$.}
  \label{t:D4}
\end{table}

The result\eq{A:D4} is the four-dimensional cube with opposite corners identified, as in Refs.\cite{rA,r6-3}.  In the language of Ref.~\cite{r6-3}, this is a quotient of the four-dimensional cube by the code $d_4=\{0000,1111\}$.  This corresponds to the fact that by either doing nothing, or by reversing all four bits of a vertex, we return to the same vertex (with perhaps a minus sign).  In reference to the code name, the topology of the four-dimensional cube with opposite corners identified will be called $D_4$.

\subsubsection{The Two Inequivalent $D_4$ Quotients}
Similarly, we can find the Adinkra for the image of $\p_-$, using $1-\G_1\G_2\G_3\G_4$ instead.  When we do so, we see that the image of $\p_-$ and of $\p_+$ look similar---indeed, they have the same topology, the $4$-cube with opposite corners identified.  But they have different patterns of dashed edges.  These patterns cannot be made to coincide even when we redefine some of the vertices by replacing them with their negatives.  There are, in fact, two distinct irreducible representations of $\Cl(0,5)$, and these are the two.  Nevertheless, we can easily describe the relationship between them: by replacing $\G_4\mapsto-\G_4$, for instance, which corresponds to replacing $Q_4$ with $-Q_4$.  This, in turn, results in reversing the sign associated to each edge with $I=4$.  The Adinkras for these two are as follows:
\begin{equation}
 \vC{
 \begin{picture}(120,40)(5,-5)
 \put(3,3){\includegraphics[height=30mm]{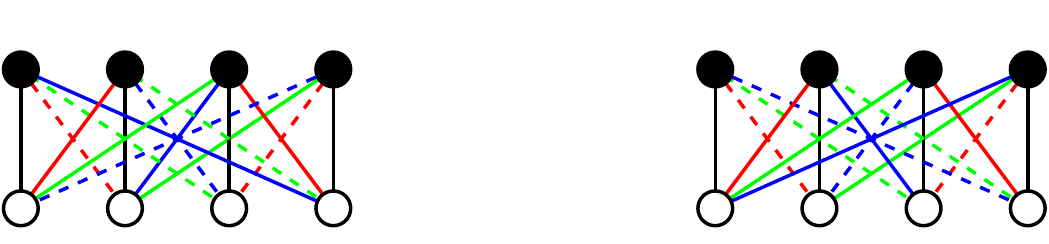}}
 \put(-14,15){$\img(\p_+):$}
 \put(0,30){$\j_{(1000)}$}
 \put(14,30){$\j_{(0100)}$}
 \put(28,30){$\j_{(0010)}$}
 \put(42,30){$\j_{(1110)}$}
 \put(0,0){$\f_{(0000)}$}
 \put(14,0){$\f_{(1100)}$}
 \put(28,0){$\f_{(1010)}$}
 \put(42,0){$\f_{(0110)}$}
 \put(68,15){$\img(\p_-):$}
 \put(82,30){$\j_{(1000)}$}
 \put(96,30){$\j_{(0100)}$}
 \put(110,30){$\j_{(0010)}$}
 \put(124,30){$\j_{(1110)}$}
 \put(83,0){$\f_{(0000)}$}
 \put(97,0){$\f_{(1100)}$}
 \put(111,0){$\f_{(1010)}$}
 \put(125,0){$\f_{(0110)}$}
 \end{picture}}
 \label{eB44a+b}
\end{equation}
We have suppressed the directions of the arrows:
 they are assumed to always point upward.  The
 replacement of $Q_4$ by $- \, Q_4$ is seen in the
 two Adinkras above by noting that all solid blue
 lines in the left hand Adinkra are replaced by
 dashed blue lines in the in the right hand Adinkra
 (and vice-versa).
 
Since the dimensional reduction of the $N=1$ chiral superfield in dimension $D=4$ down to dimension $D=1$ results in the image of $\p_+$, Readers might be tempted to think that the image of $\p_-$ arises from the dimensional reduction of the $N=1$, $D=4$ antichiral superfield.  This, however, is not the case.  Rather, it distinguishes different ways of reconstituting the $Q_1,\cdots,Q_4$ into $Q_\a$ and $Q^\dag_{\dot{\a}}$, with $\a,\dot\a=1,2$, as usual.  To go from one to the other Adinkra\eq{eB44a+b}, we must complex conjugate not both of the components of $Q_\alpha$, but only {\em half\/} of them---which is impossible without violating Lorentz symmetry in four dimensions.  If we complex conjugate all of $Q_\alpha$, we do swap chiral with antichiral superfields, but we reverse {\em\/both\/} $Q_3$ and $Q_4$. This change can be reversed by a redefinition of the real component fields, which swaps their complex combinations into the complex conjugates.
 Graphically, the action of complex conjugation would
 require that {\em {two}} colors must be used to
 implement the dashed/solid exchanges.  So for example,
 both solid blue lines {\em {and}} solid green lines
 in the left hand
 Adinkra are replaced by dashed blue lines  {\em {and}}
 dashed green lines in the in the right hand Adinkra
 (and vice-versa).

 The distinction between the two nonisomorphic supermultiplets depicted by the Adinkras\eq{eB44a+b} is thus more subtle. In fact, in more than 2-dimensional spacetimes, the $Q_I$'s are not Lorentz-invariant, and the sign of only one of them cannot be changed without violating Lorentz symmetry. In 2-dimensional spacetime the supermultiplets depicted in\eq{eB44a+b} are called chiral and twisted-chiral\cite{rGHR}, and we adopt this nomenclature also for the worldline $N=4$ supersymmetry. In $(2,2)$-supersymmetric theories in 2-dimensional spacetime, the transformation between the two supermultiplets\eq{eB44a+b} has been identified\cite{rTwSJG1,rMP1} as the root of mirror symmetry\cite{rMMYau1,rMMYau2,rMMYau3}.
 
 Now, any Lagrangian term involving only one of these types of supermultiplets can just as well be written in terms of only the other type; in this sense they may be regarded as equivalent.
 However, these two supermultiplets may well mix in a Lagrangian, and in a way that prevents rewriting the Lagrangian in terms of only one or the other type of supermultiplet, as has been done in Ref.\cite{rGHR}.
 This feature makes the two representations of supersymmetry, corresponding to two distinct irreducible representations of $\Cl(0,5)$ and depicted by the Adinkras\eq{eB44a+b}, {\em\/usefully distinct\/}.

 Note that such two nonisomorphic irreducible representations of $\Cl(0,N{+}1)$ exist precisely when $N=0\pmod4$, according to Table~\ref{cliffordtable}, discussed in Section~\ref{s:CCR}.

\subsection{Projecting Twice: the $D_6$ Valise Supermultiplet}
 \label{s:D6}
For $N=6$, define the two elements
\begin{eqnarray}
 g_1&=&\G_1\G_2\G_3\G_4,\label{g61}\\
 g_2&=&\G_3\G_4\G_5\G_6.\label{g62}
\end{eqnarray}
As before, $g_1^2=g_2^2=1$.  Also note that $g_1$ and $g_2$ commute, and in fact,
\begin{equation}
 g_1g_2=g_2g_1=-\G_1\G_2\G_5\G_6.
 \label{e:g1g2}
\end{equation}
Analogously to the $D_4$ example, we define the four projection operators
\begin{eqnarray}
 \p_{1\pm}(v)=v{\cdot}\frac{1\pm g_1}{2},\\
 \p_{2\pm}(v)=v{\cdot}\frac{1\pm g_2}{2}
\end{eqnarray}
and note that since $g_1$ and $g_2$ commute, so do $\p_{1+}$ and $\p_{2+}$.  We will now project twice: once using $\p_{1+}$ (or $\p_{1-}$) and then again, using $\p_{2+}$ (or $\p_{2-}$)---a total of four choices.  Using $\p_{1+}$ and $\p_{2+}$ (for instance) produces the Clifford representation $\img(\p_{1+}\circ\p_{2+})$.  In this, and in what follows, we will use $\p_{1+}$ and $\p_{2+}$ for notational definiteness, but it is to be understood that these may be replaced by $\p_{1-}$ or $\p_{2-}$, respectively and independently, {\em mutatis mutandis}.

The composition is
\begin{align}
 (\p_{1+}\circ\p_{2+})(v)
 &= v{\cdot}\frac{1+g_2}{2}{\cdot}\frac{1+g_1}{2}
  =\frac{1}{4}v{\cdot}(1+g_1+g_2+g_2g_1)\\
 &=\frac{1}{4}v{\cdot}(1+\G_1\G_2\G_3\G_4+\G_3\G_4\G_5\G_6-\G_1\G_2\G_5\G_6).
\end{align}
It is straightforward to prove that $\img(\p_{1+}\circ\p_{2+})=\img(\p_{1+})\cap\img(\p_{2+})$ and that this is a Clifford representation.  Writing this as $\ker(\p_{1-})\cap\ker(\p_{2-})$, we can see that for all $v$ in this image, $v=v{\cdot}g_1=v{\cdot}g_2=v{\cdot}g_1g_2$.  In particular, when $g_1$ and $g_2$ commute with $v$, we have that $v=g_1\,v=g_2\,v=g_1g_2\,v$.

Define
 \begin{equation}
e_0 := (\p_{1+}\circ\p_{2+})(1)
 = \frac{1}{4}(1+\G_1\G_2\G_3\G_4+\G_3\G_4\G_5\G_6-\G_1\G_2\G_5\G_6).
\end{equation}
Note that $e_0$ commutes with $g_1$ and $g_2$, and so $e_0 = g_1e_0 = g_2e_0 = g_1g_2e_0$.  We successively apply the various $\G_1, \G_2, \G_3, \G_5$ to $e_0$ on the left and we get a collection of 16 fields that span $\img (\p_{1+}\circ \p_{2+})$, corresponding to vectors of the form $e_{(x_1,x_2,x_3,0,x_5,0)}$.  Applying $\G_4$ will not generate any new vectors because $\G_4 = g_1\G_1\G_2\G_3$, and so for every $e_{\vec{x}}$,
\begin{eqnarray}
\G_4 e_{\vec{x}} &=& g_1\G_1\G_2\G_3\cdot e_{\vec{x}}\\
&=& g_1 \G_1\G_2\G_3\G_1^{x_1}\cdots\G_N^{x_N}\cdot e_0\\
&=& \pm \G_1\G_2\G_3\G_1^{x_1}\cdots\G_N^{x_N} g_1\cdot e_0\\
&=& \pm \G_1\G_2\G_3\G_1^{x_1}\cdots\G_N^{x_N}\cdot e_0\\
&=&\pm \G_1\G_2\G_3\cdot e_{\vec{x}}.
\end{eqnarray}
Similarly, $\G_6$ will not generate new vectors, using $\G_6=g_2\G_3\G_4\G_5$.  

The Isoscalar supermultiplet corresponding to this will have 8 bosons and 8 fermions.  The $e_{\vec{x}}$ vectors correspond, according to \Eq{eFBass}, to the various $\f_{\vec{x}}$ and $\j_{\vec{x}}$.

If we do this, we get an Adinkra whose topology we call $D_6$, which is a six-dimensional cube projected twice: once according to $g_1=\G_1\G_2\G_3\G_4$ and then according to $g_2=\G_3\G_4\G_5\G_6$:
\begin{equation}
 \vC{
 \begin{picture}(120,45)(15,-5)
 \put(4,1){\includegraphics[height=30mm]{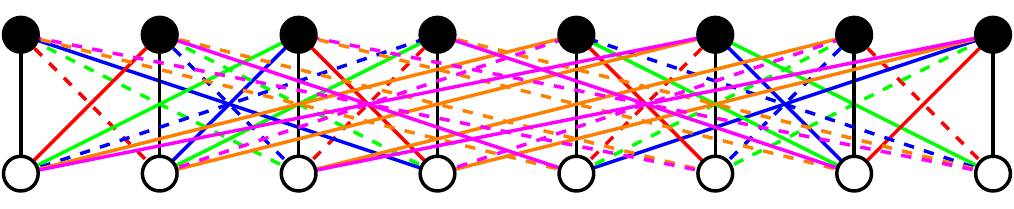}}
 \put(0,32){$\j_{(10000)}$}
 \put(20,32){$\j_{(01000)}$}
 \put(40,32){$\j_{(00100)}$}
 \put(60,32){$\j_{(11100)}$}
 \put(0,0){$\f_{(00000)}$}
 \put(20,0){$\f_{(11000)}$}
 \put(40,0){$\f_{(10100)}$}
 \put(60,0){$\f_{(01100)}$}
 \put(80,32){$\j_{(00001)}$}
 \put(100,32){$\j_{(11001)}$}
 \put(120,32){$\j_{(10101)}$}
 \put(140,32){$\j_{(01101)}$}
 \put(80,0){$\f_{(10001)}$}
 \put(100,0){$\f_{(01001)}$}
 \put(120,0){$\f_{(00101)}$}
 \put(140,0){$\f_{(11101)}$}
 \end{picture}}
 \label{eN6B88}
\end{equation}
Here the colors are as before, with orange for $Q_5$ and purple for $Q_6$.

The name $D_6$ derives from the code $d_6$, generated by the $k=2$ codewords $c_1=111100$ and $c_2=001111$; to denote this, we assemble the generator codewords as a matrix, called a {\em generator matrix} in coding theory:
\begin{equation}
 \left[\begin{smallmatrix}111100\\[3pt] 001111\end{smallmatrix}\right],
\end{equation}
but here we simply regard it as a collection of the codewords listed in its rows.  The first codeword corresponds to $\G_1\G_2\G_3\G_4$, the second with $\G_3\G_4\G_5\G_6$, and the sum of these codewords modulo 2, $110011$, corresponds to $-\G_1\G_2\G_5\G_6$.

More generally, addition in the code turns into multiplication of the corresponding products of $\G_I$ matrices, perhaps with a minus sign, because the $\G_I$ can be anticommuted past each other, and when a particular $\G_I$ appears in both $g_1$ and $g_2$, once anticommutation is done so that the two $\G_I$'s are adjacent, these two simplify to $\G_I{}^2=1$.

\subsection{Constructing a Valise Supermultiplet from a Code}
The procedure illustrated in the previous two examples can be generalized to the following construction, first applied to Clifford representations by A.~Dimakis.\cite{rDimakis}.

\begin{construction}\label{const:quotient}
Suppose we are given $N$ and a doubly-even code $C\subset(\ZZ/2)^N$ of length $N$, given by a generating set $\{c_1,\dots,c_k\}\subset C$. Writing each $c_i$ as $(x_{i1},\ldots,x_{iN})\in\{0,1\}^N$, we associate to it $g_i\Defl\G_1^{x_{i1}}\cdots\G_N^{x_{iN}}$. For instance, $c_1=101100100$ would produce $g_1=\G_1\G_3\G_4\G_7$.

The $\wt(c_i)$ being even translates into $g_i$ being even.  When $\wt(c_i)$ is even, $\wt(c_i)$ being a multiple of $4$ is equivalent to $g_i{}^2=1$.  It is an elementary fact about doubly even codes that any two elements of the codes are orthogonal (that is, share an even number of $1$s)\cite{r6-3,rCHVP}.  As a result, all of the $g_1,\ldots,g_k$ commute with each other.  They generate a group $G$ under multiplication that is isomorphic as a group to the code $C$.  The isomorphism is done analogously to converting $c_i$ to $g_i$, except that a minus sign is sometimes required.\footnote{This construction is due to J. Wood, who used it to classify 2-elementary abelian subgroups of the spin groups\cite{rJAW}.}

For each $g_i$ we have $\p_{i\pm}:\Cl(0,N)\to\Cl(0,N)$ defined as
\begin{equation}
\p_{i\pm}(v)=v{\cdot}\frac{1\pm g_i}{2}.
\end{equation}
As before, the operators $\p_{i\pm}$ are homomorphisms.  The evenness of $g_i$ implies that $\p_{i\pm}$ preserves the fermionic and bosonic statistics.  The fact that $g_i{}^2=1$ implies that $\p_{i+}$ and $\p_{i-}$ are projection operators.

The fact that the $g_i$ all commute implies that the $\p_{i+}$ and $\p_{i-}$ all commute.

The following table summarizes how the properties of the generators of the code relate to the properties of the $\p_{i\pm}$.
\begin{equation}
\begin{tabular}{c|c|c}
\boldmath\bf$c_i$ (Codeword)
                            &\boldmath$g_i$
                                          &\boldmath\bf$\p_{i\pm}$ (Projector)\\[1pt]
 \hline\hline
even weight                 & even        & preserves statistics\\
weight is~ $0\pmod4$   & $g_i{}^2=1$ & projection \\
pairwise orthogonal         & commute     & commute\\\hline
\end{tabular}
\end{equation}

We define\Ft{Strictly speaking, the notation\eqs{epiC}{eImC} should specify for each $g_i$ which sign is being used. For illustrative purposes, herein we only use $\p_{i+}$'s.  This choice will not affect chromotopology, but will affect the dashedness of the edges. Precisely which of these $2^k$ distinct choices in $\p_{1\pm}\circ\cdots\circ\p_{k\pm}$ provide (non-)isomorphic representations akin to\eq{eB44a+b} is a question we defer to a subsequent effort.}
\begin{equation}
 \p_C \Defl \p_{1+}\circ\cdots\circ\p_{k+},
 \label{epiC}
\end{equation}
and then as in the previous example,
\begin{equation}
 \img(\p_C)~=~\img(\p_{1+})\cap\cdots\cap\img(\p_{k+})
           ~=~\ker(\p_{1-})\cap\cdots\cap\ker(\p_{k-})
 \label{eImC}
\end{equation}
is the Clifford representation we want.  If we define
\begin{equation}
 e_0=\p_C(1)=\frac{1+g_1}{2}\cdots\frac{1+g_k}{2}=\frac{1}{2^k}\sum_{g\in G} g,
\end{equation}
and then successively apply the various $\G_I$ on the left to $e_0$, we obtain a set of elements of $\Cl(0,N)$.  Using the Dirac relations (\ref{eDirac}) shows that many of these are the same, up to an overall sign.  Furthermore, for all $g_i$, we have $g_i e_0 = e_0$, and more generally this is true of all elements of $G$.  It therefore follows that if we begin with an element $v$ represented by a $d_v$-dimensional matrix, the application of $\p_C$ results in a quantity representable by a $2^{-k}d_v$-dimensional matrix.

In fact, this will result in $2^{N-k+1}$ different elements of $\img(\p_C)$, occurring in $\pm$ pairs.  If we arbitrarily choose one from each $\pm$ pair, the result is a collection of $2^{N-k}$ vectors that form a basis for $\img(\p_C)$.
\end{construction}

\subsection{Examples for $N\le 8$} 

\begin{example}
For $N<4$, there are no doubly even codes, and so no such quotients are possible.  All the Adinkras for $1\leq N<4$ are therefore $N$-cubes.  We describe the topology as $I^N$, where $I=[0,1]$.
\end{example}

\begin{example}
For $N=4$, there are two Adinkra topologies: $I^4$ (for the trivial code $\{0000\}$), and $D_4$, which was already explained above.
\end{example}

\begin{example}
For $N=5$, there are several doubly even codes to choose from, but they are all permutation equivalent.  For instance, the code generated by $11110$ gives $g=\G_1\G_2\G_3\G_4$, which is the same as that in\eq{g4}.  Or we could use the code generated by $10111$, giving $g'=\G_1\G_3\G_4\G_5$, for instance.  But we cannot use both at the same time since $11110$ and $10111$ sum to $01001$, which has weight 2, and which encodes the fact:
\begin{equation}
 \big[\,\G_1\G_2\G_3\G_4\,,\,\G_1\G_3\G_4\G_5\,\big] = 2\G_2\G_5\neq0.
\end{equation}

We therefore find 5 different Adinkras:  They differ only in selecting the colors for the edges and the dashing of the edges.  Therefore, they have the same topology but different chromotopologies.  Again, the distinction between these may not at first seem relevant, inasmuch as the swapping of colors corresponds to swapping the $Q_I$ in supersymmetry, and the labeling of which $Q_I$ are which is somewhat arbitrary.  But these colorings can be relevant when a model consists of two or more of these supermultiplets used together, and we can ask whether these choices are coordinated or not. While this distinction is subtler than that one discussed in the $N=4$, $D_4$ case---after all, $\Cl(0,6)$ has only one (equivalence class of) irreducible representation(s)---it may still be useful. Two such equivalent but distinct supermultiplets, one being projected by $(1+g)/2$ and the other by $(1+g')/2$ respectively, may well mix in a Lagrangian in a way that prevents rewriting the Lagrangian in terms of only one or the other type of supermultiplet. Should this turn out to be possible, the five supermultiplets corresponding to the 5 different Adinkras would also be {\em\/usefully distinct\/}---much as, say, chiral and twisted-chiral supermultiplets are.

Let us consider the code generated by $11110$, so that $g=\G_1\G_2\G_3\G_4$.  To find the Adinkra, we can note that every vertex is equivalent to one of the form $(x_1,x_2,x_3,0,x_5)$.  To see how to connect the edges, we may apply $Q_I$ to expressions of the form $Q_1^{x_1}Q_2^{x_2}Q_3^{x_3}Q_5^{x_5}\f_{(00000)}$, and use the supersymmetry algebra and the quotient with $\G_1\G_2\G_3\G_4$.

This is the result:
\begin{equation}
 \vC{
 \begin{picture}(120,45)(15,-5)
 \put(4,1){\includegraphics[height=30mm]{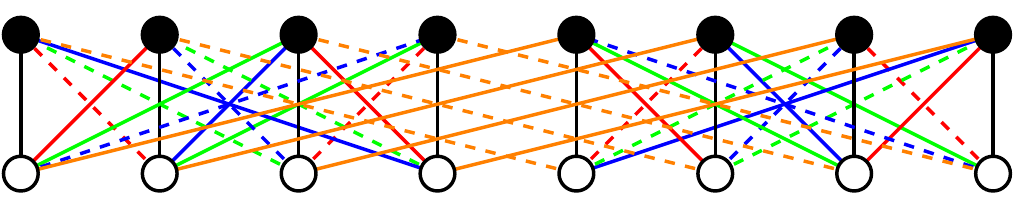}}
 \put(0,32){$\j_{(10000)}$}
 \put(20,32){$\j_{(01000)}$}
 \put(40,32){$\j_{(00100)}$}
 \put(60,32){$\j_{(11100)}$}
 \put(0,0){$\f_{(00000)}$}
 \put(20,0){$\f_{(11000)}$}
 \put(40,0){$\f_{(10100)}$}
 \put(60,0){$\f_{(01100)}$}
 \put(80,32){$\j_{(00001)}$}
 \put(100,32){$\j_{(11001)}$}
 \put(120,32){$\j_{(10101)}$}
 \put(140,32){$\j_{(01101)}$}
 \put(80,0){$\f_{(10001)}$}
 \put(100,0){$\f_{(01001)}$}
 \put(120,0){$\f_{(00101)}$}
 \put(140,0){$\f_{(11101)}$}
 \end{picture}}
 \label{eB88}
\end{equation}

Note that for this Adinkra, there is an asymmetry: if we take the $g$ as described in\eq{g4}, then $\G_5$ behaves differently from the other $\G_I$.  Of course, if we had chosen a different code, so that $g'$ involved any other four of the five $\G_I$'s, the one $\G_I$ that was omitted from $g'$ would behave differently.

If we slightly rearrange the vertices for convenience, we can see that this Adinkra is the cartesian product of the line segment with the $N=4$ projective cube Adinkra: $D_4\times I^1$:
\begin{equation}
 \vC{
 \begin{picture}(120,50)(22,-5)
 \put(4,1){\includegraphics[height=40mm]{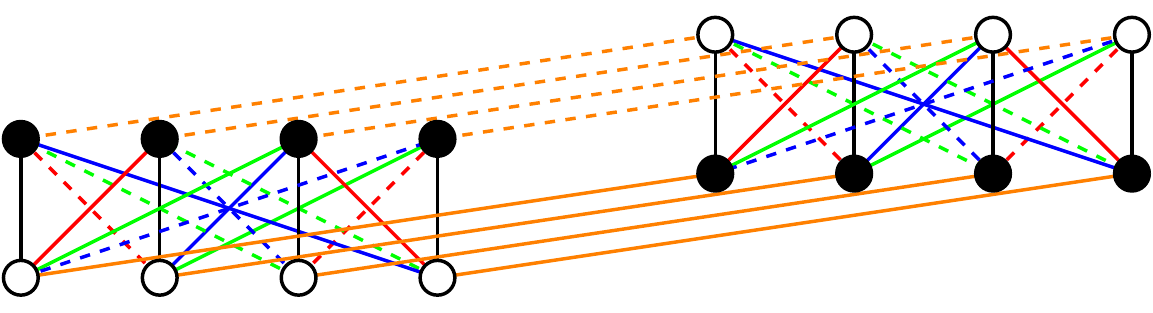}}
 \put(0,29){$\j_{(10000)}$}
 \put(17,29){$\j_{(01000)}$}
 \put(34,29){$\j_{(00100)}$}
 \put(51,29){$\j_{(11100)}$}
 \put(0,0){$\f_{(00000)}$}
 \put(17,0){$\f_{(11000)}$}
 \put(34,0){$\f_{(10100)}$}
 \put(51,0){$\f_{(01100)}$}
 \put(89,13){$\j_{(00001)}$}
 \put(106,13){$\j_{(11001)}$}
 \put(124,13){$\j_{(10101)}$}
 \put(142,13){$\j_{(01101)}$}
 \put(89,41){$\f_{(10001)}$}
 \put(106,41){$\f_{(01001)}$}
 \put(124,41){$\f_{(00101)}$}
 \put(142,41){$\f_{(11101)}$}
 \end{picture}}
 \label{eB88proj}
\end{equation}

Thus, in $N=5$, we have two Adinkra topologies: $I^5=I^4\times I^1$, and $D_4\times I^1$.  These come from the two Adinkra topologies for $N=4$, and indeed, in general, with each new $N$, we get all the topologies from the previous $N$ multiplied by $I$, and possibly new topologies.
\end{example}

The process of going from the $N=4$ Adinkra $D_4$ to the $N=5$ Adinkra $D_4\times I^1$ is a general one.  If we have an Adinkra for a particular value of $N$, with vertices $(v_1,\ldots,v_{2m})$, we can create an Adinkra for $N{+}1$-extended supersymmetry, by creating new vertices $(w_1,\ldots,w_{2m})$, where $w_i$ is bosonic if and only if $v_i$ is fermionic.  Then draw the same edges connecting the $(v_1,\ldots,v_{2m})$ with the same colors, dashings, and arrows, and drawing corresponding edges with the same colors, dashings, and arrows on the $(w_1,\ldots,w_{2m})$.  Then we draw edges in the $(N{+}1)^\text{st}$ color from each $v_i$ to its corresponding $w_i$, pointing from $v_i$ to $w_i$, solid if $v_i$ is bosonic and dashed otherwise.

The effect on the topology is that we make an extra copy of the graph, and connect each pair of corresponding vertices with an edge.  Intuitively, this is the same as taking the vertices and edges of the cartesian product with $I=[0,1]$.  In terms of the code, we extend the code by adding a new column that is always zero.  In this case, we turn $\{0000,1111\}$ into $\{00000,11110\}$.

\begin{example}
For $N=6$, there are several possible doubly even codes.  As before, there is still $I^6$.  If we use one codeword, we could, for instance, take $111100$ or indeed any other codeword consisting of four 1's and two 0's.  There are ${6\choose 4}=15$ such choices, and each results in an Adinkra with topology $D_4\times I^2$.  Again, these are only different based on the coloring of the edges.

Then there is the $D_6$ example above (Section~\ref{s:D6}), of which there are  $\frac{6!}{2^3\cdot 3!}=15$ codes that are permutation equivalent to this one, and thus, although they give different Adinkras, these only differ in their edge coloring.
\end{example}

\begin{example}
The case $N=7$ allows for the following generator set:
\begin{equation}
 e_7 \gen \left[\begin{smallmatrix}
                  1111000\\[2pt]
                  0011110\\[2pt]
                  1010101
                \end{smallmatrix}\right],
\end{equation}
resulting in the following generators for $G$:
\begin{eqnarray}
 g_1&=&\G_1\G_2\G_3\G_4,\label{g71}\\
 g_2&=&\G_3\G_4\G_5\G_6,\label{g72}\\
 g_3&=&\G_1\G_3\G_5\G_7.\label{g73}
\end{eqnarray}
This generates
\begin{align}
 g_1g_2&=-\G_1\G_2\G_5\G_6,&
 g_1g_3&=\4\G_2\G_4\G_5\G_7,\\
 g_2g_3&=\4\G_1\G_4\G_6\G_7,&
 g_1g_2g_3&=-\G_2\G_3\G_6\G_7,
\end{align}
which, together with 1, constitute the maximal doubly even code, $e_7$.  We accordingly call this Adinkra topology $E_7$:
\begin{equation}
 \vC{
 \begin{picture}(120,45)(15,-5)
 \put(4,1){\includegraphics[height=30mm]{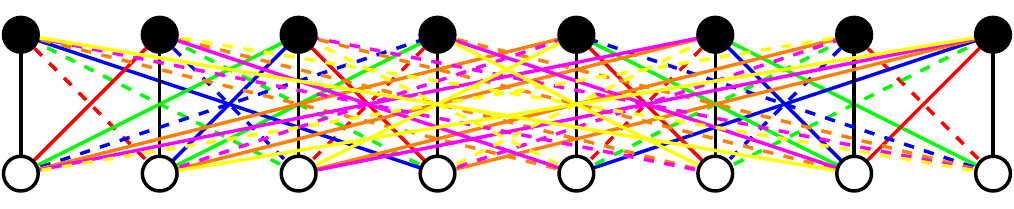}}
 \put(0,32){$\j_{(10000)}$}
 \put(20,32){$\j_{(01000)}$}
 \put(40,32){$\j_{(00100)}$}
 \put(60,32){$\j_{(11100)}$}
 \put(0,0){$\f_{(00000)}$}
 \put(20,0){$\f_{(11000)}$}
 \put(40,0){$\f_{(10100)}$}
 \put(60,0){$\f_{(01100)}$}
 \put(80,32){$\j_{(00001)}$}
 \put(100,32){$\j_{(11001)}$}
 \put(120,32){$\j_{(10101)}$}
 \put(140,32){$\j_{(01101)}$}
 \put(80,0){$\f_{(10001)}$}
 \put(100,0){$\f_{(01001)}$}
 \put(120,0){$\f_{(00101)}$}
 \put(140,0){$\f_{(11101)}$}
 \end{picture}}
 \label{eN7B88}
\end{equation}

The fields are defined in the same way as in $D_4\times I^1$ and in $D_6$.  Here, the new color, yellow, represents the action of multiplying from the left by $\G_7$.

There are $\frac{7!}{168}=30$ of permutation equivalent codes to $e_7$, and again, these differ only in edge-colorings; the topology of all these is however identical.

As before, we could take $g_1$ and $g_2$, or just $g_1$, and generate a smaller group if we wish.  In those cases, we end up with a cartesian product of a lower-dimensional cube and the Adinkra of the dimension where that generator first appeared.  The possibilities here are $I^7$, $D_4\times I^3$, and $D_6\times I$.
\end{example}

\begin{example}
The case $N=8$ allows the code $e_8$, which has the following generator set:
\begin{equation}
 e_8 \gen \left[\begin{smallmatrix}
                   11110000\\[2pt]
                   00111100\\[2pt]
                   00001111\\[2pt]
                   10101010
                \end{smallmatrix}\right].
\end{equation}
The corresponding $g_i$ are:
\begin{eqnarray}
 g_1&=&\G_1\G_2\G_3\G_4,\\
 g_2&=&\G_3\G_4\G_5\G_6,\\
 g_3&=&\G_5\G_6\G_7\G_8,\\
 g_4&=&\G_1\G_3\G_5\G_7.
\end{eqnarray}
These produce the following group elements:
\begin{subequations}
\begin{align}
 g_1g_2&=-\G_3\G_4\G_5\G_6,&
 g_1g_3&=\4\G_1\G_2\G_3\G_4\G_5\G_6\G_7\G_8,\\
 g_1g_4&=\4\G_2\G_4\G_5\G_7,&
 g_2g_3&=-\G_3\G_4\G_7\G_8,\\
 g_2g_4&=\4\G_1\G_4\G_6\G_7,&
 g_3g_4&=\4\G_1\G_3\G_6\G_8,\\
 g_1g_2g_3&=\4\G_3\G_4\G_7\G_8,&
 g_1g_2g_4&=-\G_1\G_4\G_6\G_7,\\
 g_1g_3g_4&=\4\G_2\G_4\G_6\G_8,&
 g_2g_3g_4&=-\G_1\G_3\G_4\G_8,\\
 g_1g_2g_3g_4&=\4\G_1\G_4\G_5\G_8.
\end{align}
\end{subequations}
Recall that each $g$ must be a product of a doubly even number of $\G_I$'s, explaining why are we only now seeing $g$'s with differing numbers of $\G_I$'s.

The Adinkra is below.  This has the feature that every boson is connected to every fermion.  So this is a $K(8,8)$ graph, and is denoted $E_8$.  This was introduced in Refs.\cite{rGR0,rGLPR} as the $N=8$ spinning particle in relation to what was described as a ``supergravity surprise''.  The new color, brown, corresponds to $\G_8$.
\begin{equation}
 \vC{
 \begin{picture}(120,45)(15,-5)
 \put(4,1){\includegraphics[height=30mm]{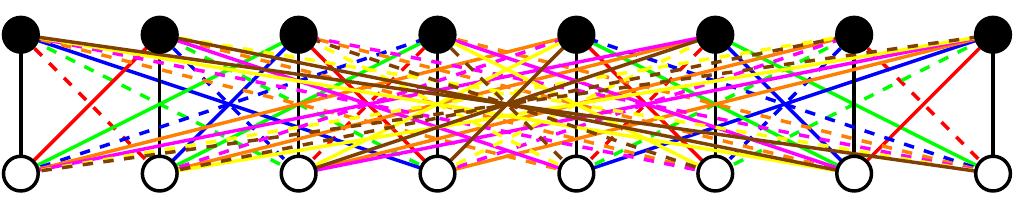}}
 \put(0,32){$\j_{(10000)}$}
 \put(20,32){$\j_{(01000)}$}
 \put(40,32){$\j_{(00100)}$}
 \put(60,32){$\j_{(11100)}$}
 \put(0,0){$\f_{(00000)}$}
 \put(20,0){$\f_{(11000)}$}
 \put(40,0){$\f_{(10100)}$}
 \put(60,0){$\f_{(01100)}$}
 \put(80,32){$\j_{(00001)}$}
 \put(100,32){$\j_{(11001)}$}
 \put(120,32){$\j_{(10101)}$}
 \put(140,32){$\j_{(01101)}$}
 \put(80,0){$\f_{(10001)}$}
 \put(100,0){$\f_{(01001)}$}
 \put(120,0){$\f_{(00101)}$}
 \put(140,0){$\f_{(11101)}$}
 \end{picture}}
 \label{eN8B88}
\end{equation}

As in the $N=4$ case, there are two irreducible representations, and one is obtained as above, while the other is obtained by reversing the sign on one of the $g_i$.  The result is the same Adinkra topology, but with various signs on the edges reversed.  For instance, we can reverse the sign on $g_1$ by reversing the signs on edges corresponding to $\G_4$.  This preserves the sign on $g_2$, $g_3$, and $g_4$.  Reversing various $\G_I$ yields apparently different Adinkras, but all of these must fall into just two isomorphism classes.

There are $\frac{8!}{1344}=30$ codes that are permutation equivalent to $e_8$.

As before, we can take a subgroup of $e_8$, but in this case, the situation is a bit more interesting: there are multiple inequivalent choices for which generator to remove.  Figure~\ref{fig:e8tree} shows all the $N=8$ doubly even codes, up to permutation equivalence.

\begin{figure}[htb]
\begin{center}
\begin{picture}(70,45)(0,5)
\put(10,45){$k=4:$}
 \put(50,45){\makebox[0in][c]{$e_8$}}
  \put(40,39){\vector(3,2){7}}
  \put(60,39){\vector(-3,2){7}}
\put(10,35){$k=3:$}
 \put(62,35){\makebox[0in][c]{$d_8$}}
 \put(38,35){\makebox[0in][c]{$e_7\oplus t_1$}}
  \put(38,29){\vector(0,1){4}}
  \put(62,29){\vector(0,1){4}}
  \put(42,28.5){\vector(3,1){15}}
\put(10,25){$k=2:$}
 \put(38,25){\makebox[0in][c]{$d_6\oplus t_2$}}
 \put(62,25){\makebox[0in][c]{$d_4\oplus d_4$}}
  \put(38,19){\vector(0,1){4}}
  \put(62,19){\vector(0,1){4}}
  \put(42,18.5){\vector(3,1){15}}
\put(10,15){$k=1:$}
 \put(38,15){\makebox[0in][c]{$d_4\oplus t_4$}}
 \put(62,15){\makebox[0in][c]{$h_8$}}
  \put(53,9){\vector(3,2){7}}
  \put(47,9){\vector(-3,2){7}}
\put(10,5){$k=0:$}
 \put(50,5){\makebox[0in][c]{$t_8$}}
\end{picture}
\caption{$N=8$ doubly even codes and their subset relationships: The arrows connecting two codes are injections (possibly after permutation): the code on the lower level is a subcode of the higher one.  The trivial code $t_8$ is at the bottom, and $e_8$ is the unique maximal doubly even code for $N=8$ and is drawn at the top.  Every doubly even code in $N=8$ is a subcode of $e_8$.  The code $h_8$ is the one generated by $11111111$.}\label{fig:e8tree}
\end{center}
\end{figure}
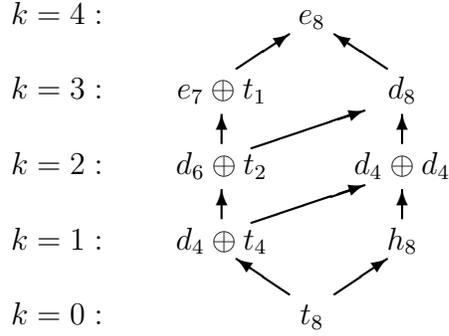

For instance, for $k=3$, we could choose $e_7\oplus t_1$, which results in an Adinkra topology $E_7\times I^1$, or we could choose $d_8$, which results in a different Adinkra topology $D_8$:
\begin{align}
 E_7\times I^1&: \vC{\includegraphics[width=140mm]{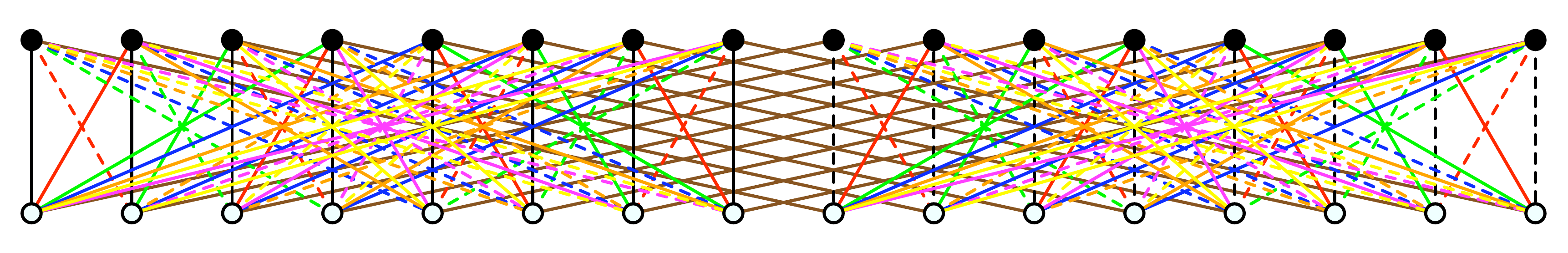}} \label{eE7xI}\\[2mm]
 D_8&:         \vC{\includegraphics[width=140mm]{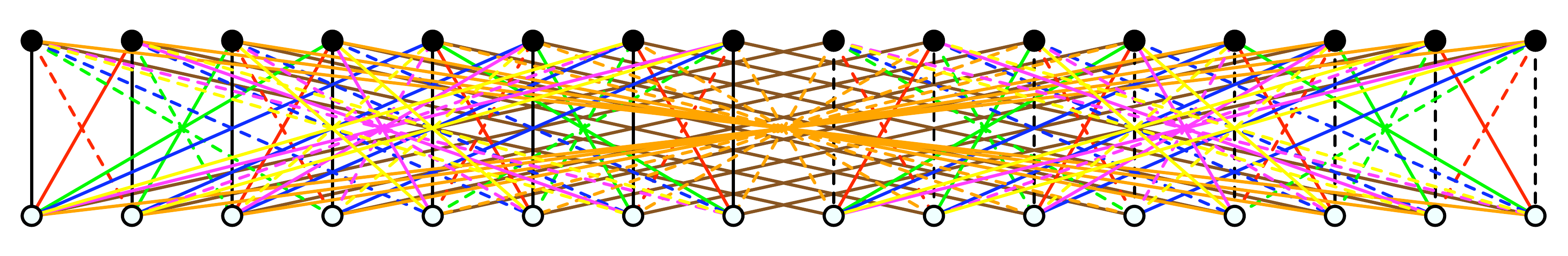}} \label{eD8}
\end{align}
The only difference between these two Adinkras is the way one of the supersymmetries acts, the one represented here by orange edges. In the $E_7\times I^1$-Adinkra\eq{eE7xI}, it acts within each of the two halves, leaving the Adinkra {\em\/1-color-decomposable\/}: only the brown edges span the whole Adinkra, and it decomposes into two identical $N=7$ Adinkras if the brown edges are erased. Each of the halves has the $E_7$ topology, except that edges of one color (black) have their dashing reversed. This becomes clearer upon rearranging the nodes a little:
\begin{equation}
 E_7\times I^1:~ \vC{\includegraphics[width=140mm]{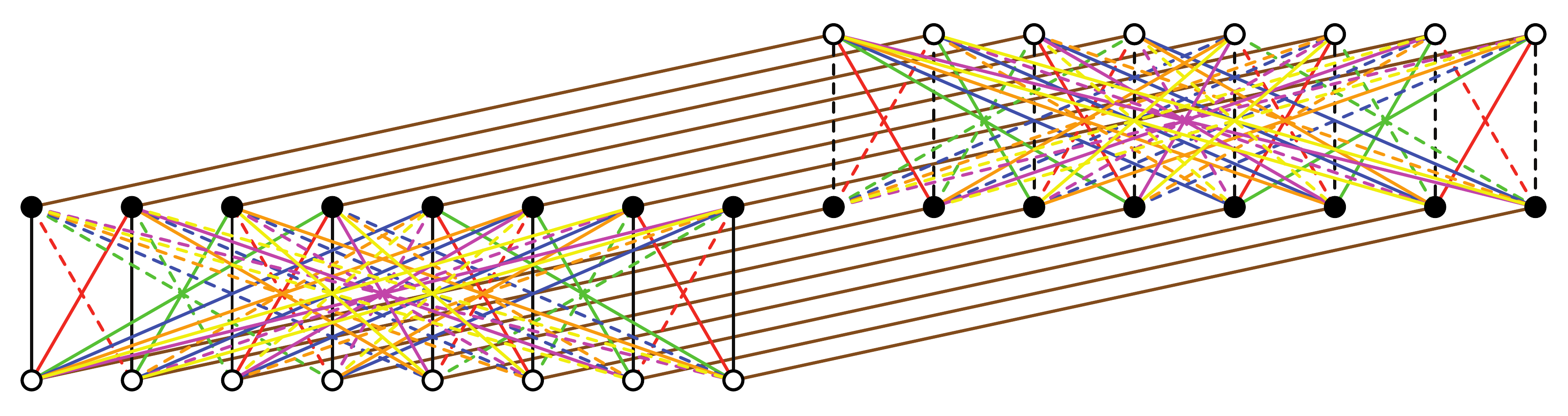}}
 \label{eE7xIa}
\end{equation}
much the same as done in the display\eq{eB88} \textit{vs.}\eq{eB88proj}, for $D_4\times I^1$.

In turn, in the $D_8$-Adinkra\eq{eD8}, the orange edges span the whole Adinkra together with the brown ones. For the sake of comparison with the $E_7\times I^1$-Adinkra\eq{eE7xIa}, we also rearrange the nodes of the $D_8$-Adinkra:
\begin{equation}
 D_8:~ \vC{\includegraphics[width=140mm]{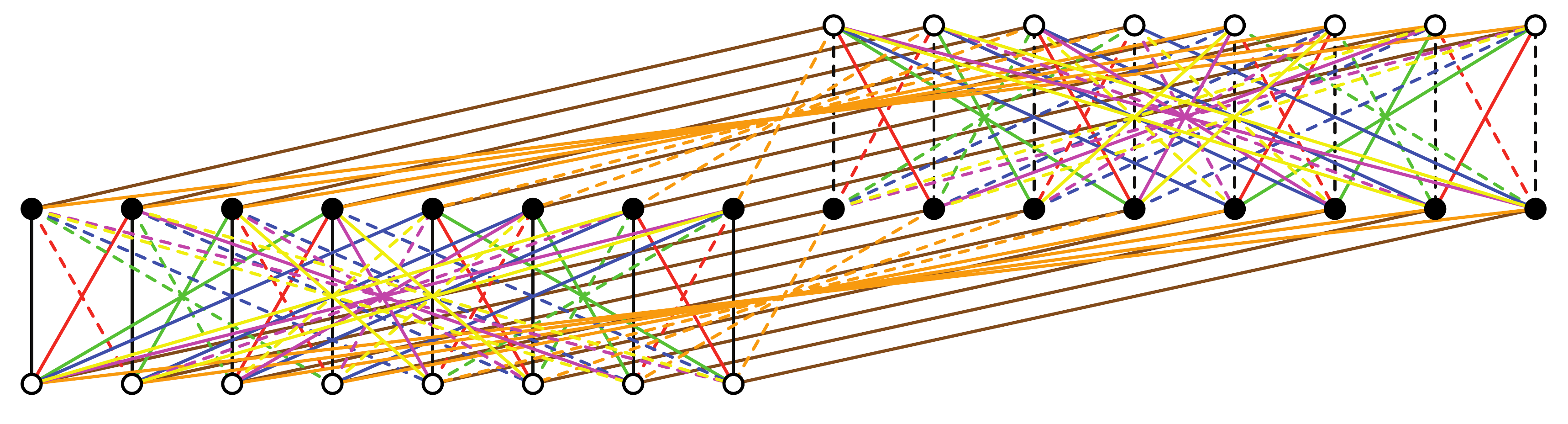}}
 \label{eD8a}
\end{equation}
To decompose this second Adinkra, one would have to erase the edges of at least two colors; we say it is 2-color-decomposable. This $n$-color-decomposability property correlates with the fact that the $e_7\oplus t_1$ code has one column of zeros, whereas $d_8$ does not:
\begin{equation}
 E_7\times I^1 \iff e_7\oplus t_1
  \gen\left[\begin{smallmatrix}
              0000\,1111\\[2pt]
              0011\,1100\\[2pt]
              0101\,0101
            \end{smallmatrix}\right],
 \qquad\text{{\it vs.\/}}\qquad
 D_8 \iff d_8
  \gen\left[\begin{smallmatrix}
              0000\,1111\\[2pt]
              0011\,1100\\[2pt]
              1111\,0000
            \end{smallmatrix}\right].
\end{equation}
Finally, both Adinkras\eq{eE7xI} and\eq{eD8} admit precisely one supersymmetry-preserving $\ZZ_2$ symmetry: for the former, it is encoded as $11110000$ and generated by $\G_1\G_2\G_3\G_4$, while the latter is symmetric with respect to the action of $\G_2\G_4\G_6\G_8$, encoded as $01010101$. The projection of each Adinkra by its respective symmetry then produces the $E_8$-Adinkra\eq{eN8B88}.

Similar relations exist between the $D_6\times I^2$- \textit{vs}.\ $D_4\times D_4$-Adinkras, and the $D_4\times I^4$- \textit{vs}.\ $H_8$-Adinkra corresponding to the $k=2$ and $k=1$ rows in the diagram in Figure~\ref{fig:e8tree}.
\end{example}

To summarize, Table~\ref{topologytable} gives the possible topologies for each $N$ up to $N=10$.  The cartesian product $\times$ refers to taking the cartesian product of the vertex set, and drawing edges between $(v_1,w_1)$ and $(v_1,w_2)$ whenever there is an edge between $w_1$ and $w_2$, and between $(v_1,w_1)$ and $(v_2,w_1)$ whenever there is an edge between $v_1$ and $v_2$.  Details of this construction are given in Section~\ref{s:multiply}.  Exponentiation means iterated cartesian products.  We note the Adinkra topologies $D_4$, $D_6$, $D_8$, and $D_{10}$ from the codes $d_4$, $d_6$, $d_8$ and $d_{10}$, as well as $E_7$ and $E_8$ from $e_7$ and $e_8$, respectively.  In addition, there turns up a code with no standard name, generated by $\{1111000000,0011111111\}$, which we tentatively call $D_4*_1\!I^6$.

\begin{table}[ht]
\begin{center}\footnotesize
\begin{tabular}{c|rc|rc|rc|rc|rc}
\boldmath$N$ & \multicolumn{2}{c|}{\boldmath$k=0$} & \multicolumn{2}{c|}{\boldmath$k=1$}
 & \multicolumn{2}{c|}{\boldmath$k=2$} & \multicolumn{2}{c|}{\boldmath$k=3$}
  & \multicolumn{2}{c}{\boldmath$k=4$}\\[1pt]
 \hline\hline
1&$I$   & $(1|1)$   &&&&&&&&\\
2&$I^2$ & $(2|2)$   &&&&&&&&\\
3&$I^3$ & $(4|4)$   &&&&&&&&\\
4&$I^4$ & $(8|8)$   & $D_4$             &$(4|4)$   &&&&&&\\[1mm]
5&$I^5$ & $(16|16)$ & $D_4\times I^1$     &$(8|8)$   &&&&&&\\
6&$I^6$ & $(32|32)$ & $D_4\times I^2$   &$(16|16)$ & $D_6$           & $(8|8)$&&&&\\
7&$I^7$ & $(64|64)$ & $D_4\times I^3$   &$(32|32)$ & $D_6\times I^1$   & $(16|16)$& $E_7$ & $(8|8)$&&\\
8&$I^8$ & $(128|128)$ & $D_4\times I^4$ & $(64|64)$& $D_6\times I^2$ & $(32|32)$& $E_7\times I^1$ & $(16|16)$& $E_8$ & $(8|8)$\\
 &      &             & $H_8$           & $(64|64)$& $D_4\times D_4$ & $(32|32)$ & $D_8$   & $(16|16)$       &  &    \\[1mm]
9&$I^9$ & $(256|256)$ & $D_4\times I^5$ & $(128|128)$ & $D_6\times I^3$ & $(64|64)$& $E_7\times I^2$ & $(32|32)$& $E_8\times I^1$ & $(16|16)$\\
 &      &             & $H_8\times I^1$   & $(128|128)$ & $D_4\times D_4\times I^1$ & $(64|64)$ & $D_8\times I^1$ & $(32|32)$&&\\
10&$I^{10}$ & $(512|512)$& $D_4\times I^6$& $(256|256)$ & $D_6\times I^4$ & $(128|128)$ & $E_7\times I^3$&$(64|64)$& $E_8\times I^2$&$(32|32)$\\
 &          &            & $H_8\times I^1$  & $(256|256)$ & $D_4\times D_4\times I^2$ & $(128|128)$ & $D_8\times I^2$ &$(64|64)$  & $D_{10}$ &$(32|32)$    \\
 &     &                 &                &             & $(D_4*_1\!I^6)^\dagger$   & $(128|128)$ & $D_4\times D_6$ &$(64|64)$ & &\\[1mm]\hline\noalign{\vglue2mm}
 \multicolumn{11}{l}{\parbox{172mm}{\baselineskip=11pt plus1pt $^\dagger$\,Herein, the code name ``$D_4*_1\!I^6$'' denotes that the $D_4$ code is padded by six zeroes and augmented by one additional, ``glue'' generator spanning the positions of the six added zeroes and a sufficient number (here, two) of 1's within $D_4$ so as to generate a doubly even code; it is generated by $\{1111\,0000\,00,0011\,1111\,11\}$.}}\\[-2mm]
\end{tabular}
\end{center}
\caption{Adinkra Topologies up to $N=10$: The number of nodes in an Adinkra with the indicated topology is shown to the right of each topology, in the form $(n_B|n_F)$, where $n_B$ is the number of bosons and $n_F$ is the number of fermions in the supermultiplet.}
\label{topologytable}
\end{table}

\subsection{A Multiplication of Adinkras}
\label{s:multiply}
We saw in the $N=5$ case an Adinkra topology we called $D_4\times I^1$, which had the topology of the 1-skeleton of the cartesian product of the graphs $D_4$ and $I^1$.  In this section, we will see that more generally, it is possible to take the cartesian product of any two Adinkras.

Let $A_1$ and $A_2$ be Adinkras.  We wish to define the {\em\/cartesian product\/} Adinkra $A_1\times A_2$.  If their vertex sets are $V_1$ and $V_2$, then the vertex set for $A_1\times A_2$ is $V_1\times V_2$.  A vertex $(v_1,v_2)\in V_1\times V_2$ is bosonic if $v_1$ and $v_2$ are either both bosonic or both fermionic; the vertex $(v_1,v_2)$ is fermionic if either one of $v_1,v_2$ is bosonic and the other fermionic.

If $A_1$ has $N_1$ supersymmetry generators, and $A_2$ has $N_2$ supersymmetry generators, then $A_1\times A_2$ has $N_1+N_2$ supersymmetry generators.  We number the supersymmetry generators from $A_1$ first, and then the supersymmetry generators from $A_2$, and likewise order the edge colors (ensuring that the colors for edges in $A_1$ are distinct from those in $A_2$).  Then if $(v_1,v_2)$ is a vertex in $V_1\times V_2$, we draw edges according to the edges in $A_1$ for colors in $A_1$, and edges according to edges in $A_2$ for colors in $A_2$, with the same orientations of the edges; except that when $v_1$ is fermionic, and the color is from $A_2$, then the dashedness of the edge is reversed (that is, solid edges become dashed and dashed edges become solid).\footnote{Other possibilities will work, but they are all equivalent to this one, up to replacing certain fields with their negatives.  For instance, the $N=5$ case determined when the edges would switch in dashedness based on whether $v_2$ was fermionic.}

To see that $A_1\times A_2$ is an Adinkra for a supermultiplet, we note that every vertex has exactly one edge of each color incident with it.  To see that the supersymmetry algebra\eq{eSuSy} holds, we consider two colors, and pick a starting vertex $(v_1,v_2)$.

If both colors are in $A_1$, then applying these two colors to $(v_1,v_2)$ will trace out a subgraph of $A_1\times A_2$, of the form $S\times \{v_2\}$, where $S$ is the graph traced out by those colors in the original $A_1$.  If $S$ is the sort of graph that makes\eq{eSuSy} true for $A_1$, then $S\times\{v_2\}$ will make it true for $A_1\times A_2$.  In particular it is a square with an odd number of dashed edges.

Likewise, suppose both colors are in $A_2$.  When $v_1$ is bosonic, the proof applies as before.  When $v_1$ is fermionic, we note that $S$ will have the dashedness of its edges switched; but note that the only condition that relates to dashedness is that the number of dashed edges must be odd, a feature that remains true when the dashedness of the all edges are switched.

Now suppose one color is in $A_1$ and another color is in $A_2$.  Then the graph traced out by the two colors from $(v_1,v_2)$ will be a square, like the one below.  The details of the orientation and the dashedness of the edges may vary, but any changes would affect two opposite edges in the square at the same time, and the result is always a square that makes\eq{eSuSy} true.  Thus, the resulting Adinkra satisfies \Eq{eSuSy}:
\begin{equation}
 \vC{\includegraphics{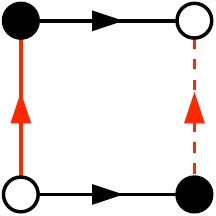}}
\end{equation}

If $A_1=I^1$, the $N=1$ base Adinkra, then the above construction for $I^1\times A_2$ is precisely what is described above in the case $N=5$ for constructing $I^1\times D_4$.  Note that $I^N$ is the $N$-fold product $I^1\times\cdots\times I^1$, so that the notation $I^N$ is appropriate.

In turn, from the point of view of codes, if $A_1$ is obtained by a code $C_1$ and $A_2$ by a code $C_2$, then $A_1\times A_2$ is obtained by the code $C_1\oplus C_2$.

\subsection{Comparison with Clifford Representations}
 \label{s:CCR}
Since Clifford representations are already classified, it is worthwhile comparing what we have just found with the known classification of Clifford representations.  The abstract algebras $\Cl(0,N{+}1)$ are known for all $N$.  The following can be found on any standard text on Clifford algebras, such as Lawson and Michelsohn's {\em Spin Geometry}\cite{rLM}.

The Clifford algebras $\Cl(0,N)$ and $\Cl(0,N{+}1)$ are given in Table~\ref{cliffordtable}.  The description of these algebras has a modulo $8$ sort of periodicity, so it is convenient to write $N=8m+s$ where $m$ and $s$ are integers and $0\le s\le 7$.  The notation $\IR(n)$, $\IC(n)$, and $\IH(n)$ denotes the algebra of $n\times n$ matrices with real, complex, and quaternionic coefficients, respectively.

\begin{table}[ht]
\begin{center}
\begin{tabular}{r|c|c|c|c}
\boldmath$s$ & \bf\boldmath Cl$(0,N)$ & \bf\boldmath Cl$(0,N{+}1)$ & \bf Irrep.
 & \bf\boldmath$\dim_\IR \text{Irrep.}$\\[1pt]\hline\hline
0&$\IR(16^m)$&$\IR(16^m)\oplus\IR(16^m)$&$\IR^{16^m}$, $\IR^{16^m}$&$16^m$\\
1&$\IR(16^m)\oplus\IR(16^m)$&$\IR(2\cdot 16^m)$&$\IR^{2\cdot 16^m}$&$2\cdot 16^m$\\
2&$\IR(2\cdot 16^m)$&$\IC(2\cdot 16^m)$&$\IC^{2\cdot 16^m}$&$4\cdot 16^m$\\
3&$\IC(2\cdot 16^m)$&$\IH(2\cdot 16^m)$&$\IH^{2\cdot 16^m}$&$8\cdot 16^m$\\[1mm]
4&$\IH(2\cdot 16^m)$&$\IH(2\cdot 16^m)\oplus\IH(2\cdot 16^m)$&$\IH^{2\cdot 16^m}$, $\IH^{2\cdot 16^m}$&$8\cdot 16^m$\\
5&$\IH(2\cdot 16^m)\oplus\IH(2\cdot 16^m)$&$\IH(4\cdot 16^m)$&$\IH^{4\cdot 16^m}$&$16\cdot 16^m$\\
6&$\IH(4\cdot 16^m)$&$\IC(8\cdot 16^m)$&$\IC^{8\cdot 16^m}$&$16\cdot 16^m$\\
7&$\IC(8\cdot 16^m)$&$\IR(16\cdot 16^m)$&$\IR^{16\cdot 16^m}$&$16\cdot 16^m$\\\hline
\end{tabular}\vspace{-3mm}
\end{center}
\caption{Clifford algebras $\Cl(0,N)$ and $\Cl(0,N{+}1)$, and their representations:
 Here $N=8m+s$ where $m$ and $s$ are integers and $0\le s\le 7$.}
\label{cliffordtable}
\end{table}

Note from Table~\ref{cliffordtable} that every Clifford algebra is either a matrix algebra, or a direct sum of two matrix algebras.  It is a classical result that every finite-dimensional real representation of such an algebra decomposes into irreducibles.  For $\IR(n)$, $\IC(n)$, and $\IH(n)$, there is up to isomorphism only one irreducible real representation: $\IR^n$, $\IC^n$, or $\IH^n$, respectively.  For $\IR(n)\oplus\IR(n)$, there are two non-isomorphic irreducible representations: one that ignores the second summand and is the standard representation on the first, and the other that ignores the first summand.  Likewise for $\IC(n)\oplus\IC(n)$ and $\IH(n)\oplus\IH(n)$.

The fourth column in Table~\ref{cliffordtable} describes the irreducible representations of $\Cl(0,N{+}1)$, based on these facts.  The last column gives the real degrees of freedom in each case, and will correspond 
to the number of vertices in the minimal Adinkra for $N$-extended supersymmetry.  Note that for $N$ a multiple of 4, there are two distinct irreducible representations of the same dimension.

The second column in Table~\ref{cliffordtable} lists $\Cl(0,N)$ because it is this representation that gives the cubical Adinkra.  Each of these must therefore be describable as a direct sum of irreducible representations of the sort listed in the fourth column.

This splitting is the focus of Table~\ref{quotienttable}.  We list in columns two and three the real dimensions of $\Cl(0,N)$ and the irreducible representation.  We divide these to see how many copies of the irreducible representation are in $\Cl(0,N)$ (when $N$ is a multiple of $4$, both types of irreducible representation are found equally often in $\Cl(0,N)$).  Since these components are obtained by repeatedly applying $\p_+$ or $\p_-$, the number of these components is also $2^k$, where $k$ is the number in the generating set for the doubly even code of maximum dimension.  This produces the last column of the table, indicating how large $k$ can be in an $[N,k]$ doubly even code.

\begin{table}[ht]
\begin{center}
\begin{tabular}{r|c|c|c|c}
\bf\boldmath$s$ & \bf\boldmath$\dim_\IR\text{Cl}(0,N)$ & \bf\boldmath$\dim_\IR \text{Irrep.}$
 &\parbox[b]{25mm}{\centering\baselineskip=10pt\bf\boldmath\#(Irrep.) in $\text{Cl}(0,N)$}
 & \bf\boldmath$\max k$\\[1pt]\hline\hline
0&$16^{2m}$&$16^m$&$16^m$&$4m$\\
1&$2\cdot 16^{2m}$&$2\cdot 16^m$&$16^m$&$4m$\\
2&$4\cdot 16^{2m}$&$4\cdot 16^m$&$16^m$&$4m$\\
3&$8\cdot 16^{2m}$&$8\cdot 16^m$&$16^m$&$4m$\\[1mm]
4&$16\cdot 16^{2m}$&$8\cdot 16^m$&$2\cdot 16^m$&$4m+1$\\
5&$32\cdot 16^{2m}$&$16\cdot 16^m$&$2\cdot 16^m$&$4m+1$\\
6&$64\cdot 16^{2m}$&$16\cdot 16^m$&$4\cdot 16^m$&$4m+2$\\
7&$128\cdot 16^{2m}$&$16\cdot 16^m$&$8\cdot 16^m$&$4m+3$\\\hline
\end{tabular}\vspace{-3mm}
\end{center}
\caption{Clifford algebras as representations, decomposed into irreducibles, and the symmetry group needed for the quotient; $N=8m+s$}
\label{quotienttable}
\end{table}

This has several consequences.  First, it tells us for every $N$ the maximum dimension, $k$, of a doubly even code of length $N$.  This agrees precisely with Gaborit's mass formula\cite{rPGMass}.  Second, it tells us for every $N$ the minimal number of real degrees of freedom in a supermultiplet, in the third column, under $\dim_\IR \mbox{irrep.}$.  Half of these are bosons; half are fermions.

But most strikingly, this table shows that there is a unique irreducible representation for the Clifford algebra when $N$ is not a multiple of 4, and that there are precisely two irreducible representations for the Clifford algebra when $N$ is a multiple of 4.  This seems strange in contrast to the multitude of codes of maximal $k$ in Table~\ref{topologytable}.  This is explained in Section~\ref{s:liq}.

\subsection{Application to 4- and Higher-Dimensional Theories}
Oxidization of $N$-supersymmetric theories from 1-dimensional time to 4 (and higher)-dimensional spacetimes necessarily requires that $N\equiv0\bmod4$, since the smallest irreducible spinor in 4-dimensional spacetime has $N=4$ components, and in higher dimensions the smallest irreducible spinors are integral multiples of this. Thus if $\cal N$ represents the number of 4D supersymmetries and $N$ represents the number of 1D supersymmetries then ${\cal N} = N/4$.  Now, listing all of these is tantamount to listing all doubly even codes---of which there is a combinatorial abundance\cite{r6-3}. So instead, we will tabulate the minimal supermultiplets, that is, the ones that correspond to the maximal codes.

We can easily read off the number of degrees of freedom from Table~\ref{quotienttable}, and for small $\cal N$, the actual topologies from Table~\ref{topologytable}.  The results are presented in Table~\ref{4DminREPsiz}.

\begin{table}[ht]
\begin{center}
\begin{tabular}{c|r|c|c|rl} 
 \boldmath$\cal N$ & \boldmath$N$ & \bf\#(Bosons) & \bf\#(Fermions) &
  \multicolumn{2}{c}{\bf \#(Adinkra Topologies)}\\[1pt]
 \hline\hline
$1$ & 4  & 4 & 4 & 1 & $D_4$  \\ 
$2$ & 8  & 8 & 8  &  1 & $E_8$ \\ 
$3$ & 12 & 64 & 64  & 2 & $D_{12}$, $E_8\times D_4$  \\ 
$4$ & 16 & 128 & 128  &  2 & $E_{16}$, $E_8\times E_8$ \\[1mm]
$5$ & 20 & 1,024 & 1,024  &  10 & $\dots^*$\\ 
$6$ & 24 & 2,048 & 2,048  & 9 & $\dots^*$\\ 
$7$ & 28 & 16,384 & 16,384  &  151 & $\dots^*$\\ 
$8$ & 32 & 32,768 & 32,768  &  85 & $\dots^*$\\[1pt] \hline\noalign{\vglue3pt}
\multicolumn{6}{l}{\parbox{125mm}{\footnotesize\baselineskip=9pt$^*$\,For ${\cal N}\geq5$ there are too many Adinkra topologies to be shown here\cite{r6-3}; see {\scriptsize\tt http://www.rlmiller.org/de\_codes/} for an up-to-date table with links to actual codes.}}
\end{tabular}\vspace{-3mm}
\end{center}
\caption{Minimal Off-Shell 4D, ${\cal N}\le 8$ Supermultiplets}
\label{4DminREPsiz}
\end{table}

Table~\ref{4DminREPsiz} can be used to make an argument about the size of the smallest irreducible off-shell representation for a given value of $\cal N$: the supersymmetry extension. If the smallest indecomposable representation coincides with the smallest irreducible representation, then for a given value of $\cal N$, the table above determines the smallest off-shell representation. For each value of $\cal N$, it is possible to consider the representation that appears at the lowest level of that column.   All of these topologies are shown together with the number of bosonic and fermionic nodes in Table~\ref{4DminREPsiz}.

For the cases of ${\cal N} = 1\text{ and }2$, the number of bosonic and fermionic degrees of freedom are in agreement with the known minimal off-shell supersymmetrical representations.  The smallest 4D, ${\cal N} = 1$ off-shell representations do indeed consist of 4 bosons and 4 fermions.  In a similar manner, the smallest 4D, ${\cal N} = 2$ off-shell representations do indeed consist of 8 bosons and 8 fermions. The case of  4D, ${\cal N} = 3$ off-shell representations is not so widely known. Nevertheless, W.~Siegel has presented an argument about the off-shell structure of conformal 4D, ${\cal N} = 3$ supergravity that indicates it describes 64 bosons and 64 fermions\cite{rWSON}.  There exist, also, one known off-shell example of a 4D, ${\cal N} = 4$ supermultiplet in Salam-Strathdee superspace. It is the conformal 4D, ${\cal N} = 4$ supergravity supermultiplet field strength\cite{rBRdW} and it consists of precisely 128 bosons and 128 fermions. All of this agrees with the first four `data' points on Table~\ref{4DminREPsiz}.

Precisely when $N$ is a multiple of $8$, the maximum value of $k$ is $N/2$, and these codes are self-dual (where the orthogonal space of the code equals the code; doubly even implies that these codes are self-orthogonal).  Thus, these relate to even unimodular lattices.  Indeed, the case ${\cal N}=2$, or $N=16$, provides the two lattices, well known to string theorists: $E_8\times E_8$, and $SO(32)$, which we call $E_{16}$, since $D_{16}$ fits in the sequence of codes $D_{2n}$ such that $D_{16}\subset E_{16}$:
\begin{equation}
 D_{16}\iff d_{16}:\left[\begin{smallmatrix}
                          1111\,0000\,0000\,0000\\[1pt]
                          0011\,1100\,0000\,0000\\[1pt]
                          0000\,1111\,0000\,0000\\[1pt]
                          0000\,0011\,1100\,0000\\[1pt]
                          0000\,0000\,1111\,0000\\[1pt]
                          0000\,0000\,0011\,1100\\[1pt]
                          0000\,0000\,0000\,1111
                         \end{smallmatrix}\right]
 \qquad\text{\it vs.}\qquad
 E_{16}\iff e_{16}:\left[\begin{smallmatrix}
                          1111\,0000\,0000\,0000\\[1pt]
                          0011\,1100\,0000\,0000\\[1pt]
                          0000\,1111\,0000\,0000\\[1pt]
                          0000\,0011\,1100\,0000\\[1pt]
                          0000\,0000\,1111\,0000\\[1pt]
                          0000\,0000\,0011\,1100\\[1pt]
                          0000\,0000\,0000\,1111\\[1pt]
                          0101\,0101\,0101\,0101
                         \end{smallmatrix}\right].
\end{equation}

It is of interest to consider the final case above, with ${\cal N}=8$, or $N=32$ supersymmetries.  For this case we have $2^{N-16}$ = $2^{16}$ = 65,536 total nodes. The Adinkra associated with this topology has 32,768 bosonic nodes and 32,768 fermionic nodes\footnote{In private communication with W.\ Siegel, we have learned that he argues that the minimal such representation has 16,384 bosonic nodes and 16,384 fermionic nodes.  It is not clear why this disagrees with our findings.}. These numbers may be familiar to anyone who has followed our development of this set of ideas.  In fact, in a previous publication\cite{rGLPR} precisely this number was given as the most likely minimum off-shell representation of 4D, ${\cal N} = 8$ supersymmetry.  The topologies for this case will include $E_8{}^4$, $E_{16}\times E_8{}^2$, $E_{16}{}^2$, and $E_{32}$, but the table above shows that there are a total of 85 distinct doubly even codes for $N=32$.  The attempt to classify these began with Conway and Pless\cite{rCP} in 1980, though a correction was found to be necessary by Conway, Pless and Sloane in 1990\cite{rCPS}.  Bilous and van Rees\cite{rBilRees} in 2007 replicated these results by performing a more systematic search and provided the list of all 85 codes on the web-site\cite{rBilReesW}.  It is amusing to note that this was achieved not long ago, that $N=32$ is the upper limit of what is known currently about self-dual codes, and that 32 is the maximal $N$ needed in our program, having in mind applications to superstrings and their $M$- and $F$-theory extensions.  Five of these $85$ codes have minimal weight 8, and it would be interesting to see if these five play a special role in ${\cal N} =8$ supersymmetries in four dimensions.

The computation of the other doubly even codes used for ${\cal N}=5,\cdots,8$, and in fact for all $4\leq N\leq32$ is currently under way and described in Ref.\cite{r6-3}.

\section{Topological Ambiguity}
\label{s:ambiguity}
Between Ref.~\cite{r6-3} and the previous section, we have shown that the chromotopologies of connected Adinkras are precisely classified by doubly even codes.  There is another issue, however: there may be circumstances where two Adinkras, with distinct chromotopologies (or indeed distinct topologies), actually define the same supermultiplet, since an Adinkra reflects not only the supermultiplet but also a choice of a basis of component fields.  This issue will be covered in the remainder of the paper.

\subsection{Three Examples of Topological Ambiguity}
\label{s:liq}
Our first example of such a situation was already discussed in Section~\ref{s:d4}.  It involves $N=4$, with splitting the 4-cubical Clifford supermultiplet into two superfields: a chiral superfield and a twisted chiral superfield.  On the one hand, we have a connected Adinkra with topology $I^4$; and on the other we have a disconnected Adinkra, consisting of two copies of the topology $D_4$.

\begin{example}
\label{decomposen4}
Begin with the Clifford Algebra Superfield described of Ref.~\cite{rA}.  This involves the following: Let $\Cl(0,4)$ act on itself by left multiplication.  We define the engineering dimensions of the bosons as $-\inv2$, and the engineering dimensions of the fermions as $0$.  To be more explicit, we take the products of various $\G_I$ in $\Cl(0,4)$ and define fields that correspond to them.  Define the boson $\f_0$ to correspond to $1$, then each $\j_I$ corresponds to each $\G_I$.  Next $\f_{12}=\G_1\G_2$, and so on, up to $\j_{123}=\G_1\G_2\G_3$, and $\f_{1234}=\G_1\G_2\G_3\G_4$.  Note that the order of the indices is numerically increasing.  We may extend this to other orders of indices, by antisymmetrization.

We then have the following transformation rules\Ft{We denote
$A_{[I}B_{J]}\Defl\inv2(A_IB_J-A_JB_I)$,
$A_{[I}B_{JK]}\Defl\inv3(A_IB_{[JK]}+A_KB_{[IJ]}+A_JB_{[KI]})$, \etc}:
\begin{subequations}
 \label{eM=}
\begin{eqnarray}
Q_I\f_0&=&\j_I,\\
Q_I\j_J&=&i\,\d_{IJ}\,\ddt\f_0 + i \,\ddt \f_{[IJ]},\\
Q_I\f_{JK}&=&2\d_{I[J}\j_{K]} + \j_{[IJK]},\\
Q_I\j_{JKL}&=&3i\,\d_{I[J}\,\ddt \f_{KL]} + i \,\ddt \f_{[IJKL]},\\
Q_I\f_{JKLM}&=&4\d_{I[J}\j_{KLM]},
\end{eqnarray}
\end{subequations}
describing the Isoscalar supermultiplet, $\sM^=_{I^4}=(\f_0,\f_{IJ},\f_{IJKL}\,|\,\j_I,\j_{IJK})$.
Its Adinkra is given here:
\begin{equation}
 \vC{
 \begin{picture}(120,45)(5,-4)
 \put(-3,3){\includegraphics[height=30mm]{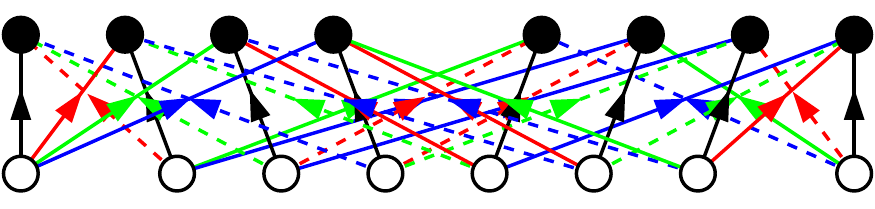}}
 \put(-2,34){$\j_1$}
 \put(13,34){$\j_2$}
 \put(28,34){$\j_3$}
 \put(44,34){$\j_4$}
 \put(-2,0){$\f_0$}
 \put(20,0){$\f_{12}$}
 \put(35,0){$\f_{13}$}
 \put(50,0){$\f_{14}$}
 \put(70,34){$\j_{123}$}
 \put(86,34){$\j_{124}$}
 \put(101,34){$\j_{134}$}
 \put(116,34){$\j_{234}$}
 \put(65,0){$\f_{23}$}
 \put(80,0){$\f_{24}$}
 \put(95,0){$\f_{34}$}
 \put(116,0){$\f_{1234}$}
 \end{picture}}
 \label{eN4B88}
\end{equation}

Now we consider $\f^{\pm}\Defl\f_0\pm \f_{1234}$.  Applying the various $Q_I$ produces $\j^\pm_I \Defl \j_I \pm \ve_I{}^{JKL}\j_{JKL}$ and $\f^\pm_{IJ}\Defl\f_{IJ}\mp \ve_{IJ}{}^{KL}\f_{KL}$.  Using these fields instead, the supermultiplet is again adinkraic, but the Adinkra is now disconnected.  We then have
\begin{eqnarray}
Q_I\f^{\pm} &=& \j^\pm_I,\\
Q_I\j^\pm_J &=& i\,\d_{IJ}\,\ddt\f^\pm + i\,\ddt \f^\pm_{[IJ]},\\
Q_I\f^\pm_{JK}&=& 2\d^{}_{I[J}\j^\pm_{K]}.
\end{eqnarray}
In this way, the above Adinkra splits into two like this:
\begin{equation}
 \vC{
 \begin{picture}(120,40)(15,-5)
 \put(1,2){\includegraphics[height=28mm]{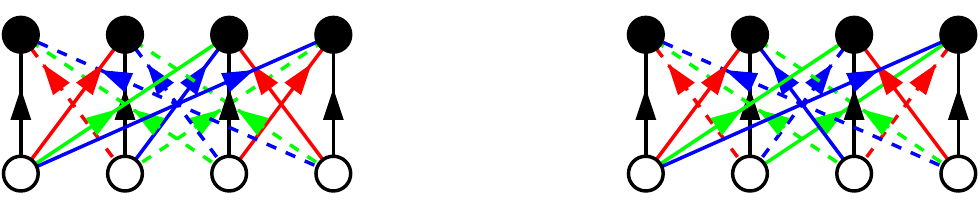}}
 \put(0,31){$\j^+_1$}
 \put(14,31){$\j^+_2$}
 \put(28,31){$\j^+_3$}
 \put(42,31){$\j^+_4$}
 \put(0,0){$\f^+_0$}
 \put(14,0){$\f^+_{12}$}
 \put(28,0){$\f^+_{13}$}
 \put(42,0){$\f^+_{23}$}
 \put(85,31){$\j^-_1$}
 \put(99,31){$\j^-_2$}
 \put(113,31){$\j^-_3$}
 \put(127,31){$\j^-_4$}
 \put(85,0){$\f^-_0$}
 \put(99,0){$\f^-_{12}$}
 \put(113,0){$\f^-_{13}$}
 \put(127,0){$\f^-_{23}$}
 \end{picture}}
 \label{eN4B44+-}
\end{equation}
These are the same two Isoscalar supermultiplets, with topology $D_4$, as depicted by the Adinkras\eq{eB44a+b}.

As we see, it is possible for a single supermultiplet to have two distinct Adinkra descriptions: in one situation, it is connected, in the other, it is not.  In the disconnected case, it is more apparent that this supermultiplet decomposes into a direct sum of two other supermultiplets:
\begin{equation}
 \underbrace{(\f_0,\f_{IJ},\f_{IJKL}\,|\,\j_I,\j_{IJK})}_{\ttt\sM^=_{I^4}}
 ~=~ \underbrace{(\f_{+0},\f_{+IJ}\,|\,\j_{+I})}_{\ttt\sM^=_{D_4{}^+}}
  ~+~ \underbrace{(\f_{-0},\f_{-IJ}\,|\,\j_{-I})}_{\ttt\sM^=_{D_4{}^-}}.
\end{equation}
We say that $\sM^=_{I^4}$ is {\em\/decomposable\/}.

This decomposition may also be seen as follows: An adinkraic supermultiplet is constructed by picking any linear combination of component fields, $\vf_0$, and completing the $Q$-orbit, \ie, computing its superpartners by applying the $Q_I$ successively to $\vf_0$: $\c_I\Defl Q_I(\vf_0\/)$, $\vf_{IJ}\Defl Q_I(\c_J)$, \etc, until further application of the $Q_I$ would produce only the already defined component fields and their time-derivatives.
 For a generic choice of $\vf_0\in\sM^+_{I^4}$, this produces $2^4=16$ linearly independent component fields and the result spans all of $\sM^=_{I^4}$. However, choosing $\vf_0\Defl(\f_0{+}\f_{1234})$ produces a complete $Q$-orbit only half as large as $\sM^=_{I^4}$; \eg, 
 $\G_1 \G_2 \G_3 \G_4$ applied to $\vf_0$ produces $\vf_0$ itself.  We thus have spanned a sub-supermultiplet, $\sM^=_{D_4{}^+}$. By starting with $\vf_0\Defl(\f_0{-}\f_{1234})$ instead, we get the complementary sub-supermultiplet, $\sM^=_{D_4{}^-}$.
\end{example}

More generally, if a valise Adinkra has chromotopology $I^N/{C_1}$, where $C_1$ is some doubly even code, and there is another doubly even code $C_2$ with $C_1\subset C_2$, we can choose a basis in which the original valise Adinkra decomposes, the pieces coming from quotienting $I^N$ by $C_2$ instead.  As to whether or not this can be done with other kinds of supermultiplets, where the component fields have more than two different engineering dimensions, depends on the node choice symmetries described in Section~\ref{s:ncs}.

\begin{construction}\label{const:split}
Let $N$ be fixed and let $C_1$ and $C_2$ be doubly even codes of length $N$, and suppose $C_1\subset C_2$.  Pick a generating set $g_1,\ldots,g_k$ for $C_2$ in such a way that $g_1,\ldots,g_j$ is a generating set for $C_1$ (see Appendix~\ref{s:linalg}).  Define
\begin{equation}
\p_1(v)=v{\cdot}\frac{1+g_1}{2}\cdots\frac{1+g_j}{2},
\end{equation}
and for each vector $\vec{x}=(x_1,\ldots,x_{k-j})\in\{0,1\}^{k-j}$, we define
\begin{equation}
\p_{2,\vec{x}}(v)=v{\cdot}\frac{1+g_1}{2}\cdots\frac{1+g_j}{2}
                    {\cdot}\frac{1+(-1)^{x_1}g_{j+1}}{2}\cdots\frac{1+(-1)^{x_{k-1}}g_k}{2}.
\end{equation}
Define $V=\img(\p_1)$, the Clifford representation defined by $C_1$.  For each vector $\vec{x}=(x_1,\ldots,x_{k-j})\in \{0,1\}^{k-j}$, we define
\begin{equation}
e_{\vec{x}} = \p_1\left(\frac{1+(-1)^{x_1}g_{j+1}}{2}\cdots\frac{1+(-1)^{x_{k-j}}g_k}{2}\right).
\end{equation}
Note that this is equal to $\p_{2,\vec{x}}(1)$.  Successively applying the various $\G_I$ to $e_{\vec{x}}$ on the left produces $2^{N-k+1}$ different vectors, which occur in $\pm$ pairs.  This spans an $(N{-}k)$-dimensional Clifford representation $\img(\p_{2,\vec{x}})$.  The corresponding Isoscalar supermultiplet will have an Adinkra with chromotopology $I^N/C_2$.

Thus, $\img(\p_1)$ splits as a direct sum of the various $\img(\p_{2,\vec{x}})$ for each $\vec{x}\in\{0,1\}^{k-j}$, and each of these summands has an Adinkra with chromotopology $I^N/{C_2}$.  In this way, we get $2^{k-j}$ connected components of the Adinkra, each corresponding to the code $C_2$, and taken as a whole, corresponding to the Clifford representation $V$.
\end{construction}

We now consider a different kind of example, where the Adinkra remains connected, but nevertheless changes its topology.

\begin{example}
Both $e_8\oplus t^2$ and $d_{10}$ are $[10,4]$ doubly even codes.  As such they both describe $N=10$ Adinkras, each with the same number of degrees of freedom ($64$ fields, 32 of which are bosons, 32 of which are fermions).  Corresponding to each Adinkra there is a valise supermultiplet. However, these two supermultiplets in fact turn out to be equivalent.

The fact that these are equivalent is an easy consequence of the fact that $\Cl(0,11)$ is isomorphic to $\IC(32)$, the algebra of $32\times 32$ complex matrices, and the only irreducible representation of this is the standard representation on $\IC^{32}$ (which has 64 real dimensions).  Thus, there is, up to isomorphism, only one $64$-dimensional Clifford representation, and the Clifford representation from $e_8\oplus t^2$ and $d_{10}$ must both be isomorphic, and therefore, the valise supermultiplets must similarly be isomorphic.

It is instructive, however, to see this more directly.  We consider the following generators:
\begin{equation}
 e_8\oplus t_2\gen
 \left[\begin{smallmatrix}
          1111\,0000\,00\\[2pt]
          0011\,1100\,00\\[2pt]
          0000\,1111\,00\\[2pt]
          1010\,1010\,00
       \end{smallmatrix}\right],\qquad
 d_{10}\gen
 \left[\begin{smallmatrix}
          1111\,0000\,00\\[2pt]
          0011\,1100\,00\\[2pt]
          0000\,1111\,00\\[2pt]
          0000\,0011\,11
       \end{smallmatrix}\right].
\end{equation}

Based on this, we define the following products of $\G$ matrices:
\begin{alignat}{3}
 c_1&=1111000000, \qquad&\iff\qquad g_1&=\G_1\G_2\G_3\G_4\\
 c_2&=0011110000, \qquad&\iff\qquad g_2&=\G_3\G_4\G_5\G_6\\
 c_3&=0000111100, \qquad&\iff\qquad g_3&=\G_5\G_6\G_7\G_8\\
 c_4&=1010101000, \qquad&\iff\qquad g_4&=\G_1\G_3\G_5\G_7\\
 c_5&=0000001111, \qquad&\iff\qquad g_5&=\G_7\G_8\G_9\G_{10}.
\end{alignat}
Each of these defines a pair of projection operators
\begin{equation}
\p_{i\pm}(v)=v{\cdot}\frac{1\pm g_i}{2}
\end{equation}
and we can define the projection operators
\begin{eqnarray}
\p_{e_8\oplus t^2} &=& \p_{1+}\circ \p_{2+}\circ\p_{3+}\circ\p_{4+}\\
\p_{d_{10}} &=& \p_{1+}\circ \p_{2+}\circ\p_{3+}\circ\p_{5+}.
\end{eqnarray}
These are both linear operators from $\Cl(0,10)$ to itself, and we define the images
\begin{eqnarray}
M_1&=&\p_{e_8\oplus t^2}(\Cl(0,10))\\
M_2&=&\p_{d_{10}}(\Cl(0,10)).
\end{eqnarray}
As before, $M_1$ and $M_2$ are Clifford representations, and the corresponding Isoscalar supermultiplets have Adinkras with chromotopologies $E_8\times I^2$ and $D_{10}$, respectively.  But it is important to recognize the role played by the starting vertex.  Namely, we take $e:=\p_{e_8\oplus t^2}(1)$ and apply the various $\G_I$ successively to the left, and in this way we can obtain 128 elements of $\Cl(0,10)$, in $\pm$ pairs.  If we arbitrarily choose one element from each $\pm$ pair, the result is a basis for $\img(\p_{e_8\oplus t^2})$.  The Isoscalar supermultiplet construction provides explicit $\IL_I$ and $\IR_I$ matrices corresponding to this basis, and thus, an explicit matrix for the $Q_I$.  The resulting Adinkra has chromotopology $E_8\times I^2$.

The trick is to find an isomorphic copy of $M_2$ inside $M_1$ by starting with a different vertex.  In this case, we take $f:=\p_{e_8\oplus t^2}(1+g_5)$ and successively apply the various $\G_I$ on the left.  The resulting fields, and thus their span, are in the image of $\p_{e_8\oplus t^2}$, and so is in $M_1$.  But now note that $g_i f=f$ for $i=1$, $2$, $3$, and $5$.  But it is no longer the case that $g_4 f= f$.  So the Adinkra for the Isoscalar supermultiplet that starts from $f$ has the topology of $D_{10}$.

It is important to note here that these are simply different bases for the same supermultiplet.  Since the Adinkra description depends on the basis, a change of basis may change the Adinkra.

Now, suppose we are given a $(32|32)$-dimensional valise supermultiplet, but that we find that we can identify it as one that can be described as follows:
\begin{enumerate}\itemsep=-3pt\vspace{-3mm}
 \item the component fields may be labeled $\f_{(0000000000)}, \dots, \j_{(1000000000)}, \dots$ with however the identifications made generated by applying successively and systematically by all $\G_I$'s on the system:
\begin{subequations}
 \label{eIds}\vspace{-2mm}
 \begin{alignat}{3}
   \f_{(1111000000)}&=\f_{(0000000000)},\\
   \f_{(0011110000)}&=\f_{(0000000000)},\\
   \f_{(0000111100)}&=\f_{(0000000000)},\\
   \f_{(1010101000)}&=\f_{(0000000000)};
\end{alignat}
\end{subequations}
 \item where the $N=10$ supersymmetry acts as follows:
\begin{eqnarray}
Q_I \f_{(x_1,\dots,x_{10})}&=&\j_{(x_1,\dots,x_{I-1},1-x_I,x_{I+1},\dots,x_{10})},\\
Q_I \j_{(x_1,\dots,x_{10})}&=&i\,\ddt\f_{(x_1,\dots,x_{I-1},1-x_I,x_{I+1},\dots,x_{10})},
\end{eqnarray}
with the identifications generated by the system\eq{eIds} understood.
\end{enumerate}
It is straightforward that the constraints\eq{eIds} may be rewritten in the form
\begin{subequations}
 \label{eConstrs}\vspace{-2mm}
 \begin{alignat}{3}
   \f_{(0000000000)}-\f_{(1111000000)} &= \p_{1-}(\f_{(0000000000)})&&=0,\\
   \f_{(0000000000)}-\f_{(0011110000)} &= \p_{2-}(\f_{(0000000000)})&&=0,\\
   \f_{(0000000000)}-\f_{(0000111100)} &= \p_{3-}(\f_{(0000000000)})&&=0,\\
   \f_{(0000000000)}-\f_{(1010101000)} &= \p_{4-}(\f_{(0000000000)})&&=0,
\end{alignat}
\end{subequations}
so that the constraints\eq{eIds} together with the result of applying all the $\G_I$'s successively on them end up spanning
\begin{equation}
 \ker\big(\p_{1-}\circ\p_{2-}\circ\p_{3-}\circ\p_{4-}\big).
\end{equation}
This then may be identified with the Adinkra given originally. On the other hand, we know that
\begin{equation}
 \ker\big( \p_{1-}\circ\p_{2-}\circ\p_{3-}\circ\p_{4-} \big)
 =\img\big( \p_{1+}\circ\p_{2+}\circ\p_{3+}\circ\p_{4+} \big)
 \Defr\img\big( \p_{e_8\oplus t_2} \big),
\end{equation}
thus identifying the original Adinkra with a supermultiplet with the $E_8\times I^2$ topology.

But, we are free to define:
\begin{alignat}{3}
\vf_{(0000000000)}&\Defl\f_{(0000000000)}+\f_{(0000001111)},\label{redef0}\\
\c_{(1000000000)}&\Defl\j_{(1000000000)}+\j_{(1000001111)}&&=\G_1(\vf_{(0000000000)}),
\label{eqn:redefine}
\end{alignat}
and so on for the remaining fields, by successive action of the $\G_I$'s; this will sometimes require a minus sign instead of a plus sign, depending on whether the number of $1$s in the last four columns is even or odd.

Since the definition\eq{redef0} is manifestly invariant under the action of $\p_{5+}(v)=v{\cdot}\frac{1+g_5}2$, it is not hard to show that the supermultiplet obtained through the change of basis generated by the successive application of all $\G_I$'s on the definition\eq{redef0} is $\img(\p_{d_{10}})$ and so in fact has the chromotopology given by $D_{10}$.

This proves that this particular supermultiplet admits two distinct but equivalent bases of component fields: one depicted by an Adinkra of the chromotopology given by $E_8\times I^2$, and the other depicted by an Adinkra of the chromotopology given by $D_{10}$.
\end{example}

Thus, the same supermultiplet may indeed have two different Adinkra descriptions, each of a different topology. One way to view this is to consider that an Adinkra is more than a description of a supermultiplet: it defines a supermultiplet together with a special basis for the set of component fields in this supermultiplet.

In particular, it considers a single node and successively applies the various $Q_I$ to it (and taking integrals if going against the orientation).  But what if a different vector for a starting node were taken?  It is possible we would end up with the same Adinkra, but it could be different.

\begin{example}
To consider a simpler example, consider the $N=5$ codes $C_1=d_4\oplus t_1$ generated by $11110$, and $C_2$ generated by $11101$.  The corresponding Adinkras $A_1$ and $A_2$ are given in Figure~\ref{fig:twoN5}.
\begin{figure}[htb]
\begin{center}
\begin{picture}(170,45)(0,0)
\put(0,0){\includegraphics{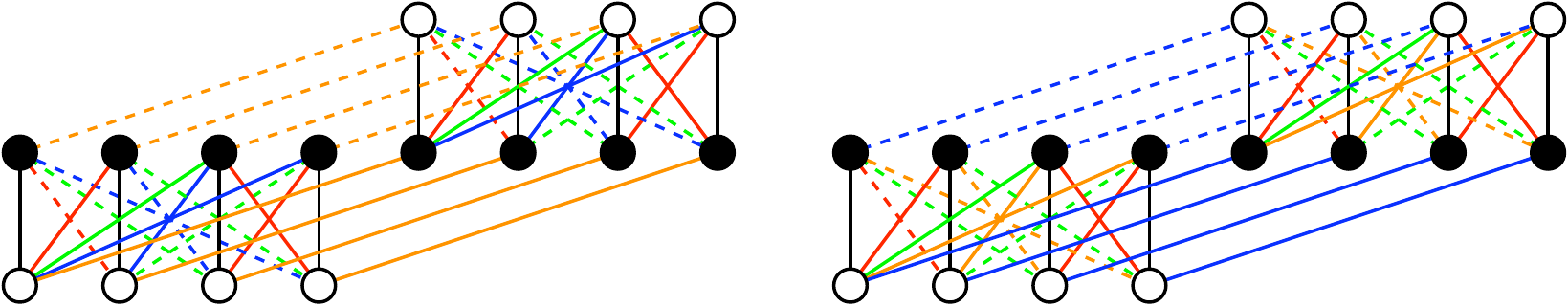}}
\put(2,-3){\makebox[0in][c]{$\f_0$}}
\put(2,19){\makebox[0in][c]{$\j_1$}}
\put(13,19){\makebox[0in][c]{$\j_2$}}
\put(24,19){\makebox[0in][c]{$\j_3$}}
\put(35,19){\makebox[0in][c]{$\j_4$}}
\put(13,-3){\makebox[0in][c]{$\f_{12}$}}
\put(24,-3){\makebox[0in][c]{$\f_{13}$}}
\put(35,-3){\makebox[0in][c]{$\f_{14}$}}
\put(46,12){\makebox[0in][c]{$\j_5$}}
\put(57,12){\makebox[0in][c]{$\j_{125}$}}
\put(68,12){\makebox[0in][c]{$\j_{135}$}}
\put(79,12){\makebox[0in][c]{$\j_{145}$}}
\put(46,34){\makebox[0in][c]{$\f_{15}$}}
\put(57,34){\makebox[0in][c]{$\f_{25}$}}
\put(68,34){\makebox[0in][c]{$\f_{35}$}}
\put(79,34){\makebox[0in][c]{$\f_{45}$}}
\put(90,-3){\makebox[0in][c]{$\vf_0$}}
\put(90,19){\makebox[0in][c]{$\c_1$}}
\put(101,19){\makebox[0in][c]{$\c_2$}}
\put(112,19){\makebox[0in][c]{$\c_3$}}
\put(123,19){\makebox[0in][c]{$\c_5$}}
\put(101,-3){\makebox[0in][c]{$\vf_{12}$}}
\put(112,-3){\makebox[0in][c]{$\vf_{13}$}}
\put(123,-3){\makebox[0in][c]{$\vf_{15}$}}
\put(134,12){\makebox[0in][c]{$\c_4$}}
\put(145,12){\makebox[0in][c]{$\c_{124}$}}
\put(156,12){\makebox[0in][c]{$\c_{134}$}}
\put(167,12){\makebox[0in][c]{$\c_{154}$}}
\put(134,34){\makebox[0in][c]{$\vf_{14}$}}
\put(145,34){\makebox[0in][c]{$\vf_{24}$}}
\put(156,34){\makebox[0in][c]{$\vf_{34}$}}
\put(167,34){\makebox[0in][c]{$\vf_{45}$}}
\end{picture}
\end{center}
\caption{The Adinkras $A_1$ (left) and $A_2$ (right):  They have the same topology but not the same chromotopology---following the colors for 1, 2, 3, and 4 on the left (black, red, green, blue) brings you back to where you started, but not on the right.  Instead, the Adinkra on the right requires you to follow colors for 1, 2, 3, and 5 (black, red, green, orange) to get back where you started.}
\label{fig:twoN5}
\end{figure}
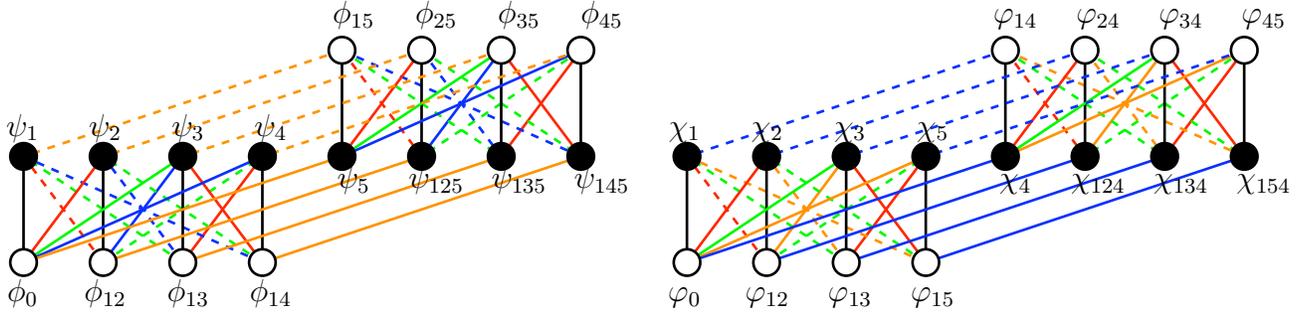

We define
\begin{eqnarray}
g_1&=&\G_1\G_2\G_3\G_4\\
g_2&=&\G_1\G_2\G_3\G_5
\end{eqnarray}
and as before, $\p_1(v)=v{\cdot}\frac{1+g_1}{2}$ and $\p_2(v)=v{\cdot}\frac{1+g_2}{2}$.  The Adinkra $A_1$ consists of the node $\f_0:=\p_1(1)$ and the results of applying the various $\G_I$ to the left: define $\j_i:=\G_I \f_0$ for $I=1, \ldots, 5$; then $\f_{1I}:=\G_I \j_I$ for $I=2, 3, 4, 5$; then $\f_{I5}:=\G_I \j_5$ for $I=2, 3, 4$; and $\f_{1I5}:=\G_1 \f_{I5}$ for $I=2, 3, 4$.  Note that $\G_1\G_2\G_3\G_4 \f_0=\f_0$.

Now we define $\vf_0:=\f_0+\f_{45}=(\Ione{+}\G_1\G_2\G_3\G_5)\f_0=\p_1(\p_2(1))$.  Then define $\c_I:=\G_I \vf_0$ for $I=1,\ldots,5$; then $\vf_{1I}:=\G_1 \c_I$ for $I=2,3,4,5$; then $\vf_{I4}:=\G_I\c_4$ for $I=2,3,5$; and $\c_{1I4}:=\G_1 \vf_{I4}$ for $I=2,3,5$.  The result will be the Adinkra $A_2$.  Note that the span of the $\vf$'s and $\c$'s are in the image of $\p_1$, and indeed, as Clifford representations, are equal. Having chosen two different bases, the two Adinkras are different: in particular, they have different chromotopologies, but the same topology.

That is, this component field basis change in fact implements a permutation in the supersymmetry generators.
\end{example}

\subsection{General Topology Ambiguity}
More generally, the large number of codes for high $N$ and the relatively small number of Clifford representations suggests that this example is fairly representative of valise supermultiplets. Indeed:
\begin{construction}\label{const:liq}
If $C$ and $D$ are both doubly even codes of length $N$ and dimension $k$, and if $M$ is a Clifford representation whose Isoscalar supermultiplet has an Adinkra with chromotopology given by quotienting $I^N$ by $C$, then it is possible to choose a field redefinition in $M$ that changes the Adinkra.  If $C$ and $D$ are maximal doubly even codes of length $N$, then this field redefinition will describe $M$ as an Adinkraic Clifford representation whose chromotopology is given by quotienting $I^N$ by $D$.

The main idea is as before, with one complication: whenever the codes $C$ and $D$ have non-trivial intersection, there is a possibility that one would want to quotient by $\p_{i+}$ and the other by $\p_{i-}$.  The result would be the zero field, and indeed, by successively applying $\G_I$s, we will only get the zero supermultiplet.  To prevent this, we need to coordinate the codes.

To do this, we choose the generators for the codes so that a generator set for $C\cap D$ is used as a starting point for both.  Then we can agree on the use of $\p_{i+}$ or $\p_{i-}$ for these shared generators.

Let $m$ be the dimension of $C\cap D$.  Using Proposition~\ref{prop:codeadd}, we find a generating set $\{c_1,\ldots,c_k\}$ for $C$ and $\{d_1,\ldots,d_k\}$ for $D$ so that $c_i=d_i$ for $1\le i \le m$, and $\{c_1,\ldots,c_m\}$ is a generating set for $C\cap D$.  Proposition~\ref{prop:codeadd} also says that $\{c_1,\ldots,c_k,d_{m+1},\ldots,d_k\}$ is a generating set for a binary linear code $C+D$.  This code is of length $N$ and dimension $2k-m$.  Note that $C+D$ is not generally doubly even, but since the bitwise sum of any two even codewords is even (an easy consequence of Proposition~\ref{prop:weights}), it follows that all codewords in $C+D$ are even.  As before, for each $c_i$ and $d_i$ we define the corresponding product of $\G_I$ matrices $g_i$ and $h_i$, respectively.

Define
\begin{eqnarray}
\p_1(v)&=&v{\cdot}\frac{1+g_1}{2}\cdots\frac{1+g_k}{2},\label{eqn:pivg}\\
\p_2(v)&=&v{\cdot}\frac{1+h_1}{2}\cdots\frac{1+h_k}{2}.
\end{eqnarray}
As before, the image of $\p_1$ is a Clifford representation, and if we take as basis $e:=\p_1(1)$ and the various $\G_I$ successively applied to the left of $e_0$, this will give a basis in which the Isoscalar supermultiplet has Adinkra with chromotopology given by $C$.  If we want this Clifford representation to be isomorphic to $M$, we might have to alter some of the signs in\eq{eqn:pivg}, or equivalently, redefine the $g_i$ by arranging its $\G_I$ matrices in a different order.

Define $\vf_0:=\p_1(\p_2(1))$.  If we successively apply $\G_I$ to the left of $\vf_0$ these will also be in the image of $\p_1$, and yet if we take a product of $\G_I$ indicated by a codeword in $D$, say, $\G_\omega$, then $\G_\omega \vf_0 = \p_1(\p_2(\G_\omega))=\p_1(\pm \G_\omega)=\pm \G_\omega \vf_0$.

We now check that $\vf_0$ is not zero.  Otherwise, the resulting Clifford representation would be trivial.  We first write
\begin{equation}
\vf_0=\p_1(\p_2(1))=\frac{1+h_1}{2}\cdots\frac{1+h_k}{2}{\cdot}\frac{1+g_1}{2}\cdots\frac{1+g_k}{2}
\end{equation}
and anticommute the factors involving $h_1,\ldots,h_m$ until they are next to the $h_1,\ldots,h_m$ factors, and use the fact that the square of any of these factors is itself, to get
\begin{equation}
\vf_0=\frac{1+h_{m+1}}{2}\cdots\frac{1+h_k}{2}{\cdot}\frac{1+g_1}{2}\cdots\frac{1+g_k}{2}.
\end{equation}
As in Construction~\ref{const:quotient}, we can multiply out this product, and get $2^{2k-m}$ terms, one for each codeword in $C+D$ (but written as a product of $\G_I$, perhaps with a minus sign).  The fact that all of these terms are different follows from the fact that $\{c_1,\ldots,c_k,d_{m+1},\ldots,d_k\}$ is a linearly independent set.  As a consequence of this, these terms do not cancel, and thus, $\vf_0$ is not zero.

Now if $C$ and $D$ are maximal doubly even codes of length $N$, then $M$ is irreducible, and so the Clifford representation obtained by successively applying $\G_I$ to $\vf_0$ on the left must be $M$, since it is not $\{0\}$.
\end{construction}

\subsection{Impact on Classification}
The importance of these examples and constructions~\ref{const:quotient}, \ref{const:split} and~\ref{const:liq} is that for every $N$, one can decompose the Clifford representation---and so also the valise Adinkras---into irreducibles, each of which has the chromotopology of the quotient of $I^N$ by a maximal doubly even code for the given $N$.
 It is irrelevant which maximal code was used, since they can all be related to each other.  This helps explain the paucity of irreducible Clifford representations in light of the abundance of doubly even codes.  The cases where there are two irreducible Clifford representations (when $N$ is a multiple of 4) is not due to the multiplicity of doubly even codes, but from the choices of $+$ or $-$ in the projections $\p_{i\pm}$.  These correspond to choices as to which of the edges will be dashed.

We note that Construction~\ref{const:liq} is made possible by the fact that, in a valise supermultiplet, all component fields of the same statistics have the same engineering dimension. This permits us to make linear combinations (of component fields of the same statistic, naturally) at will, including the linear combinations necessary to transform the valise supermultiplet of a given chromotopology and topology into a supermultiplet of a different chromotopology---or even different topology. This is true of all valise supermultiplets, both Isoscalar and Isospinor supermultiplets.

Supermultiplets in which the fields of the same statistics do not all have the same engineering dimension may be mapped to valise supermultiplets by vertex raises/lowerings\cite{r6-1}, called originally ``automorphic duality''\cite{rGR2}. It is then possible to inverse-map the Construction~\ref{const:liq} to those non-valise supermultiplets. However, this will in general result in non-local linear combinations of the component fields of the non-valise supermultiplets, thus restricting strongly the applicability of Construction~\ref{const:liq} to non-valise supermultiplets. We now turn to this issue.

\subsection{Supermultiplets With Only One Adinkraic Description}
 \label{s:onehook}
The previous section showed that the combinatorial multitude of doubly even codes, as presented in Ref.\cite{r6-1}, is irrelevant for the classification problem of valise supermultiplets: for any given $N$, we should take the maximal $k$ to get an irreducible representation, and any maximal doubly even code will give the same representation as another---up to, perhaps, two different equivalence classes of dashing the edges.  But for non-valise supermultiplets, the wide array of doubly even codes is still relevant, since the various field redefinitions to convert between supermultiplets for different codes no longer need be applicable---they may miss degrees of freedom or involve nonlocal transformations such as integrals.

The extreme example is when the Adinkra is {\em one-hooked}\cite{r6-1}, that is, there is only one vertex $v$ of lowest engineering dimension, and that this is the unique vertex having all its adjacent edges oriented away from it.  This image---but not the eventual conclusion---differs from Ref.~\cite{r6-1} in that we now imagine all the other vertices as floating upward from this hooked vertex, instead of hanging downward.

Suppose we are given a one-hooked Adinkra with all remaining vertices floating upward, and suppose that there is another Adinkra for this supermultiplet.  It must include one vertex at this lowest engineering dimension, and from it one must be able to acquire all the other degrees of freedom in the supermultiplet. Being one-hooked, there is only one such vertex, $v$.  It cannot involve derivatives of other fields, since a derivative only increases the engineering dimension.  It could be that a non-zero scalar multiple of $v$ is actually used, but this does not affect the situation, since if we multiply all component field variables in a supermultiplet by a non-zero constant scalar, the resulting Adinkra looks exactly the same as before.  So without loss of generality, our new Adinkra includes $v$.  We apply all the various $Q_I$ to $v$, and eventually we reach all the other vertices in the Adinkra.  So the only adinkraic choice of variables is the one we started with.

In the one-hooked case, therefore, no field redefinitions are possible, and so there is only one Adinkra possible, with only one Adinkra topology, and with only one doubly even code capable of describing it.  Therefore, all doubly even codes are necessary in describing Adinkra topologies and the supermultiplets depicted by one-hooked Adinkras are inequivalent. There are no redundancies.

\begin{example}\label{ExCli4d}
Here is an example of a hypercube $I^4$ topology that does not split into two $D_4$ supermultiplets:
\begin{equation}
  \sM^\diamond_{I^4}~:~
  \vC{\begin{picture}(100,90)(-10,-5)
  \put(-2,-2){\includegraphics[height=74mm]{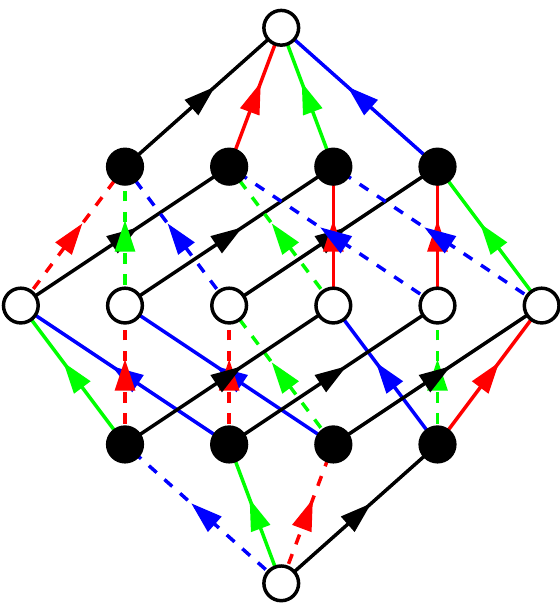}}
  \put(35,-3){$y$}
  \put(8,13){$\h_1$}
  \put(21,13){$\h_2$}
  \put(40,13){$\h_3$}
  \put(54,13){$\h_4$}
  \put(-9,34){$Y_{12}$}
  \put(4.5,34){$Y_{13}$}
  \put(17.5,34){$Y_{14}$}
  \put(42.5,34){$Y_{23}$}
  \put(56,34){$Y_{24}$}
  \put(69.5,34){$Y_{34}$}
  \put(7,55.5){$\Y_{123}$}
  \put(20,55.5){$\Y_{124}$}
  \put(40,55){$\Y_{134}$}
  \put(53,55){$\Y_{234}$}
  \put(35,72){$\cY_{1234}$}
 \end{picture}} \label{eB14641}
\end{equation}
The lowest node here, $y$, is the only degree of freedom with its engineering dimension (which we here declare to be $-\inv2$).  Since time derivatives increase engineering dimension by one, no number of time derivatives can cause another field to have engineering dimension $-\inv2$.  Any decomposition of this supermultiplet would have to have one of its factors with some degree of freedom with engineering dimension equal to $-\inv2$, and the only possibility is a scalar multiple of $y$.  Once this is obtained, the fact that the Adinkra is connected requires that the other fields in this supermultiplet must be in this factor, and thus, the decomposition is trivial.  Hence, this supermultiplet does not decompose as a direct sum of two $D_4$ supermultiplets.

In turn, it is not hard to use the vertex-raising map to relate the $\sM^\diamond_{I^4}$, depicted by the Adinkra\eq{eB14641}, and $\sM^=_{I^4}$, depicted by the Adinkra\eq{eN4B88}:
\begin{equation}
\begin{gathered}
 y=\f_0,\qquad Y_{IJ}=\dot\f_{IJ},\qquad \cY_{1234}=\ddot\f_{1234},\\
 \h_I=\j_I,\qquad \Y_{IJK} = \dot\j_{IJK}.
\end{gathered}
 \label{eN4V->1h}
\end{equation}
Using this to translate the change of basis that led to decomposing this $(8|8)$-dimensional supermultiplet into two $(4|4)$-dimensional ones\eq{eN4B44+-} would however lead to non-local expressions such as:
\begin{subequations}
 \label{eNL}
\begin{alignat}{5}
 \f^\pm &\Defl \f_0\mp\f_{1234} \qquad&&\iff&\qquad
  y^\pm &\Defl y \mp \ddt^{-2}\cY_{1234},\\
 \j^\pm_I &\Defl \j_I\pm\frac{1}{3!}\ve_I{}^{JKL}\j_{JKL} \qquad&&\iff&\qquad
  \h^\pm_I &\Defl \h_I \pm \frac{1}{3!}\ve_I{}^{JKL}\ddt^{-1}\Y_{JKL},\\
 \f^\pm_{IJ} &\Defl \f_{IJ}\pm\frac{1}{2!}\ve_{IJ}{}^{KL}\f_{KL} \qquad&&\iff&\qquad
   Y^\pm_{IJ} &\Defl Y_{IJ}\pm\frac{1}{2!}\ve_{IJ}{}^{KL}Y_{KL}.
\end{alignat}
\end{subequations}
Thus, the supermultiplet depicted by the Adinkra\eq{eB14641} does not decompose\Ft{Non-local field redefinitions\eq{eNL} are typically permissible neither in quantum mechanics nor in field theory.} into two $D_4$ supermultiplets like the valise supermultiplet does; however, it is reducible: see also appendix~\ref{s:red} and Ref.\cite{r6-4.2}.
\end{example}

\begin{example}
We next one-hook the $N=10$ cases $E_8\times I^2$ and $D_{10}$, and let the remaining nodes float upward from it:
\begin{align}
 E_8\times I^2&:\vC{\includegraphics[width=140mm]{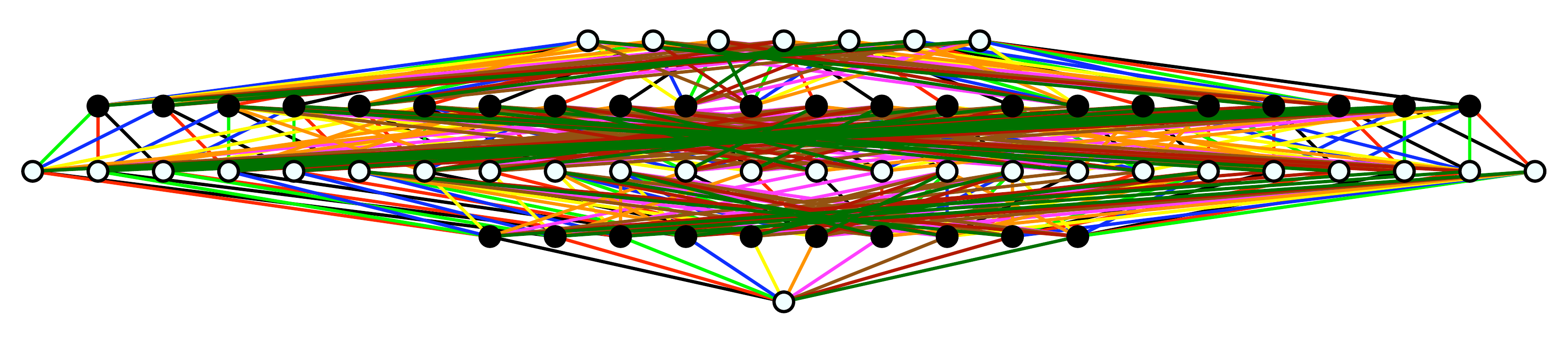}}\\
 D_{10}       &:\vC{\includegraphics[width=140mm]{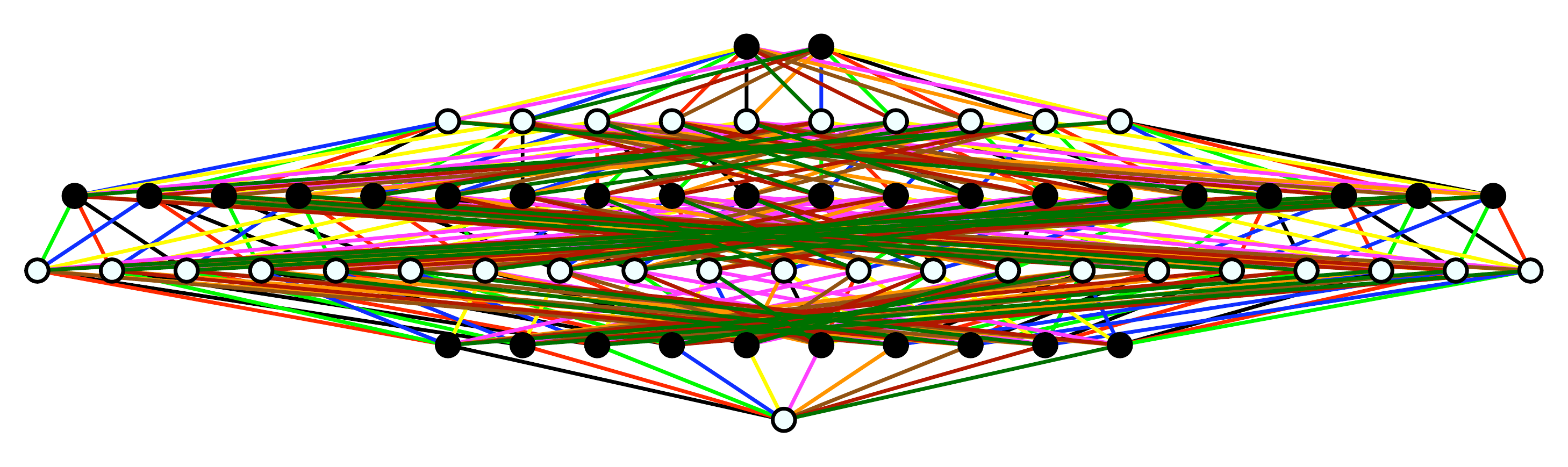}}
\end{align}
Since these are one-hooked, they cannot be isomorphic: any isomorphism would have to take the node of lowest engineering dimension of the $E_8\times I^2$ Adinkra on the left to the corresponding lowest node of the $D_{10}$ Adinkra on the right.  Then on the left, $Q_1Q_3Q_5Q_7$ takes us back to the lowest node, while on the right, it does not.  Therefore, these cannot be isomorphic.

In fact, the distinction is obvious even by just considering the number of vertices of each engineering dimension: it differs between these examples. The $E_8\times I^2$ supermultiplet is $(1|10|24|22|7)$-dimensional, whereas the $D_{10}$ one is $(1|10|21|20|10|2)$-dimensional.  Thus, these cannot be isomorphic.
\end{example}

\begin{example}
Similarly, we one-hook the $N=16$ cases $E_8\times E_8$ and $E_{16}$ in the same buoyant fashion:
\begin{align}
 E_8\times E_8&:\vC{\includegraphics[width=140mm]{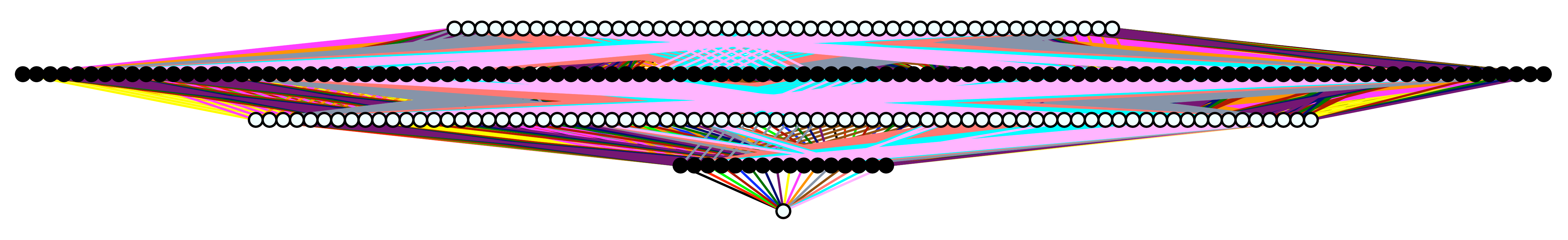}}\\
 E_{16}       &:\vC{\includegraphics[width=140mm]{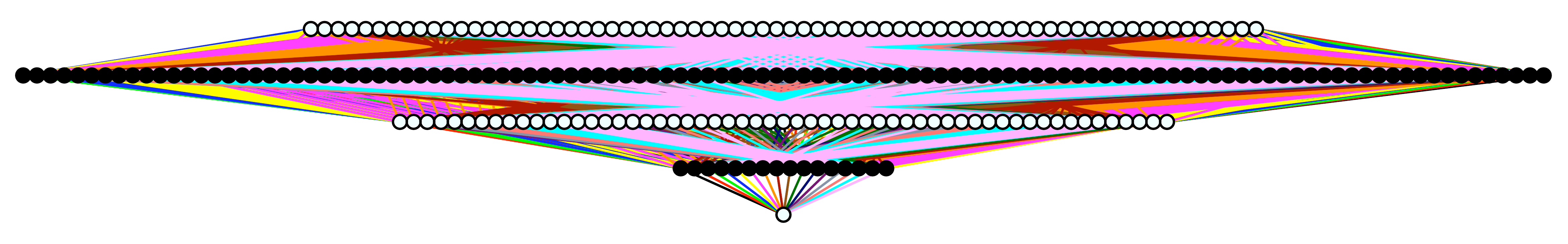}}
\end{align}
Again, besides the general argument, the simple degree-of-freedom counting shows that these cannot be isomorphic: The $E_8\times E_8$ supermultiplet is $(1|16|78|112|49)$-dimensional, whereas the $E_{16}$ one is $(1|16|57|112|70)$-dimensional. In turn, the valise supermultiplets corresponding to distinct valise Adinkras with these topologies are in fact isomorphic.
\end{example}

\subsection{Node Choice Symmetries}
\label{s:ncs}
The use of Construction~\ref{const:liq} to obtain an isomorphism between two non-valise supermultiplets is obstructed by the fact that it is no longer possible to make arbitrary linear combinations of component fields of the same statistics: Now that some of the fields have different engineering dimensions, it will be necessary to either insert time derivatives (in which case the constant term of a field is lost) or integrals (in which case the field redefinition is non-local).

More generally, this provides a criterion by which we can determine whether two Adinkras are depicting the same supermultiplet.  We can look for field redefinitions of the type described in Construction~\ref{const:liq} that are suggested by the two codes, and see if the resulting engineering dimensions agree.  If they do not agree and inverse powers of $\ddt$ are needed, then such a field redefinition is non-local and so not allowed.  If they do agree, then Construction~\ref{const:liq} provides a field redefinition that relates the two supermultiplets to reveal that they are isomorphic, provided that the dashing of edges also agrees.

Suppose a connected Adinkra is given, and pick a starting node, which we call $v_0$.  If we wish to find other Adinkras describing the same supermultiplet, such Adinkras are obtained by taking a starting node $v_1$ for the new Adinkra as a linear combination of the nodes of the original Adinkra.  As just discussed, in order to avoid losing degrees of freedom, the nodes involved in this linear combination must all be of the same engineering dimension.  We successively apply the various $Q_I$ to $v_1$ on the left, taking integrals when the resulting field is a derivative of a field in the supermultiplet.  In so doing, by the linearity of $Q_I$, we can study the individual component nodes of $v_1$ and apply $Q_I$ to each.  Thus, we need to determine, if a certain sequence of $Q_I$ sends the original starting node $v_0$ to $v_1$, how the other nodes are permuted.

Suppose $v_0$ and $v_1$ are connected by an edge of color $I$.  The resulting Adinkra taking $v_1$ as starting node will swap all vertices that are connected by an edge of color $I$.  More generally, assuming the Adinkra is connected, if we make any other vertex in the Adinkra the starting node, then there is a path of edges connecting $v_0$ to that vertex.  And the result involves composing the permutations obtained by considering starting node changes for each edge in the path.

In Ref.\cite{r6-3} we defined a family of maps $q_1,\ldots,q_N$ on the set of vertices that correspond to the $Q_1,\ldots,Q_N$, so that if $v$ and $w$ are vertices connected by an edge of color $I$, then $q_I(v)=w$ and $q_I(w)=v$.  When the Adinkra has the chromotopology of the $N$-cube $[0,1]^N$, $q_I$ is the reflection along the $I^\text{th}$ axis, \ie, across the hyperplane at $x_I=\inv2$.  The set of $q_I$ generates a group, and for all $I$, $q_I^2$ is the identity and the $q_I$ commute with one another.

For each $\vec{x}=(x_1,\ldots,x_N)\in(\ZZ_2)^N$, we define $q_{\vec{x}}=q_1^{x_1}\circ\cdots\circ q_N^{x_N}$.  In Ref.\cite{r6-3}, we defined a code $C$ to be the set of $\vec{x}$ so that $q_{\vec{x}}(v)=v$ for all vertices $v$, and saw that the chromotopology of an Adinkra was $I^N/C$. We are now in the position to specify:

\begin{defn}\label{d:NCG}
Given an Adinkra of the chromotopology $I^N/C$, where $[v]$ denotes the engineering dimension of the component field corresponding to the vertex $v$, and given the group $(\ZZ_2)^N$ of reflections $q_I:x_I\to(1{-}x_I)$ and writing $q_{\vec{x}}=q_1^{x_1}\circ\cdots\circ q_N^{x_N}$, the {\em\bfseries\/node choice group (NCG)\/} $H\subset(\ZZ_2)^N$ consists of all elements such that $[q_{\vec{x}}(v)]=[v]$. Such (nontrivial) elements are called {\em\/node choice symmetries\/}.
\end{defn}
Clearly, $C$ is a (doubly even) subgroup of $H$.

Since no boson can have the same engineering dimension as a fermion (for a connected Adinkra), node choice symmetries must have even weight.  Thus, the node choice group must be an even code.  For instance, a one-hooked Adinkra with the topology $I^N/C$ has $H=C$, and a valise Adinkra with that same topology has $H$ equal to the set of all even-weight codewords in $(\ZZ_2)^N$. Of course, the action of $H=C$ on the one-hooked $I^N/C$ Adinkra is trivial, and the action of $H\subset(\ZZ_2)^N$ on an $I^N/C$ valise Adinkra is equivalent to that of a freely acting $(\ZZ_2)^{N-k-1}$ group, where $k=\dim(C)$. Raising/lowering nodes has a very nontrivial effect on the node choice group, and depends not only on the number of nodes at the various levels of engineering dimension, but also on their connectedness. For example, the following three Adinkras all have the same number of nodes at various levels:
\begin{equation}
  \vC{\includegraphics[width=1in]{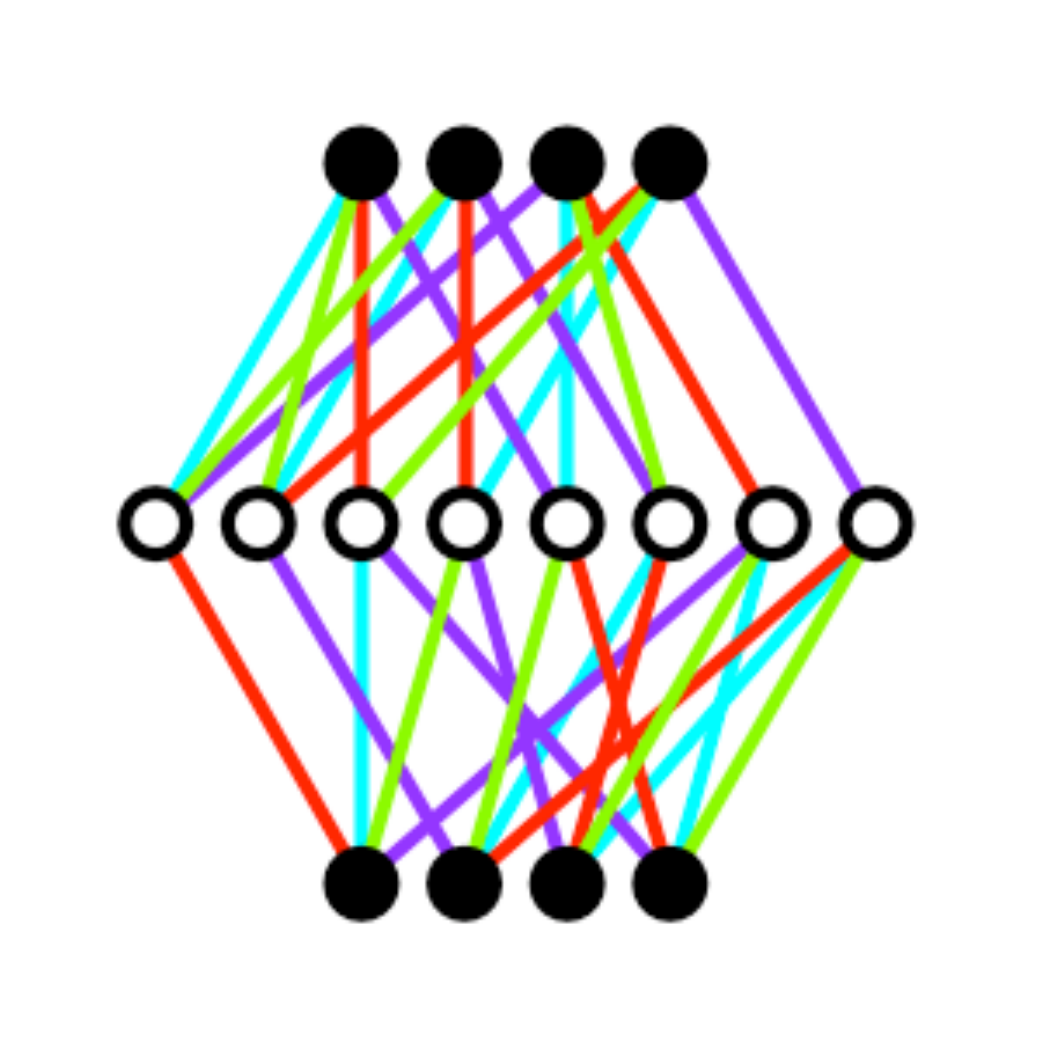}}
   \qquad
  \vC{\includegraphics[width=1in]{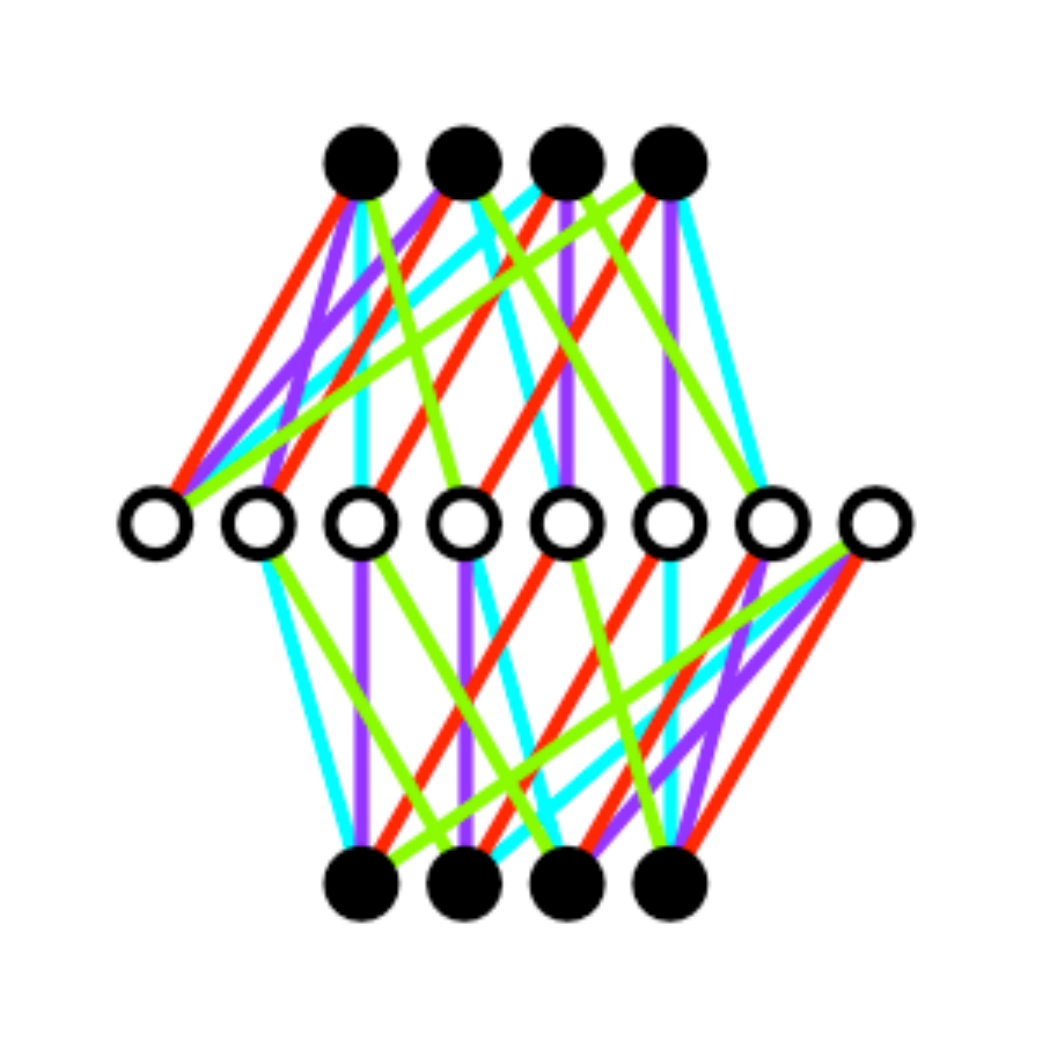}}
   \qquad
  \vC{\includegraphics[width=1in]{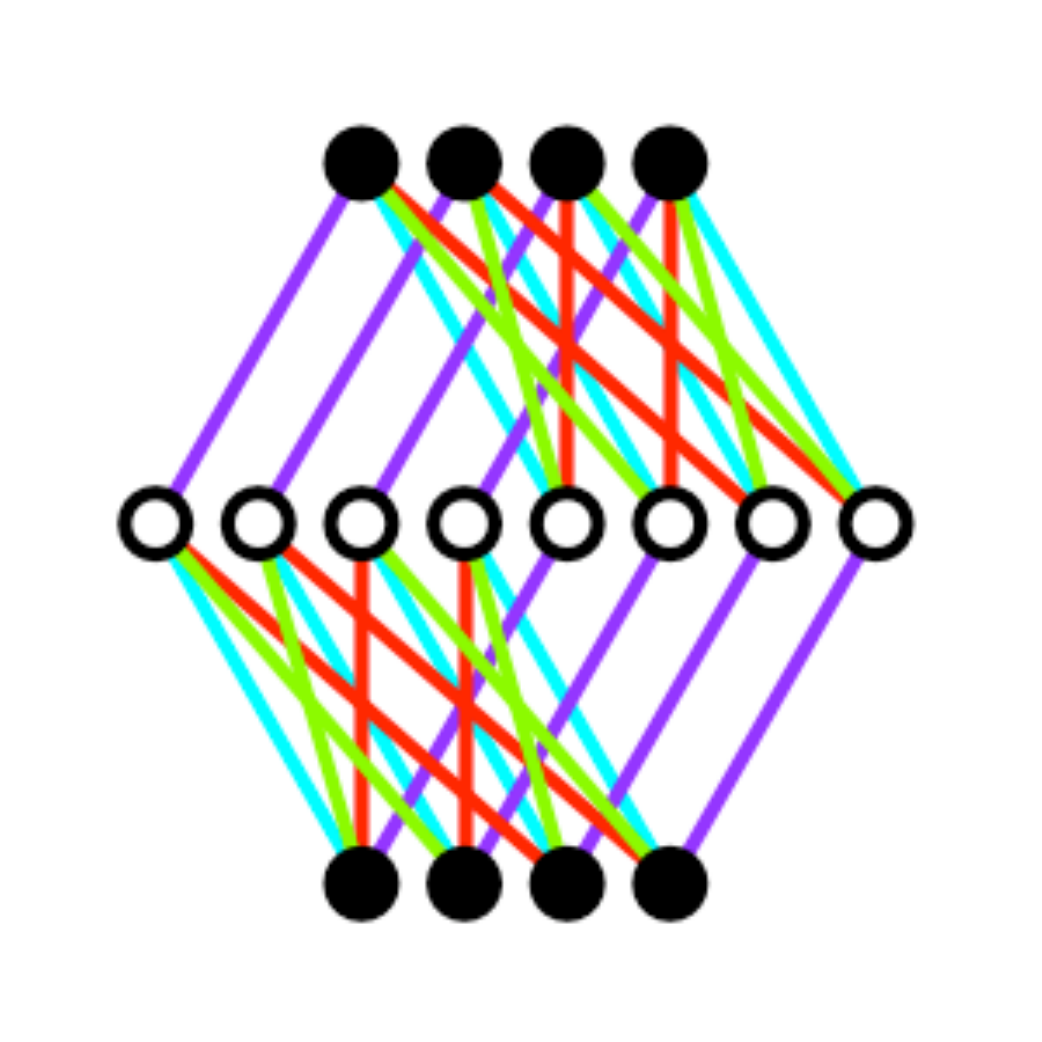}}
\end{equation}
but only the last one has a nontrivial node choice symmetry, as the; see Table~\ref{ncstable}. As with Feynman diagrams, the appearance of the Adinkras seems highly suggestive of the symmetry properties of the supermultiplets that they depict, even on a first glance.

\subsection{Removing the Topological Ambiguity}
Now we are ready to explain how we can determine when an Adinkra has an alternate description using a different code.  Suppose we have an Adinkra with the chromotopology given by $I^N/C_1$, where $C_1$ is an $[N,k]$ doubly even code.  If the node choice group is $H$, then $C_1\subset H$.  If there is a larger doubly even code $C_2$ with $C_1\subsetneq C_2\subset H$, then we can use Construction~\ref{const:split} to decompose the supermultiplet into a direct sum of smaller supermultiplets.  So it suffices to consider maximal doubly even codes inside $H$.  By Theorem~\ref{thm:maxcode} in the Appendix, all maximal doubly even codes inside $H$ have the same dimension.  We can choose any one of these, and they will be equivalent, by Construction~\ref{const:liq}.  Conversely:
\begin{corl}
If $C_3$ is any doubly even code not contained in $H$, then no field redefinitions can change the Adinkra with the chromotopology given by  $I^N/C_2$ into an Adinkra with the chromotopology given by $I^N/C_3$.
\end{corl}

One approach to the classification of Adinkras, then, would be for any $N$ and $k$, to find the even codes and the maximal doubly even codes that they contain.  Then quotient the cube by the chosen node choice group, and find all the ways to hang it.  This object is not an Adinkra, but it specifies how to hang the original $N$-cube at those heights. Now quotient this by the doubly even codes to obtain the actual quotients. Recall that as long as these doubly even codes are maximal in the node choice group, these quotients are all equivalent, so only one is needed. Finally, select a choice of dashing of the edges.

For the valise Adinkras, the node choice group is the set of all even codewords: as intuitively clear, in valise supermultiplets, all component fields have the same engineering dimension so that they all can mix. The valise supermultiplets clearly have the maximal node choice symmetry. So, to choose a maximal doubly even code in the node choice group is simply to choose a maximal doubly even code.  For one-hooked Adinkras, the node choice group is equal to the given doubly even code, so the maximal doubly even code contained in it is the original doubly even code, and there is no possible field redefinition. Tables~\ref{redtable} and~\ref{ncstable} provide also all the ``intermediate'' supermultiplets.

\section{Conclusions}
 \label{s:C}
Let us summarize where this leaves us concerning the classification of Adinkras, and the classification of $D=1$ $N$-extended supermultiplets.

An Adinkra is determined by its chromotopology, a hanging of the vertices, and a choice of which edges are dashed.  The results of this paper and of Ref.~\cite{r6-3} show that chromotopologies are equivalent to doubly-even codes.  In fact, this paper presents a method (construction~\ref{const:quotient}) to turn a doubly-even code into a chromotopology. Ref.~\cite{r6-1} then describes how to determine the various ways of hanging the vertices of the Adinkra.  The question of how to choose which edges are dashed turns out to involve cohomology and will be discussed in a separate effort.

However, this does not yet classify $D=1$ $N$-extended supermultiplets for several reasons: first, not all supermultiplets have an adinkraic description.  Second, some of these Adinkras may describe the same supermultiplet.  This second issue can be addressed through the node choice group.

The approach, then, should be modified by first choosing an even code $H$ to be a node choice group.  To determine irreducible supermultiplets, we identify a doubly even subcode $C\subset H$ of maximal dimension.  Various possible maximal doubly even subcodes of $H$ will all result in the same supermultiplet, by the observations of the previous section, so only one such choice need be considered.

We can quotient $I^N$ by $H$, and the various ways of hanging the resulting diagram will correspond to the various ways of hanging the $I^N/C$ Adinkra so that $H$ is the node choice group.  Choosing which edges are dashed would then determine an Adinkra for a supermultiplet.

So, adinkrizable supermultiplets are thereby classified up to choosing which edges are dashed.

Tables~\ref{redtable} and~\ref{ncstable} show the resulting classification for $N \leq 4$, and Table~\ref{N_5_table} shows the analogous numbers (but not the Adinkras since there are too many) for $N = 5$. Note that the node choice group and the doubly even kernel, whose generator matrices are listed, are listed uniquely up to permutation equivalence of binary linear codes, and that the Adinkras are listed uniquely up to graph theoretic isomorphisms which preserve heights and permute edge colors (and swapping bosons for fermions).

Finally, to each tabulated Adinkra and supermultiplet, there corresponds another, differing only in that the black/white node-coloring is swapped, \ie, with fermions and bosons exchanged. To save space, we have not tabulated these Klein-flipped Adinkras. Also, ``NCG'' denotes the node choice group, denoted $H$ in the text, and ``DE'' denotes doubly even, in these tables referring to the maximal doubly even subcode $C\subset H$.

\begin{table}[ht]
  \centering
  \begin{tabular}{p{10mm}|>{\centering\baselineskip=10pt\arraybackslash}p{25mm}
                         |>{\centering\baselineskip=10pt\arraybackslash}p{25mm}
                         |>{\centering\baselineskip=10pt\arraybackslash}p{30mm}
                         |c}
 \boldmath$N$ & \bf\boldmath NCG~$H$ Generators
              & \bf\boldmath DE~$C\subset H$ Generators
              & \bf Decomposable Adinkra$^\dagger$
              & \bf Indecomposable Components$^*$ \\[2mm]
    \hline\hline
    4 & $\Cd{1&1&1&1}$ & $\Cd{\text{|}}$
      & \includegraphics[width=1in]{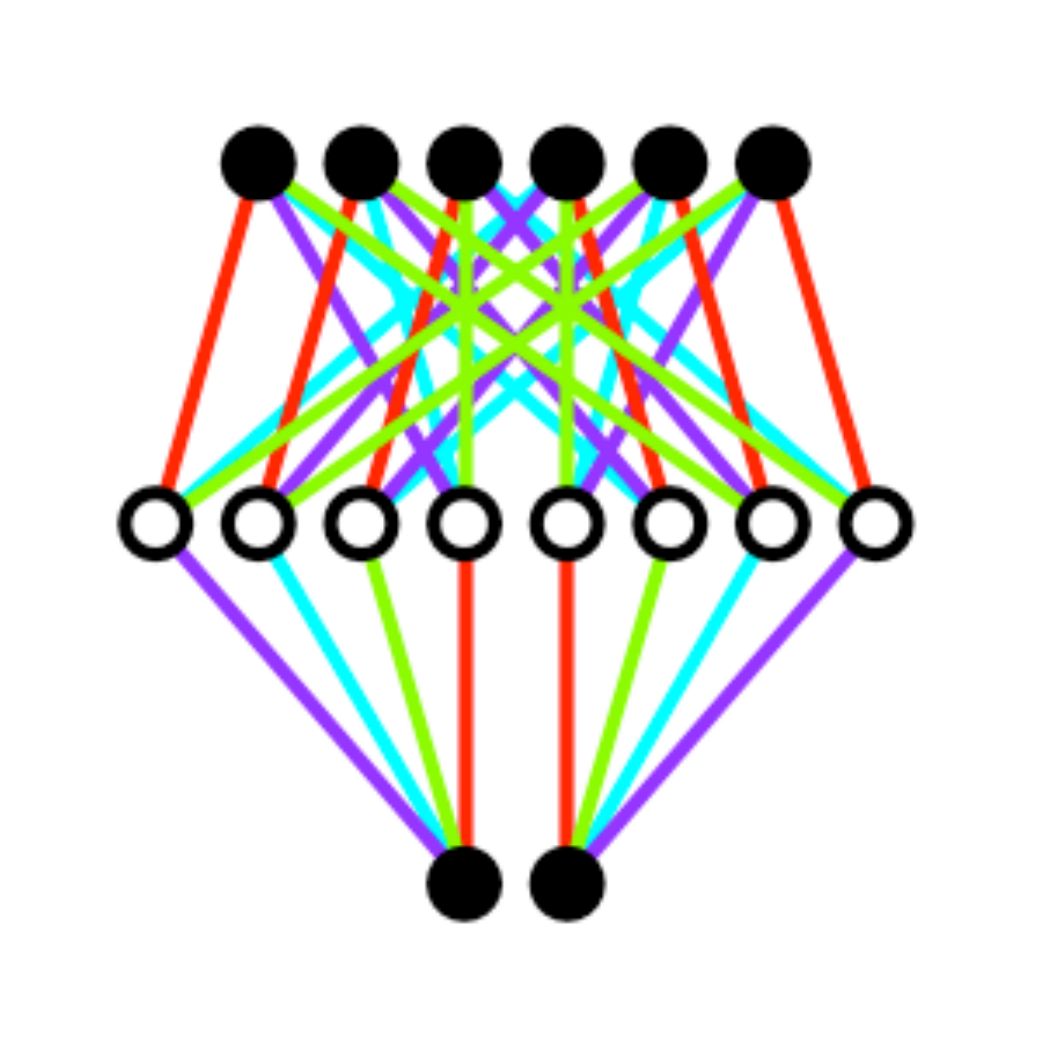}
      & \includegraphics[width=1in]{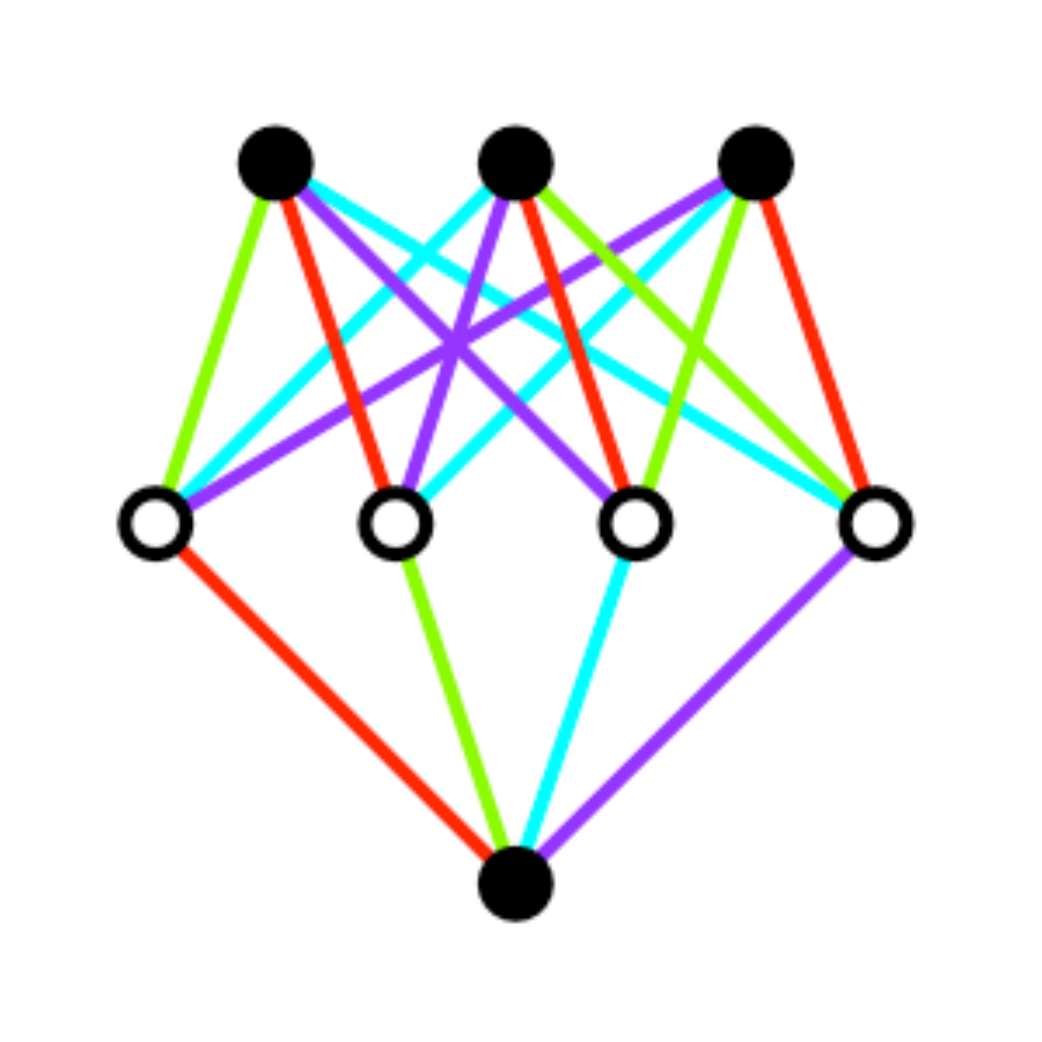}
        \includegraphics[width=1in]{Pix/A_4_5.pdf} \\ 
      &  & 
      & \includegraphics[width=1in]{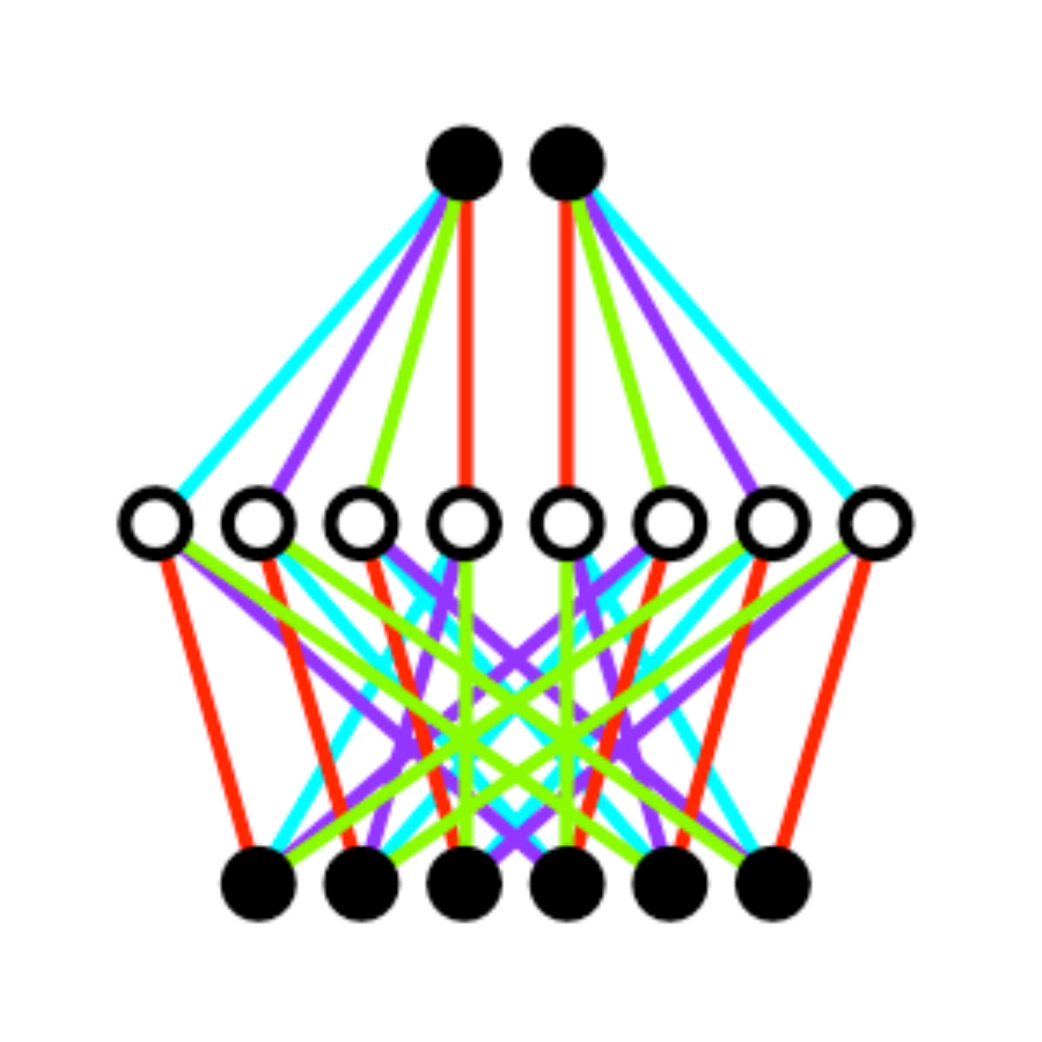}
      & \includegraphics[width=1in]{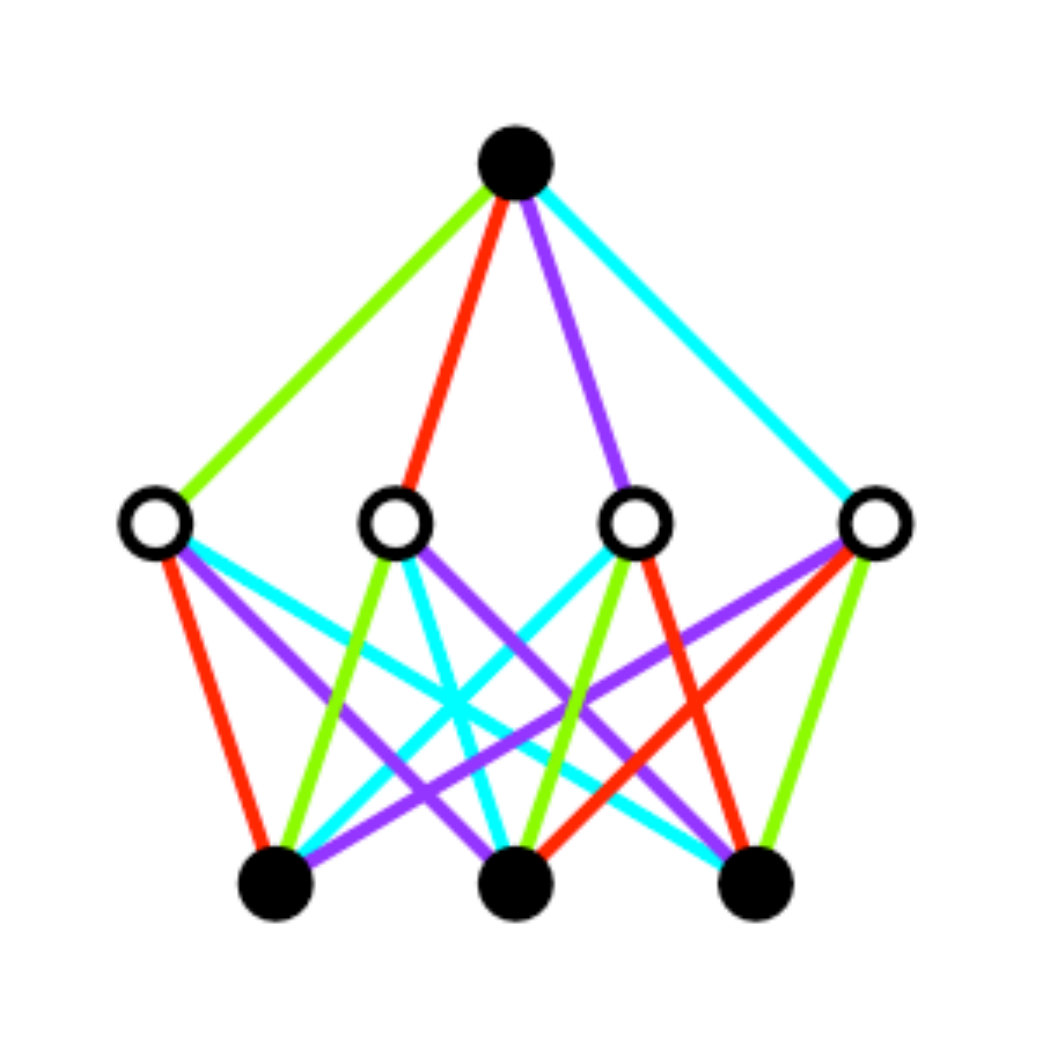}
        \includegraphics[width=1in]{Pix/A_4_6.pdf} \\[1mm]
    \cline{2-5}
      & \raisebox{10mm}[1.05\height]{$\Cd{1&1&1&1}$}
      & \raisebox{10mm}[1.05\height]{$\Cd{\text{|}}$}
      & \raisebox{-1mm}[1.05\height]{\includegraphics[width=1in]{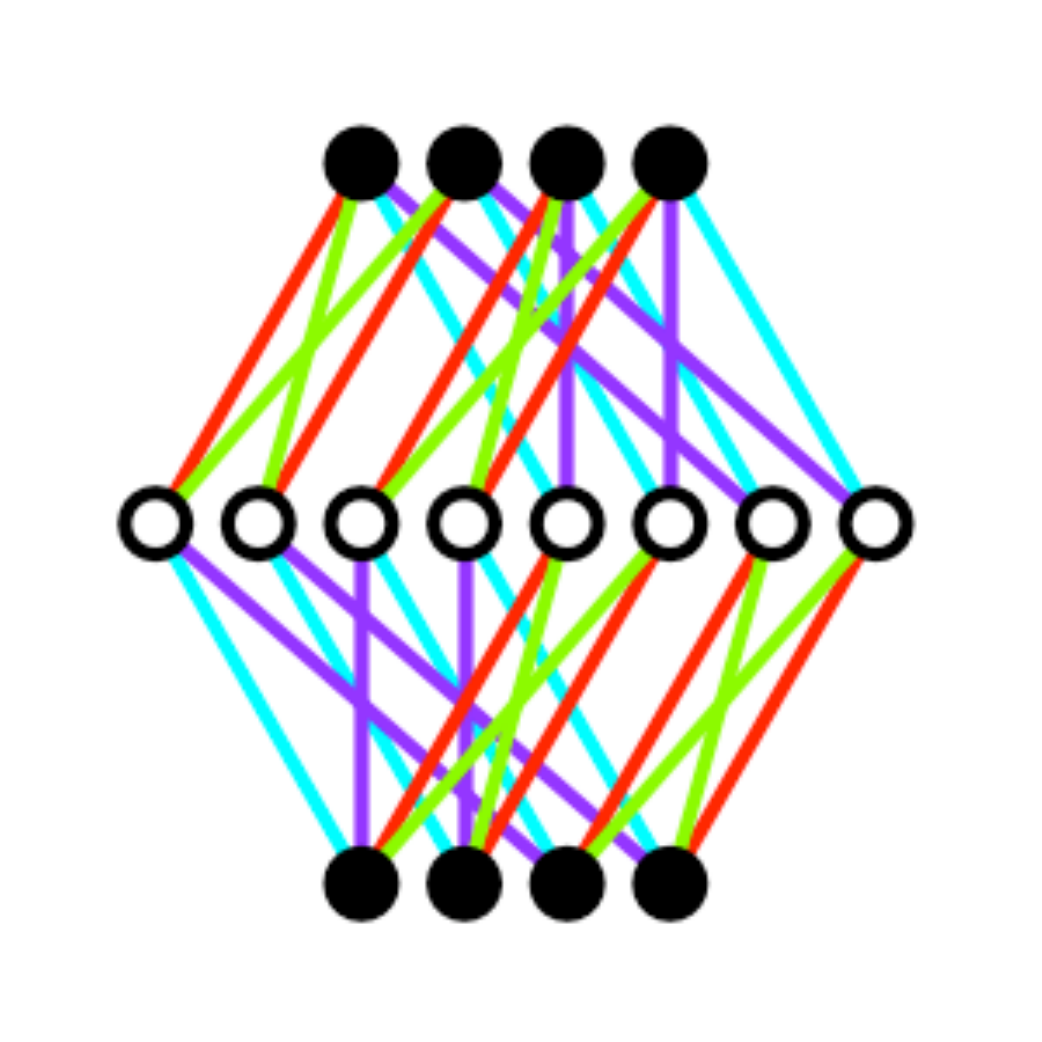}}
      & \raisebox{-1mm}[1.05\height]{\includegraphics[width=1in]{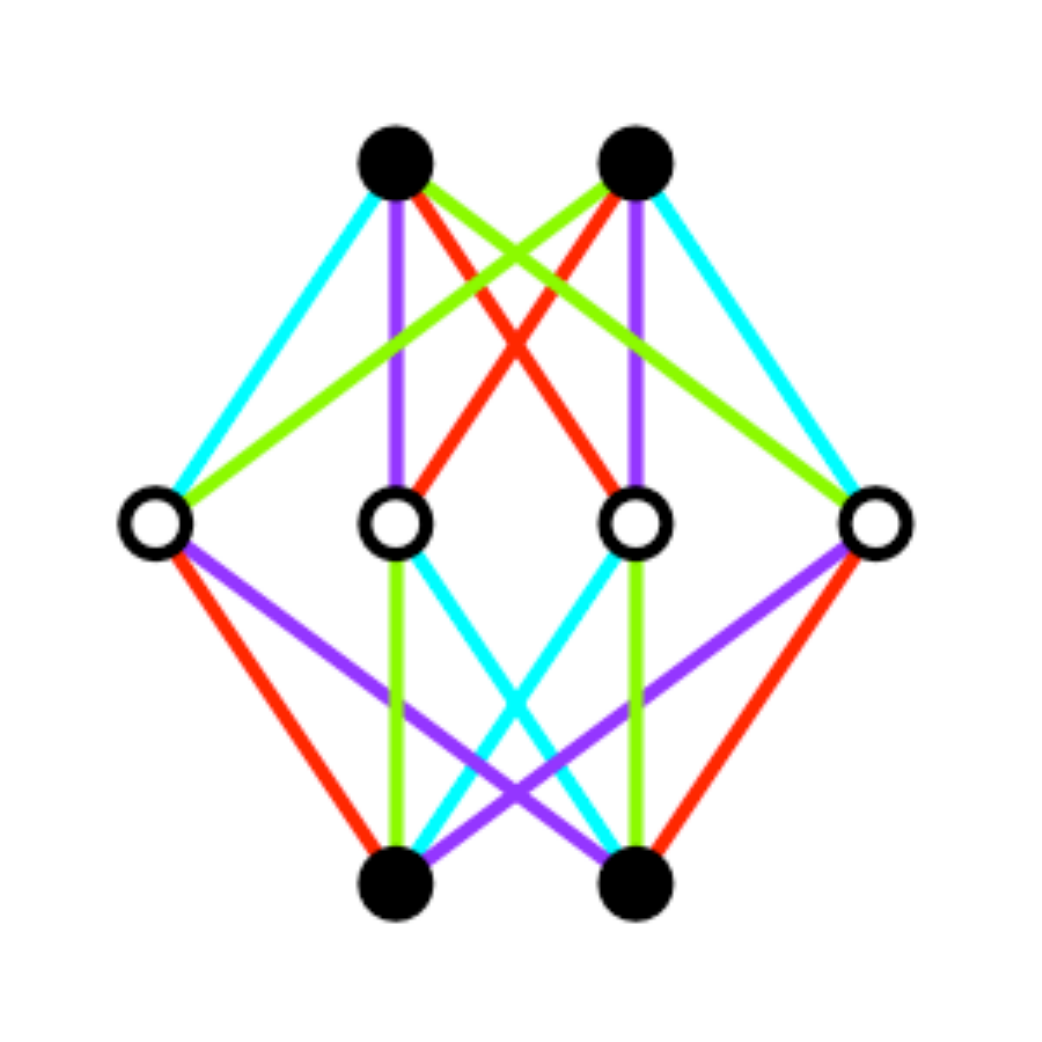}}
        \raisebox{-1mm}[1.05\height]{\includegraphics[width=1in]{Pix/A_4_3.pdf}} \\ 
    \cline{2-5}
   & \raisebox{10mm}[1.05\height]{$\Cd{1&1&0&0\\[1pt]1&0&1&0\\[1pt]1&1&1&1}$}
   & \raisebox{10mm}[1.05\height]{$\Cd{\text{|}}$}
   & \raisebox{-1mm}[1.05\height]{\includegraphics[width=1in]{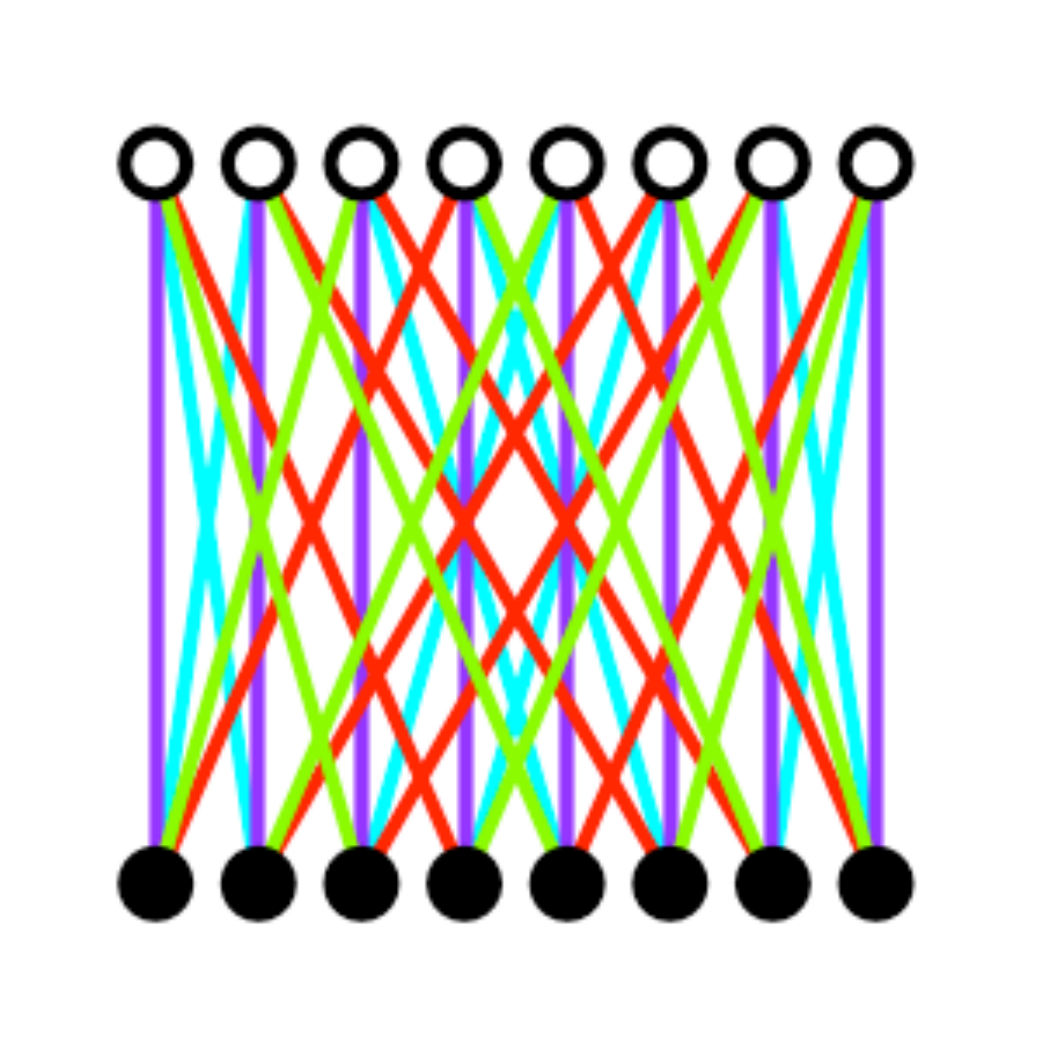}}
   & \raisebox{-1mm}[1.05\height]{\includegraphics[width=1in]{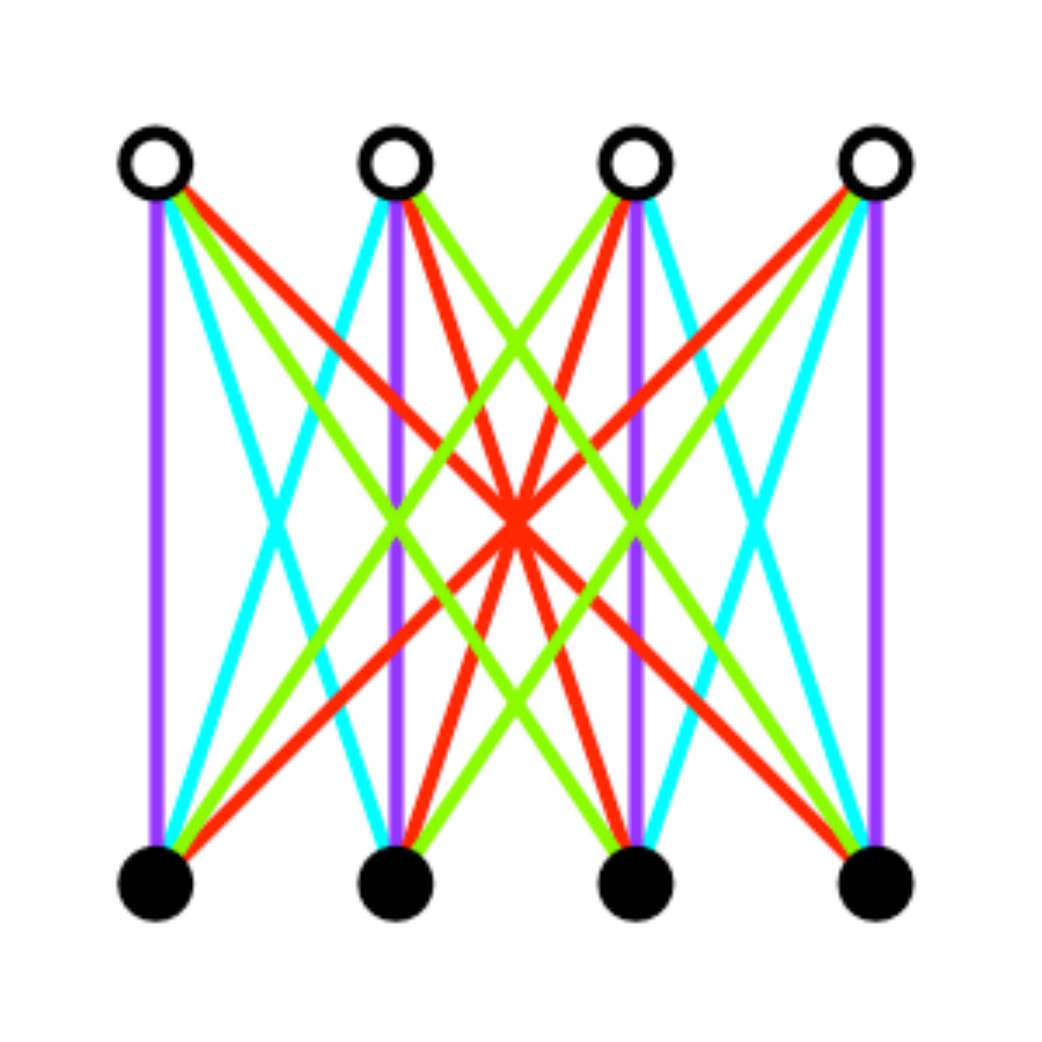}}
     \raisebox{-1mm}[1.05\height]{\includegraphics[width=1in]{Pix/A_4_1.pdf}} \\
    \hline \noalign{\vglue2mm}
 \multicolumn{5}{l}{\parbox{\linewidth}{\footnotesize $^*$ Here, the two components in each row appear identical, but in fact admit two inequivalent choices of edge-dashing. In general, this determination involves cohomology and will be discussed separately.}}\\\noalign{\vglue2mm}
  \end{tabular}
  \caption{Decomposable Adinkraic Supermultiplets for $N \leq 4$}
  \label{redtable}
\end{table}

\begin{longtable}{p{.03\linewidth}|>{\centering\arraybackslash}m{.15\linewidth}|>{\centering\arraybackslash}m{.15\linewidth}|>{\arraybackslash} m{.50\linewidth}}
 \boldmath$N$ & \baselineskip=10pt \bf\boldmath NCG:~$H$ Generators
              & \baselineskip=10pt \bf\boldmath DE~$C\subset H$ Generators
              & \bf Adinkras Depicting the Supermultiplet \endhead
 \hline\hline
 1 & $\Cd{\text{|}}$ & $\Cd{\text{|}}$
   & \raisebox{-1mm}[1.02\height]{\includegraphics[width=1in]{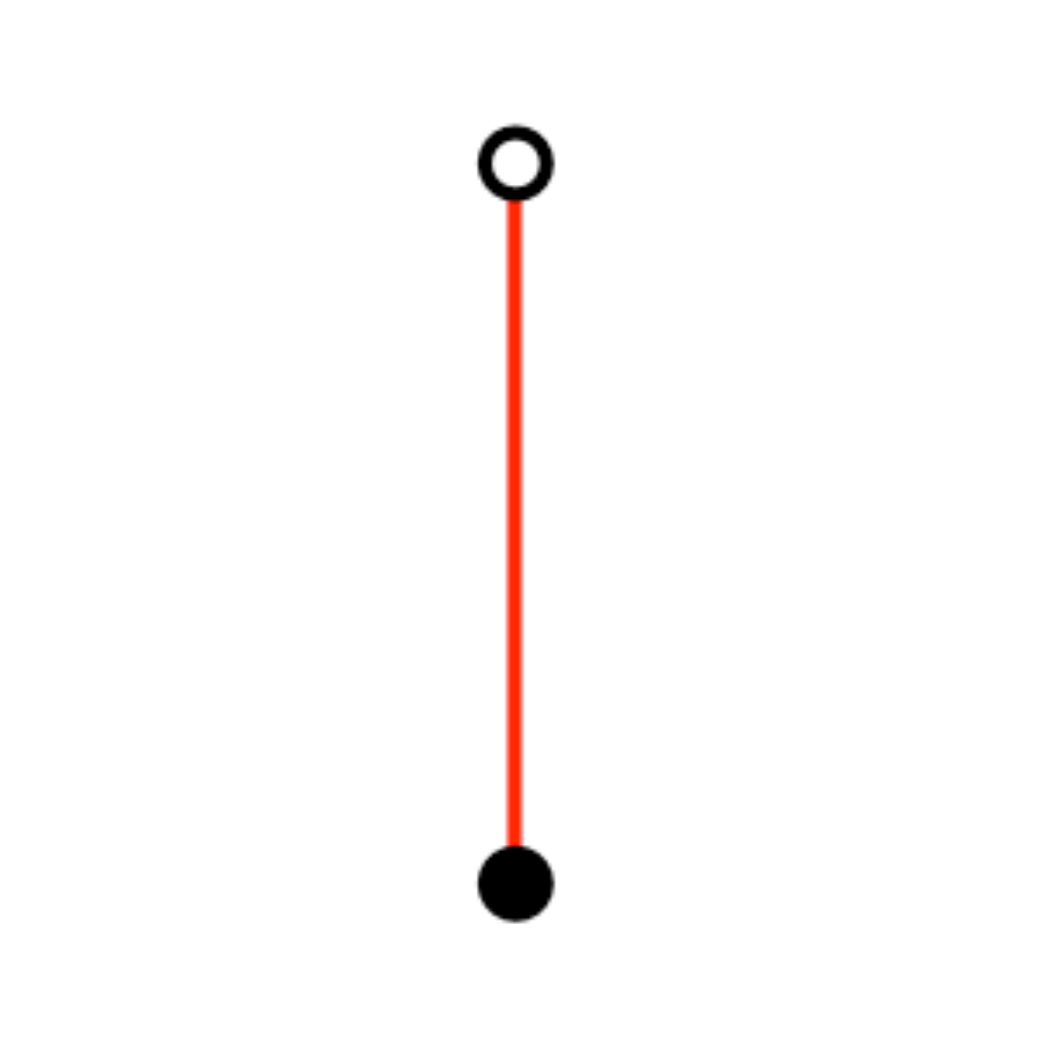}} \\
 \hline
 2 & $\Cd{\text{|}}$ & $\Cd{\text{|}}$
   & \raisebox{-1mm}[1.02\height]{\includegraphics[width=1in]{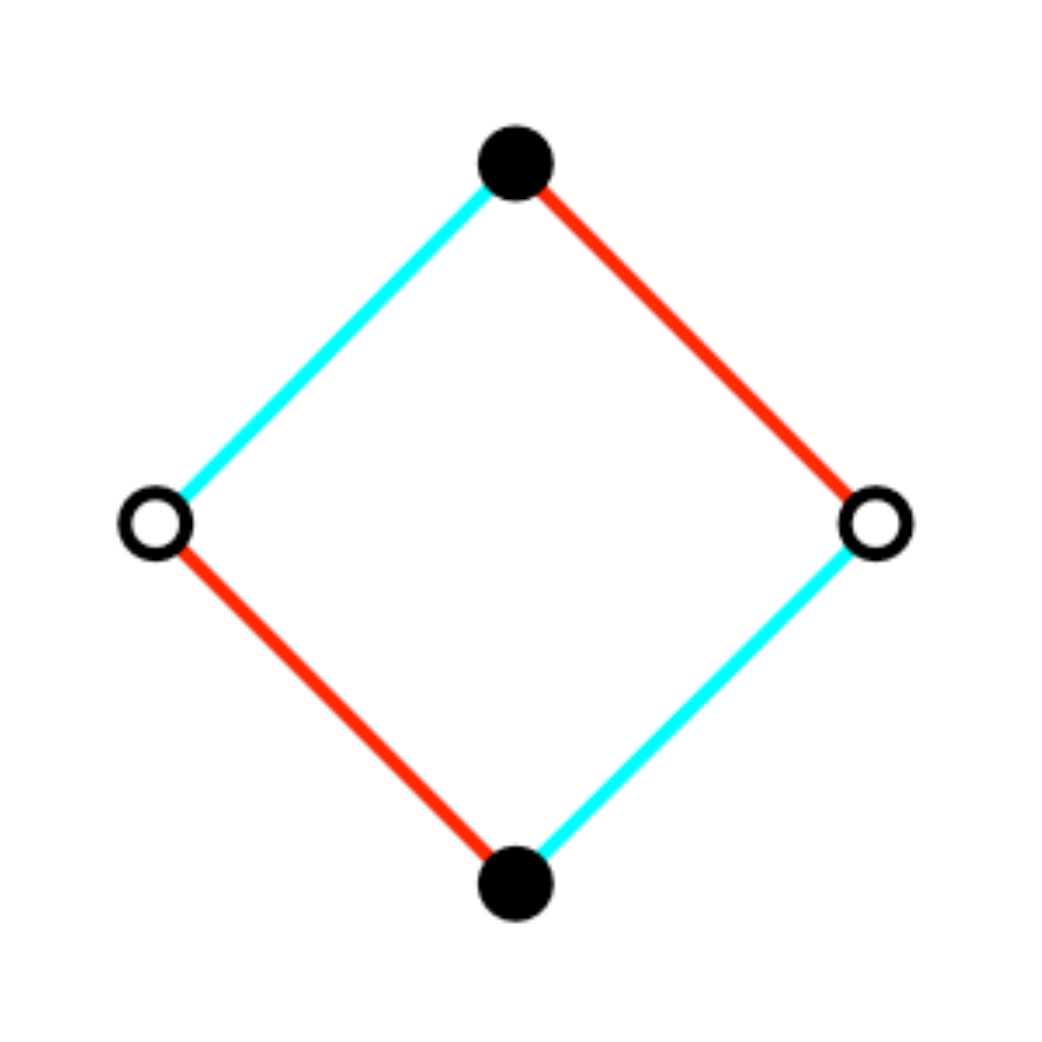}} \\
 \cline{2-4}
   & $\Cd{1&1}$ & $\Cd{\text{|}}$
   & \raisebox{-1mm}[1.02\height]{\includegraphics[width=1in]{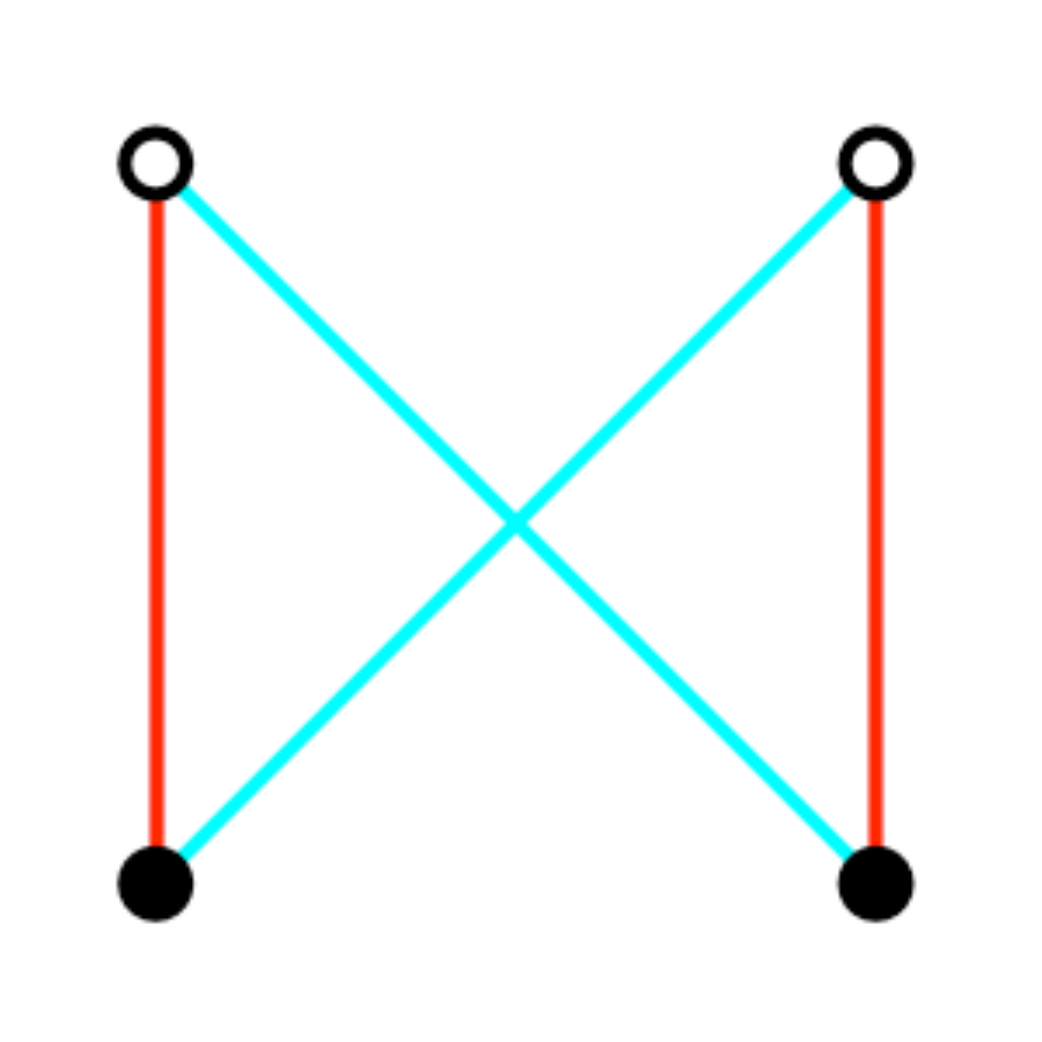}} \\
 \hline
 3 & $\Cd{\text{|}}$ & $\Cd{\text{|}}$
   & \raisebox{-1mm}[1.02\height]{\includegraphics[width=1in]{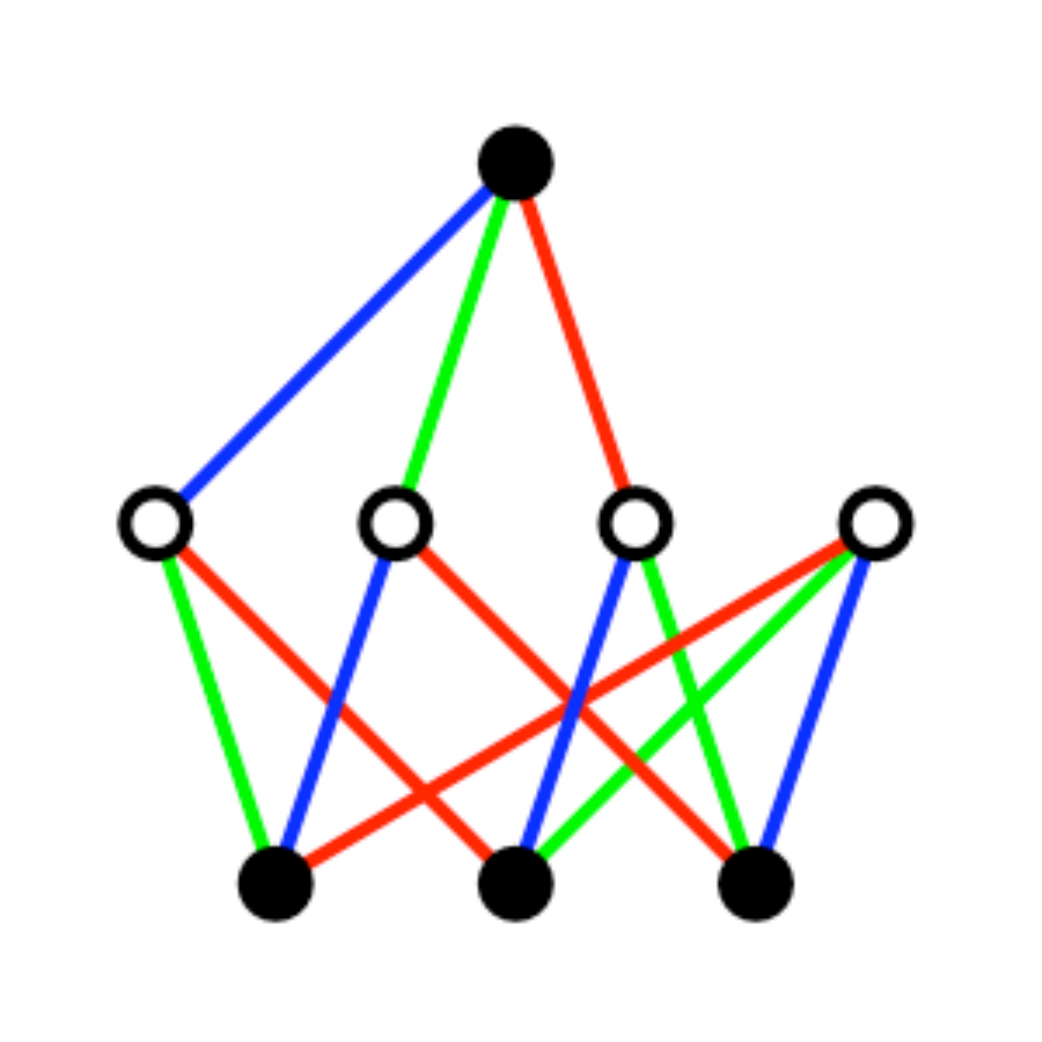}}
     \raisebox{-1mm}[1.02\height]{\includegraphics[width=1in]{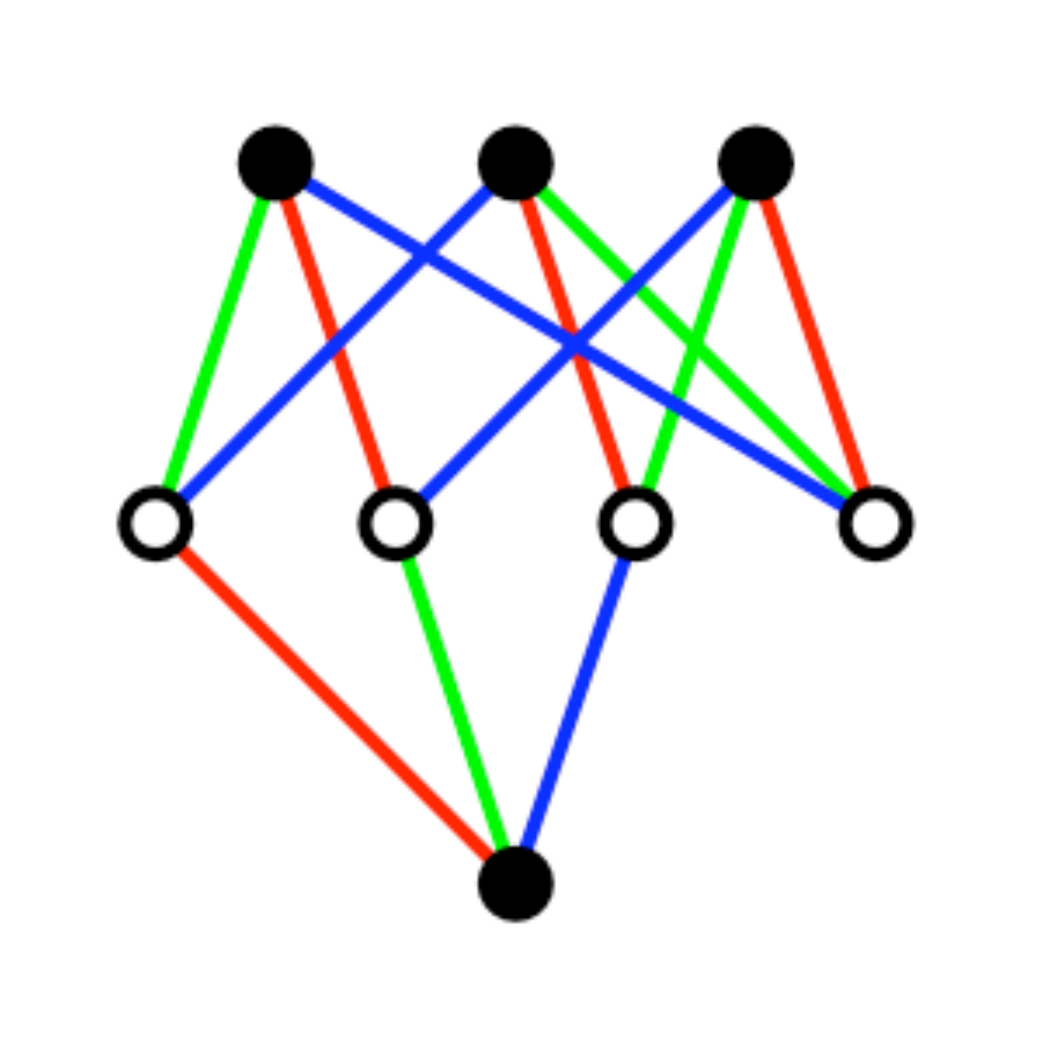}}
     \raisebox{-1mm}[1.02\height]{\includegraphics[width=1in]{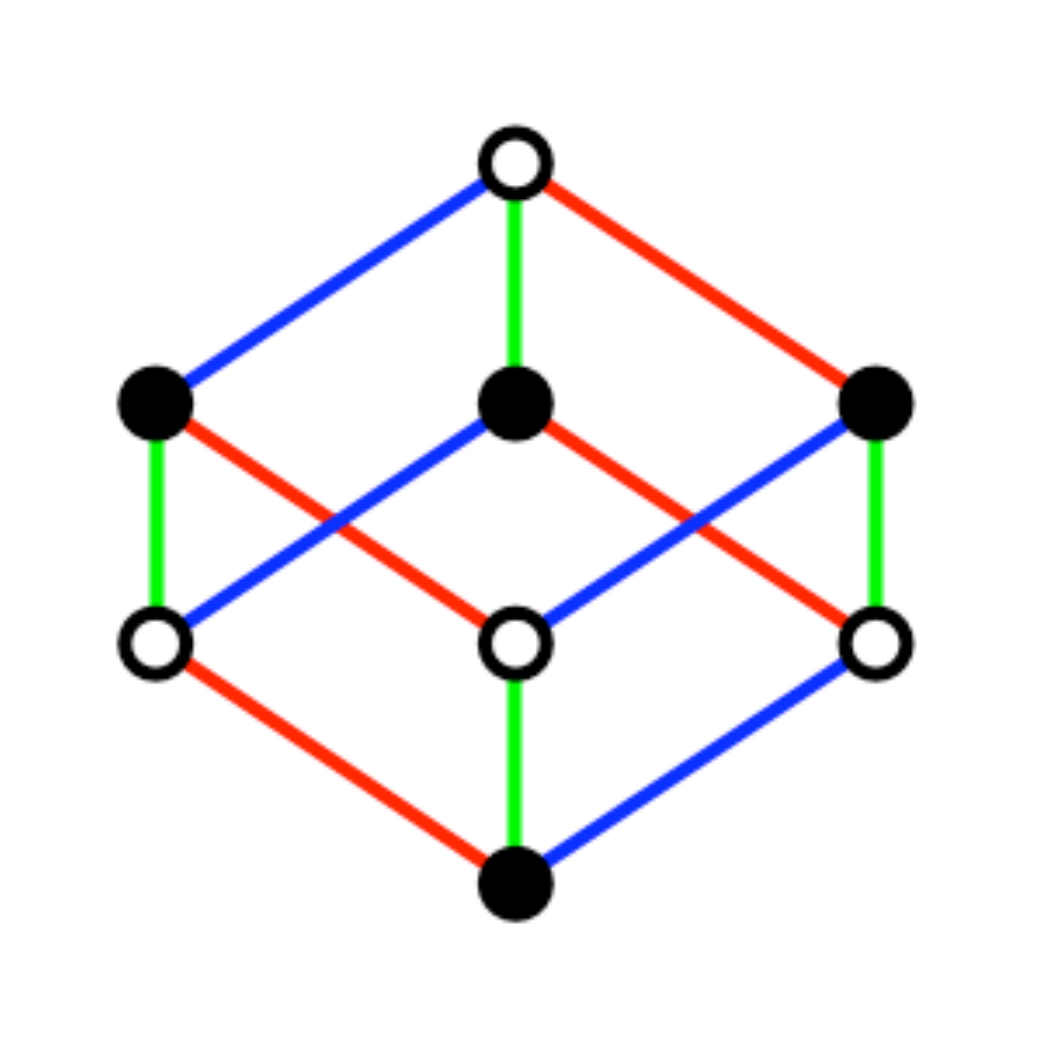}} \\
 \cline{2-4}
   & $\Cd{1&1&0}$ & $\Cd{\text{|}}$
   & \raisebox{-1mm}[1.02\height]{\includegraphics[width=1in]{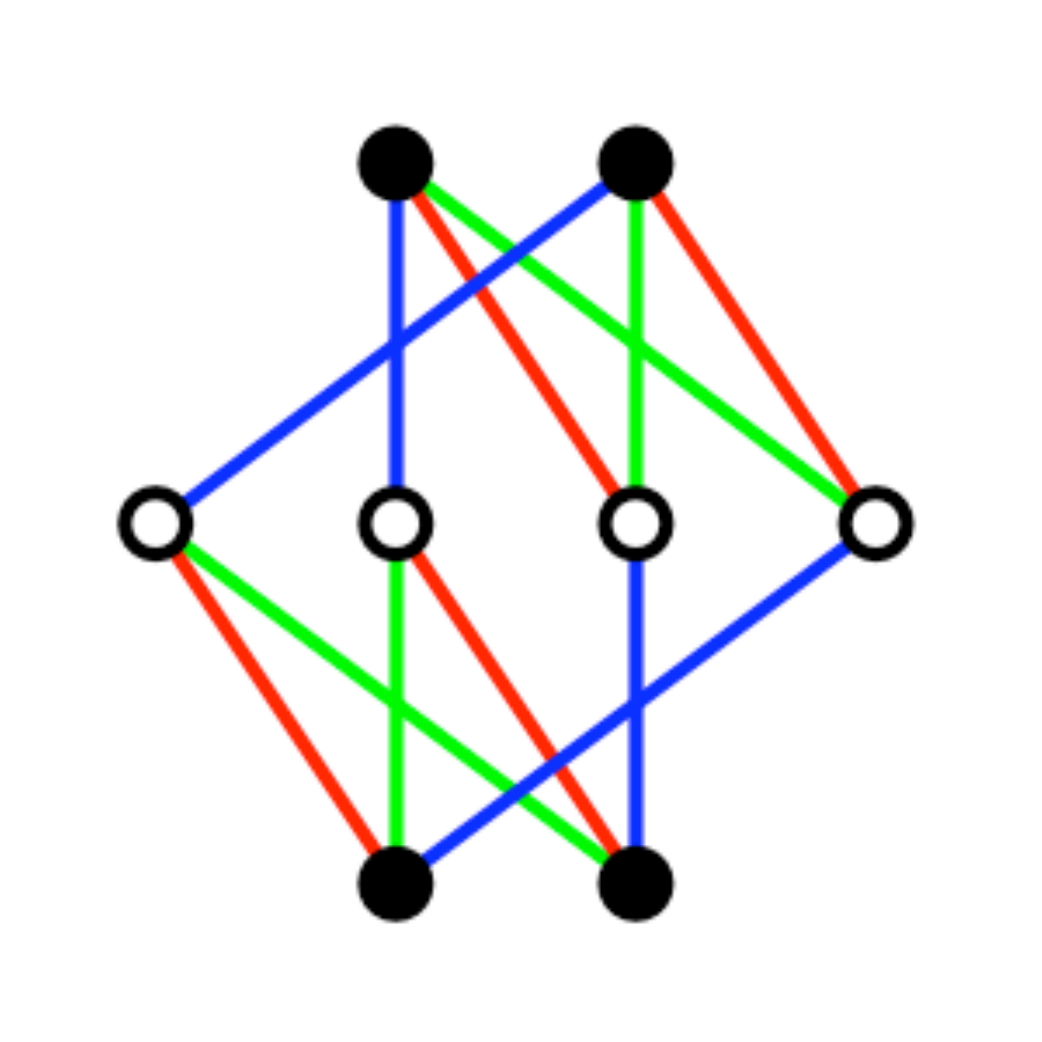}} \\
 \cline{2-4}
   & $\Cd{1&1&0\\[1pt]1&0&1}$ & $\Cd{\text{|}}$
   & \raisebox{-1mm}[1.02\height]{\includegraphics[width=1in]{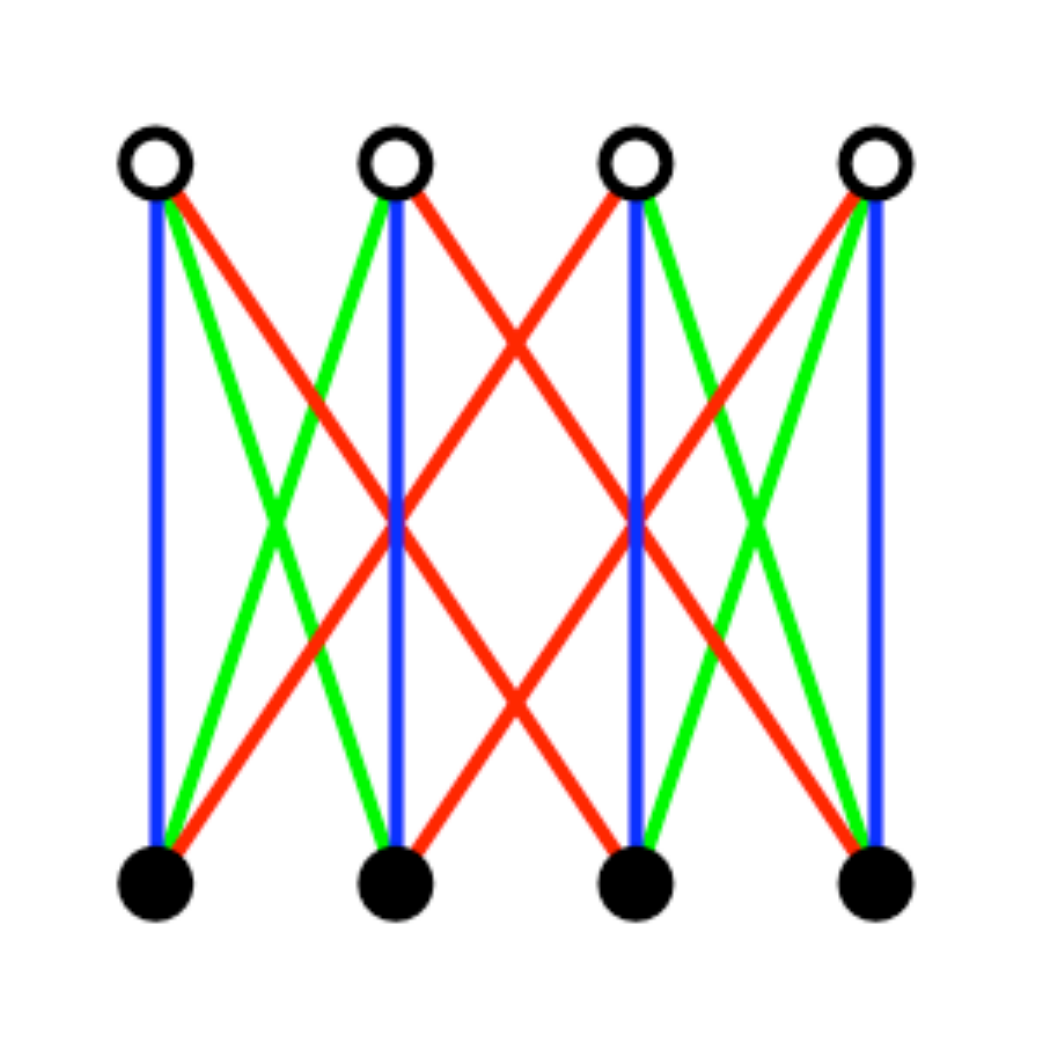}} \\
 \hline
 \multirow{6}{*}{4} & \multirow{6}{*}{$\Cd{\text{|}}$} & \multirow{6}{*}{$\Cd{\text{|}}$}
   & \raisebox{-1mm}[1.02\height]{\includegraphics[width=1in]{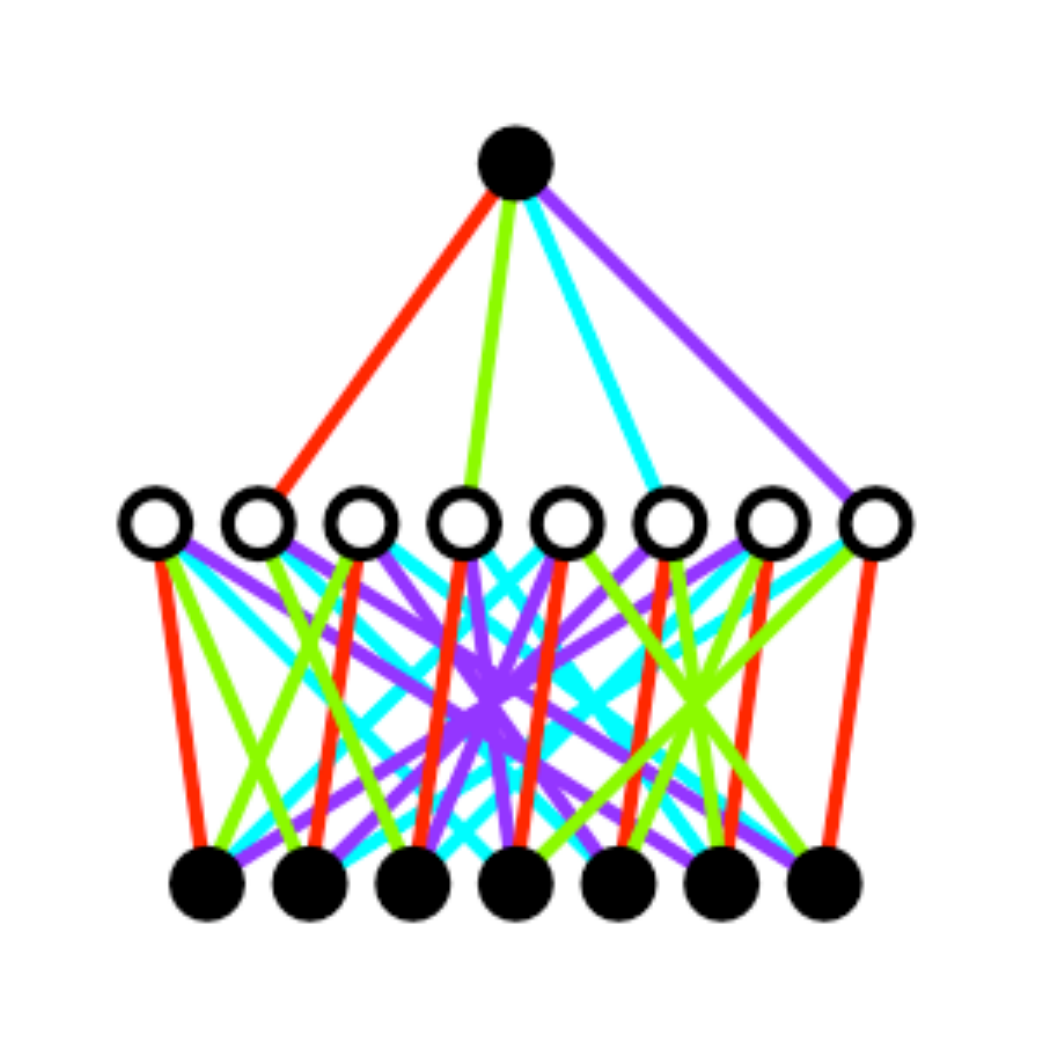}}
     \raisebox{-1mm}[1.02\height]{\includegraphics[width=1in]{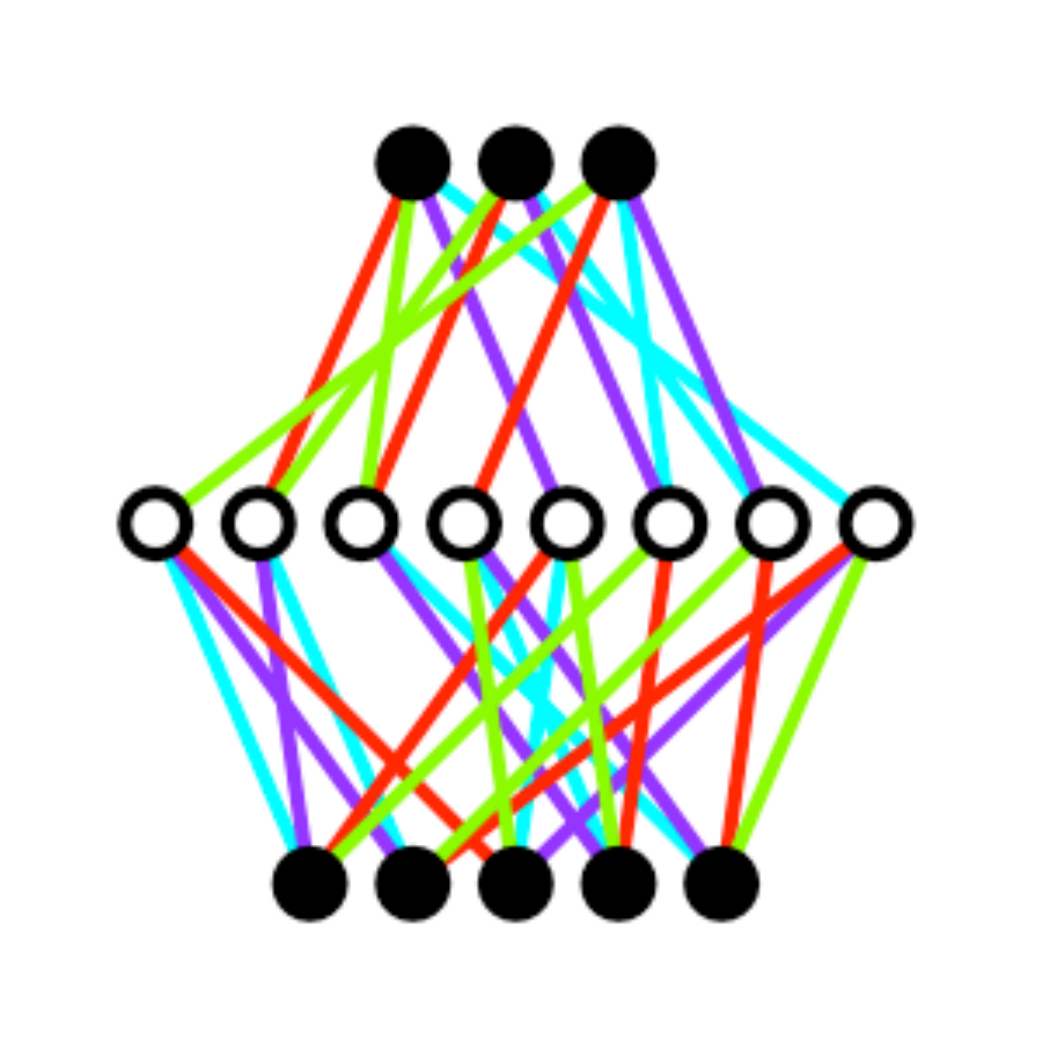}}
     \raisebox{-1mm}[1.02\height]{\includegraphics[width=1in]{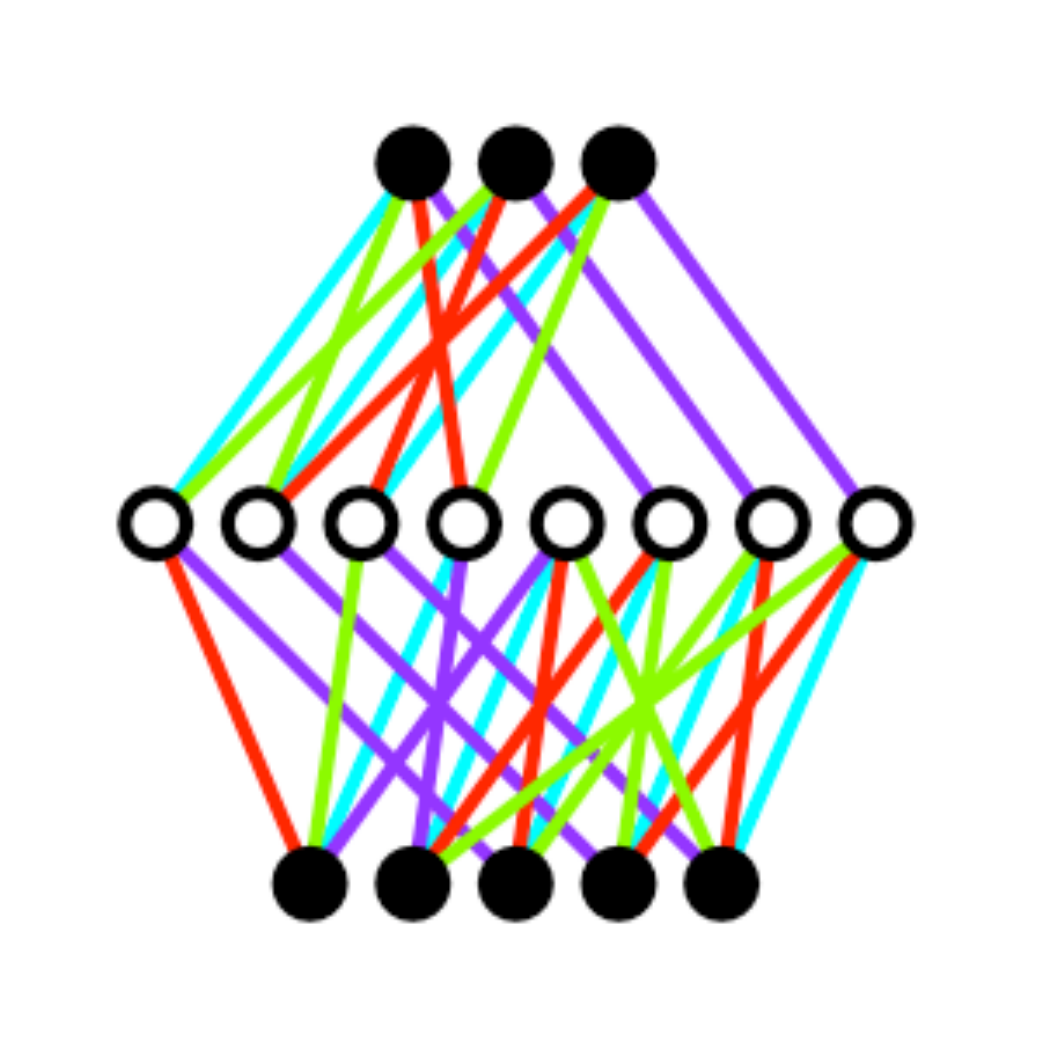}} \\
 & & & \includegraphics[width=1in]{Pix/A_4_16.pdf}
       \includegraphics[width=1in]{Pix/A_4_17.pdf}
       \includegraphics[width=1in]{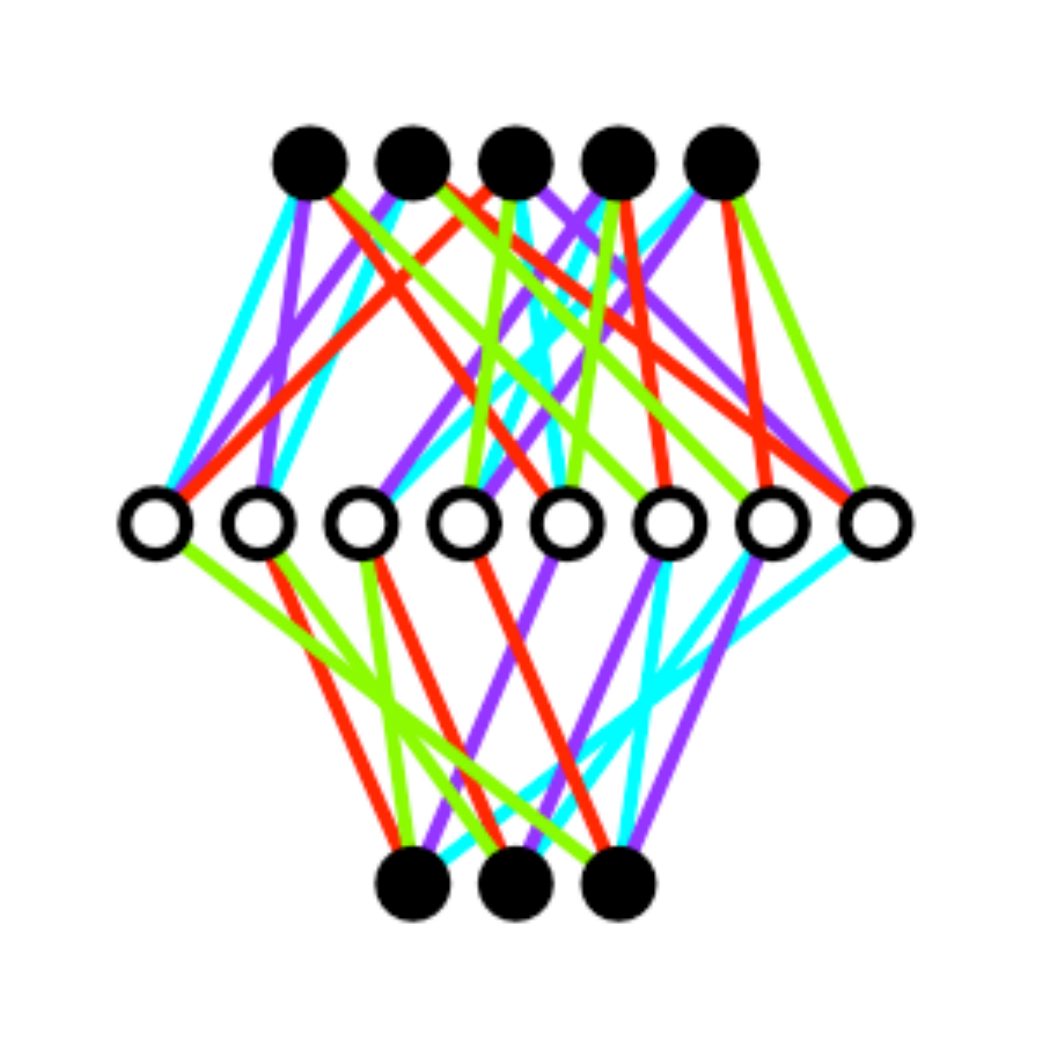} \\
 & & & \includegraphics[width=1in]{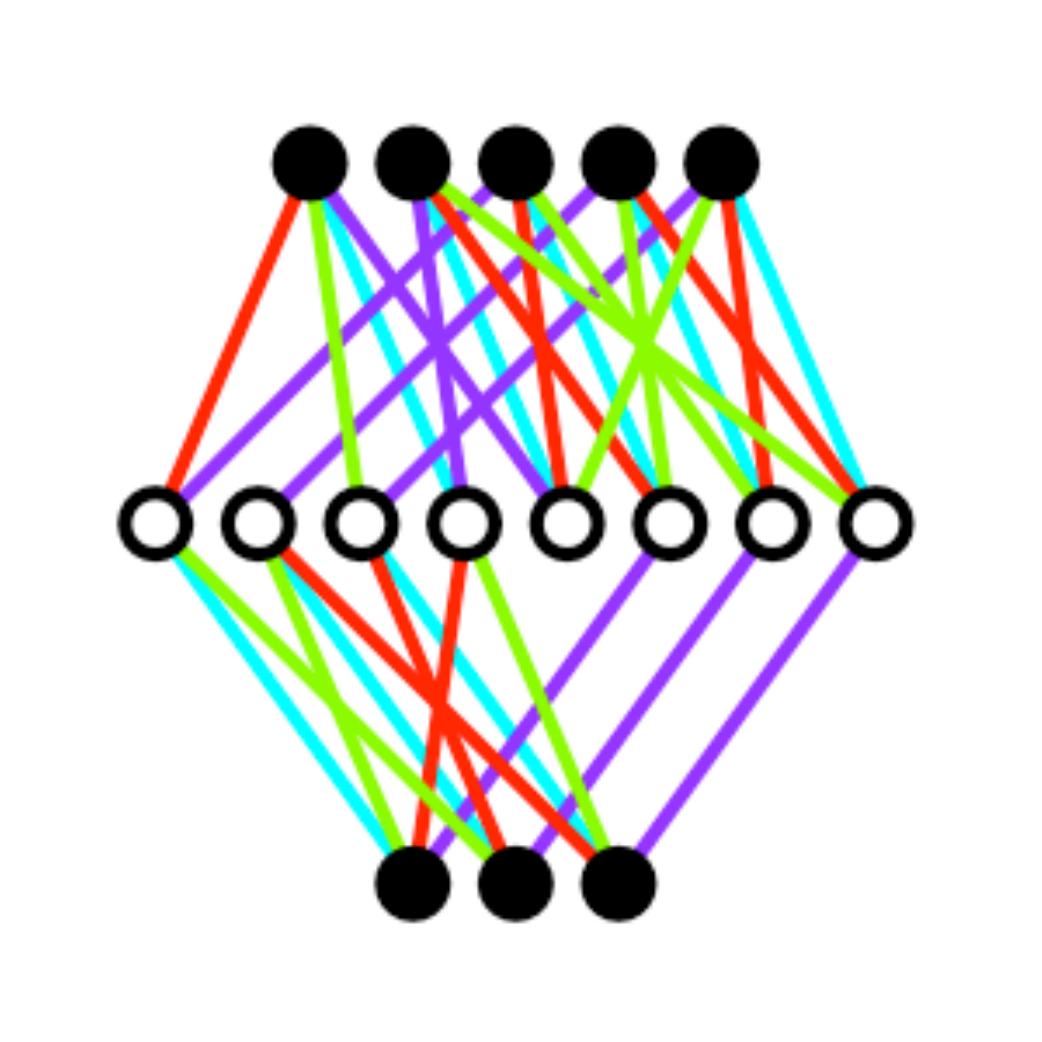}
       \includegraphics[width=1in]{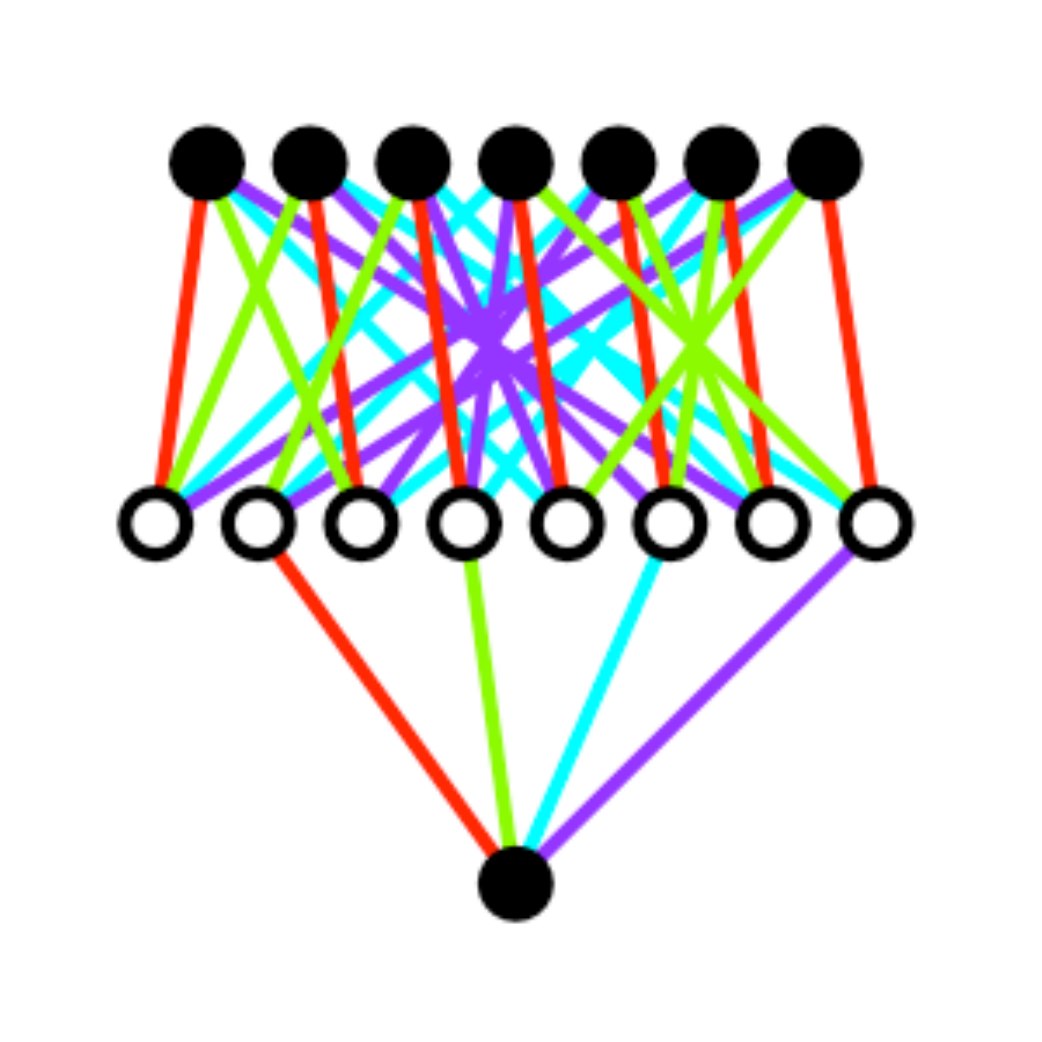}
       \includegraphics[width=1in]{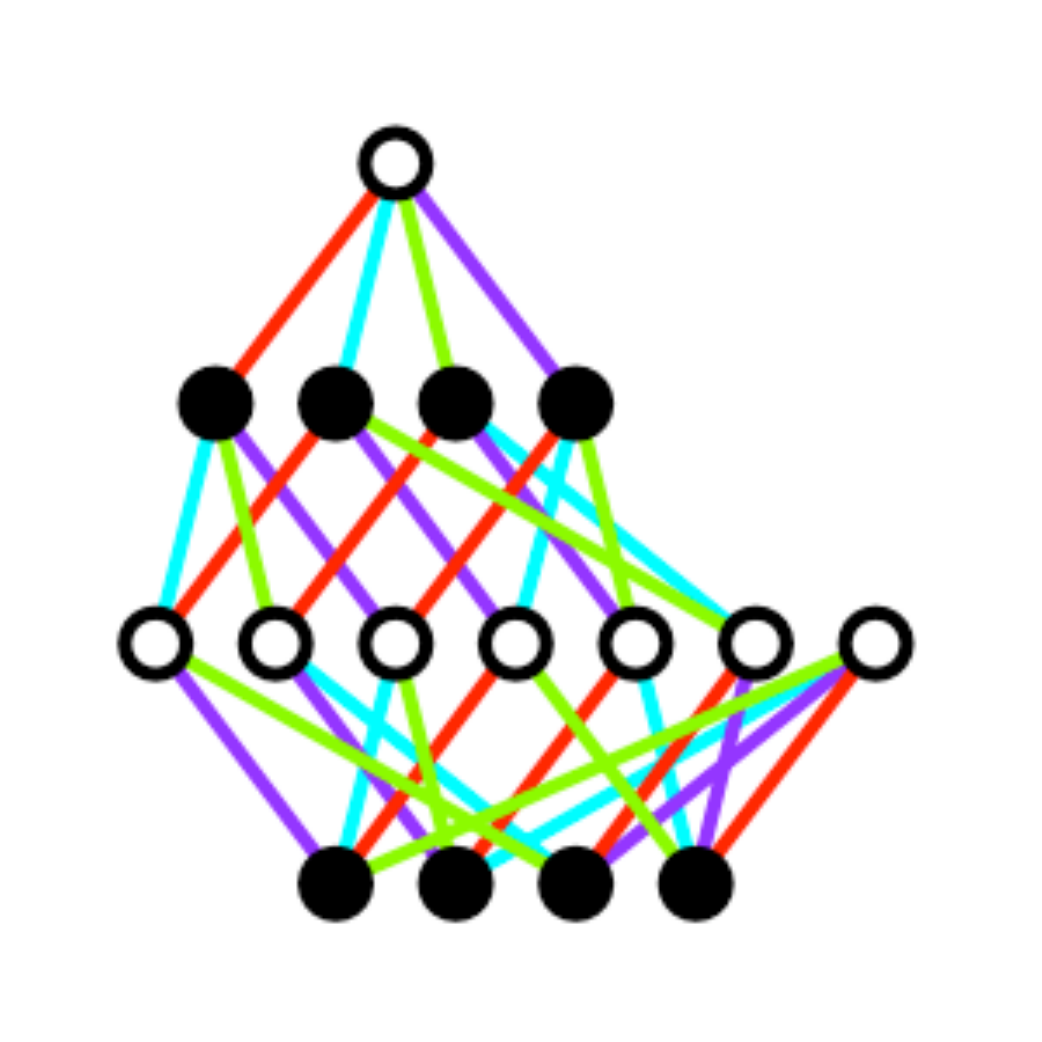} \\
 & & & \includegraphics[width=1in]{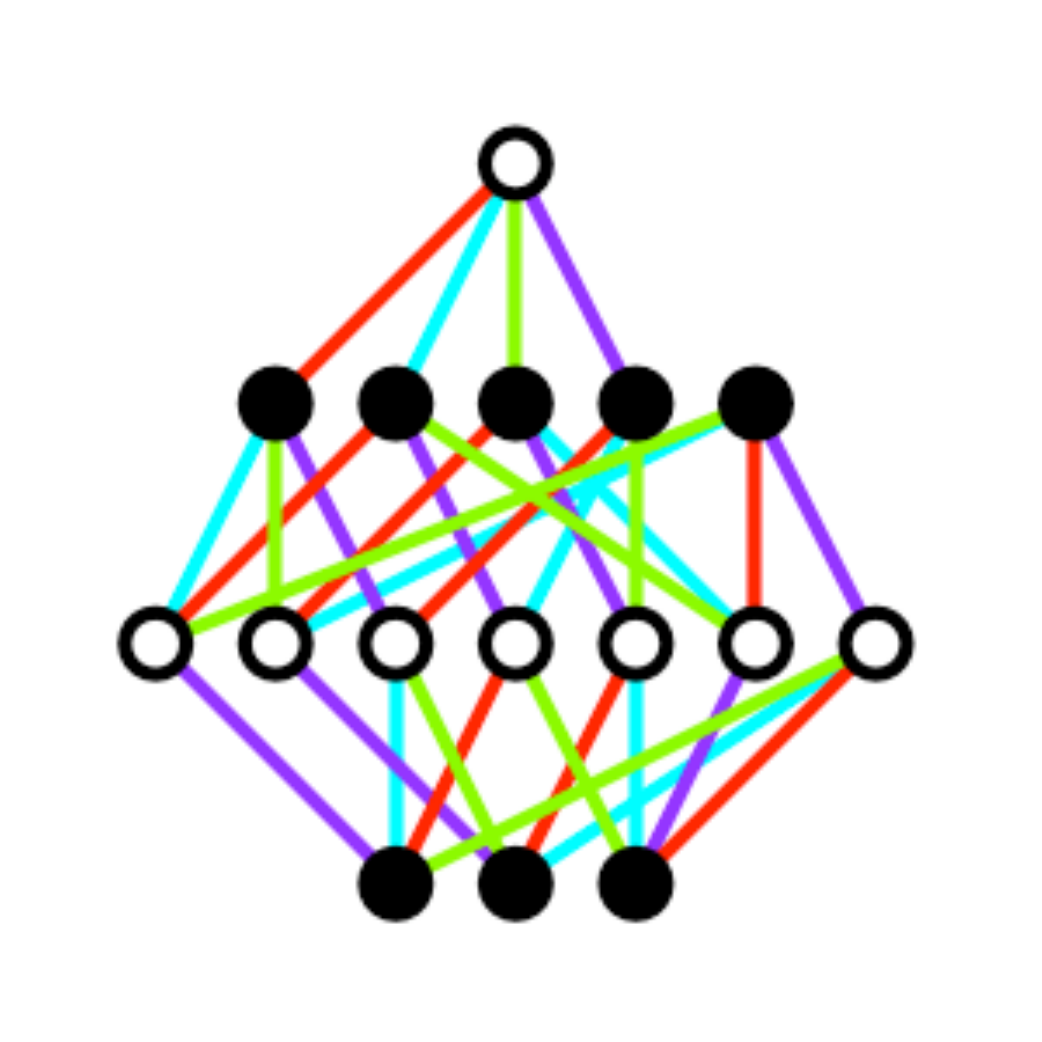}
       \includegraphics[width=1in]{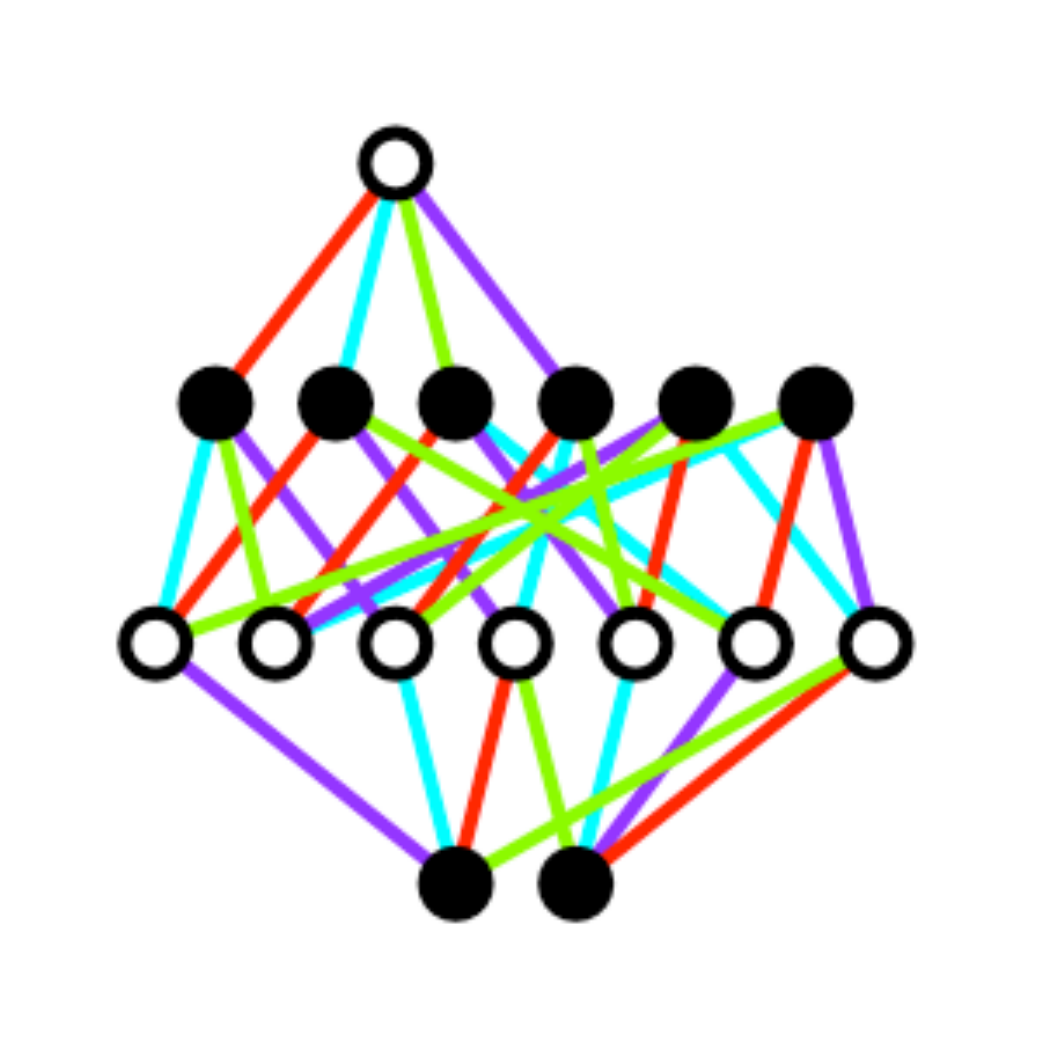}
       \includegraphics[width=1in]{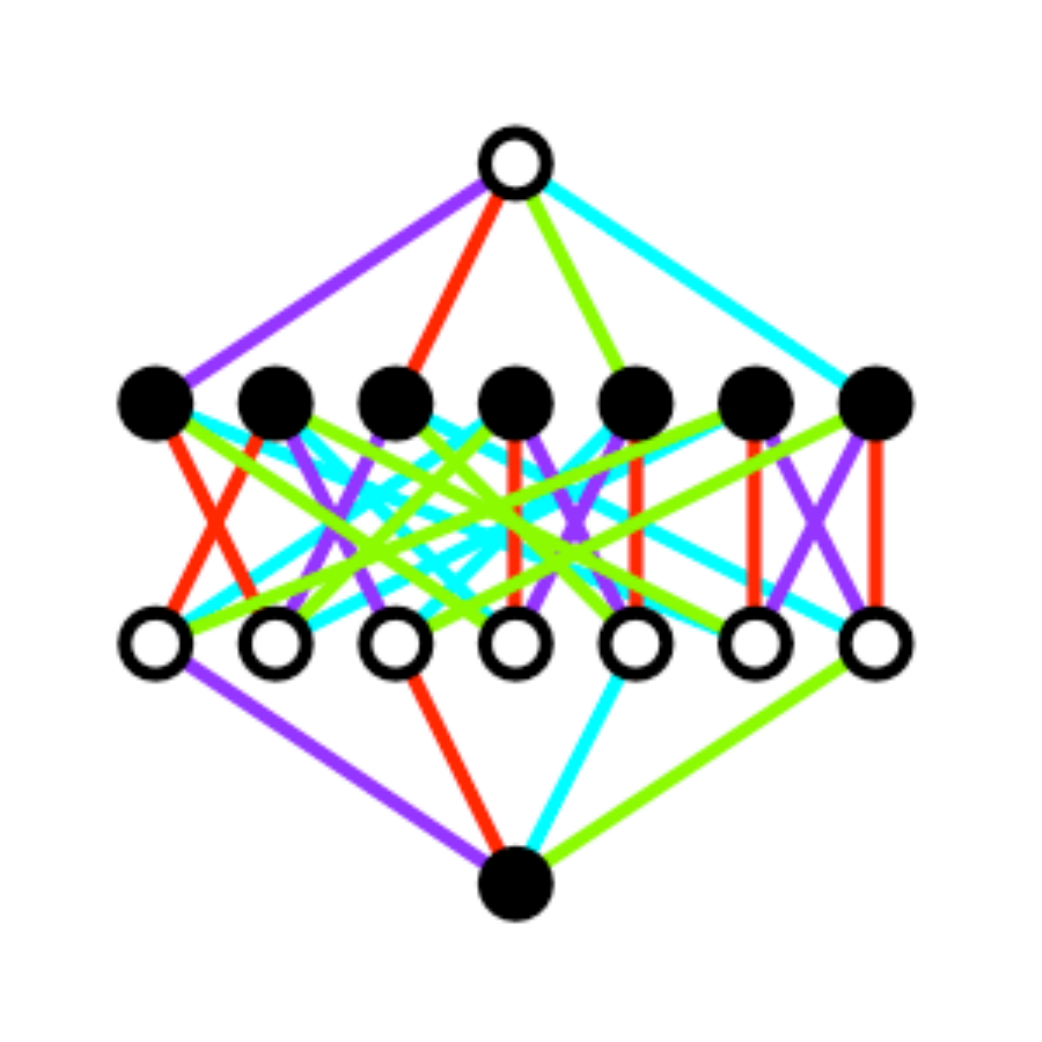} \\
 & & & \includegraphics[width=1in]{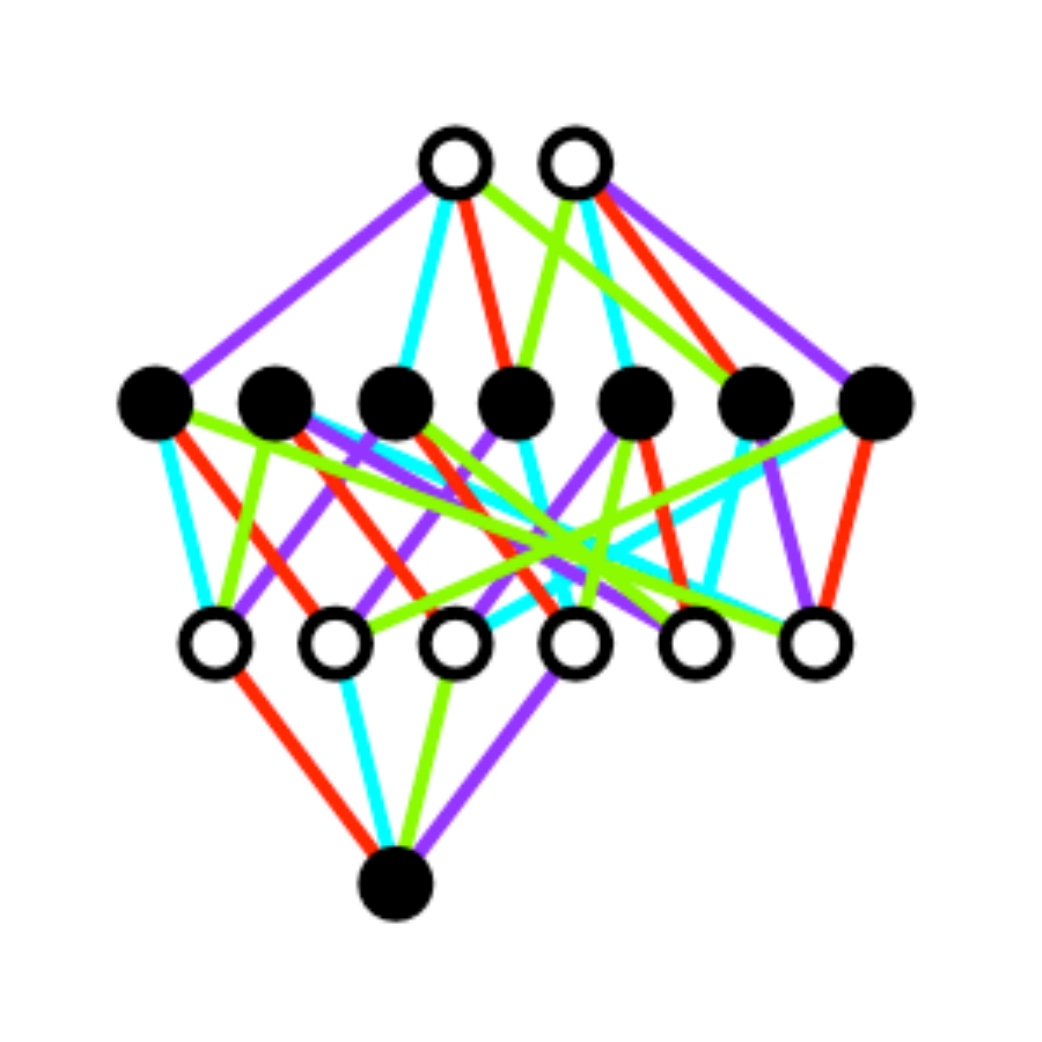}
       \includegraphics[width=1in]{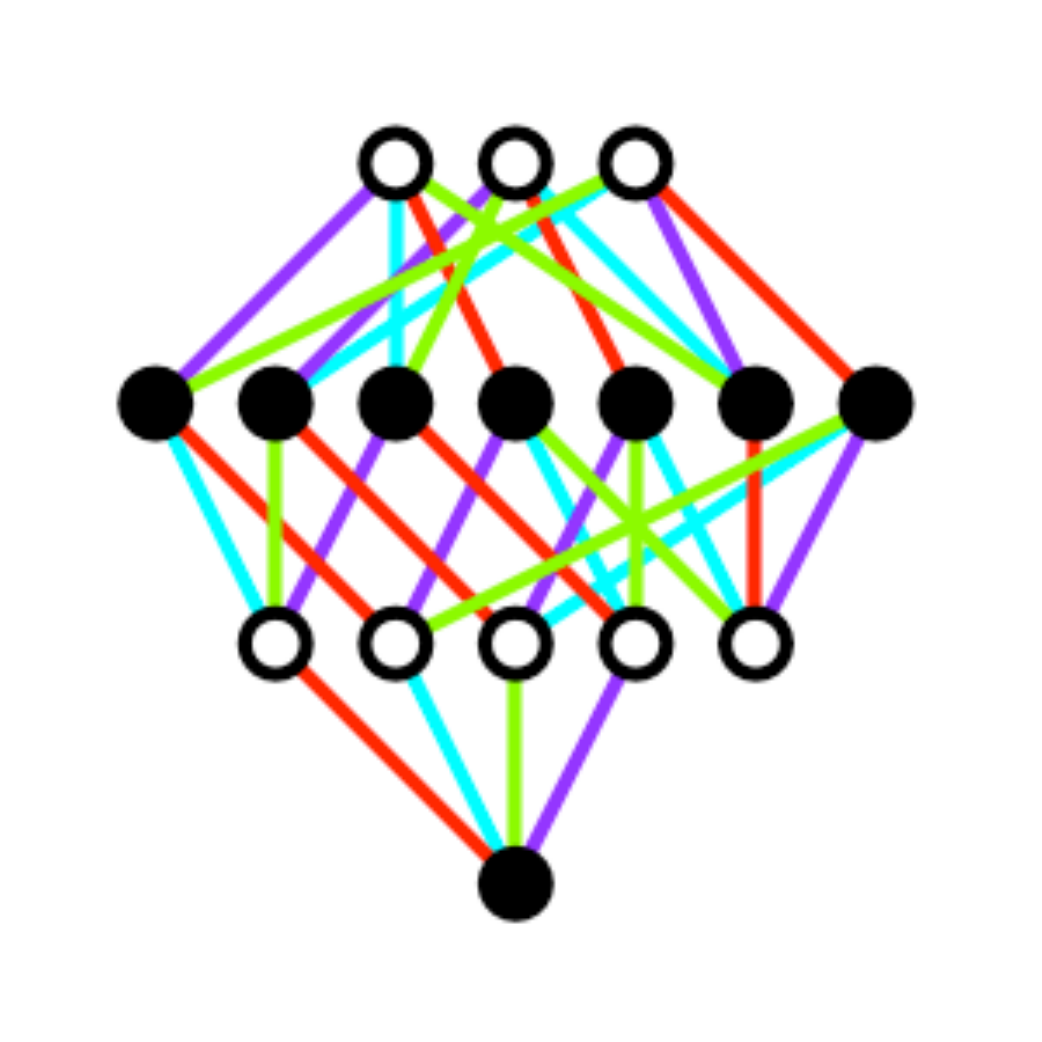}
       \includegraphics[width=1in]{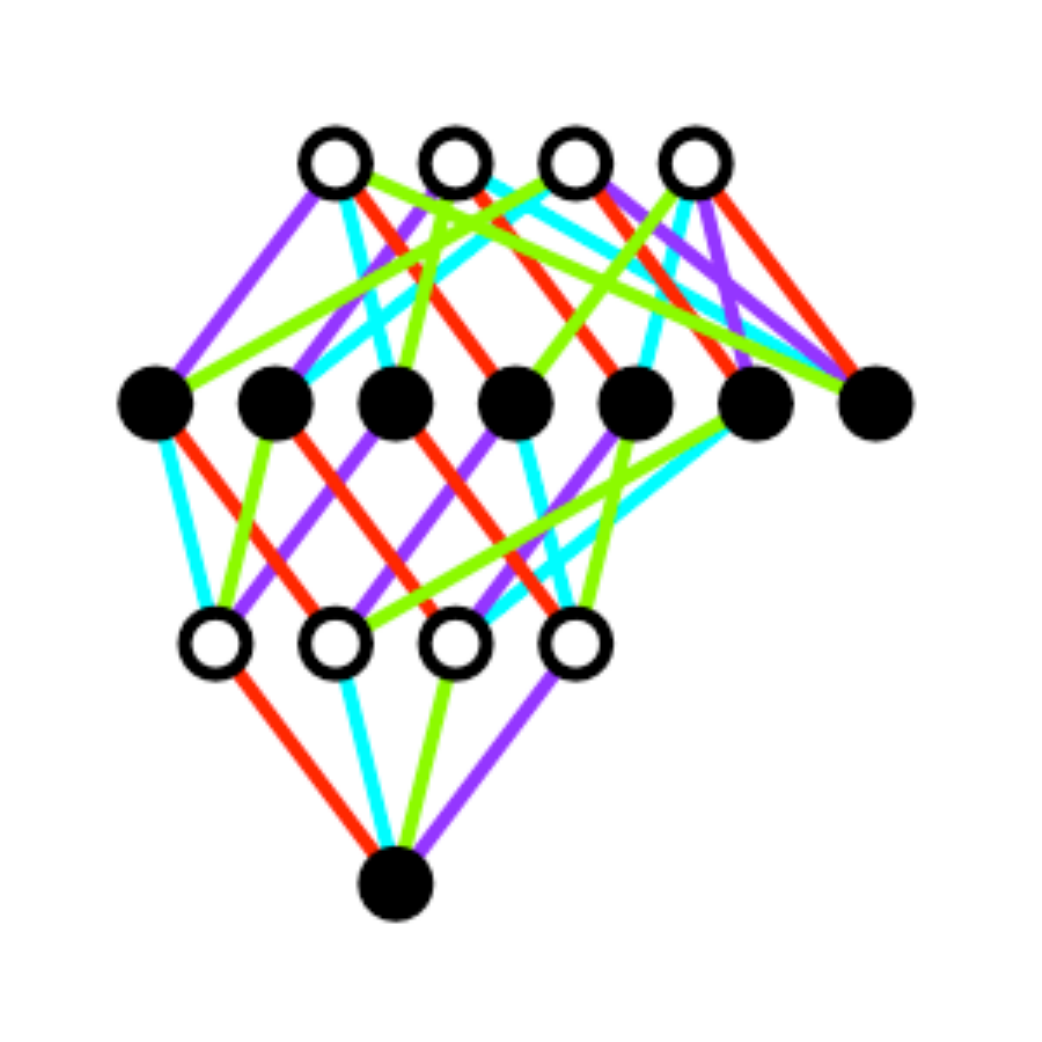} \\
 & & & \includegraphics[width=1in]{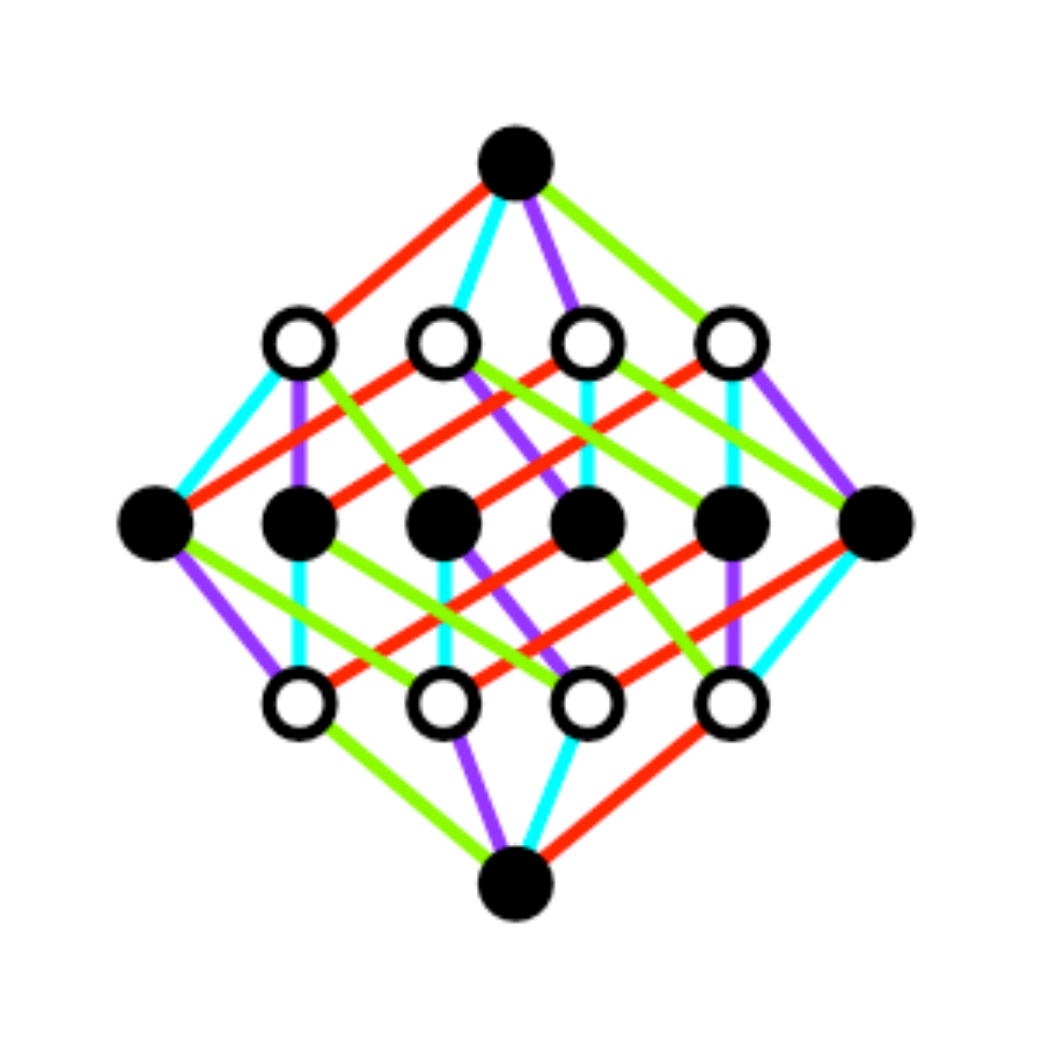} \\
 \cline{2-4}
   & $\Cd{1&1&0&0}$ & $\Cd{\text{|}}$
   & \raisebox{-1mm}[1.02\height]{\includegraphics[width=1in]{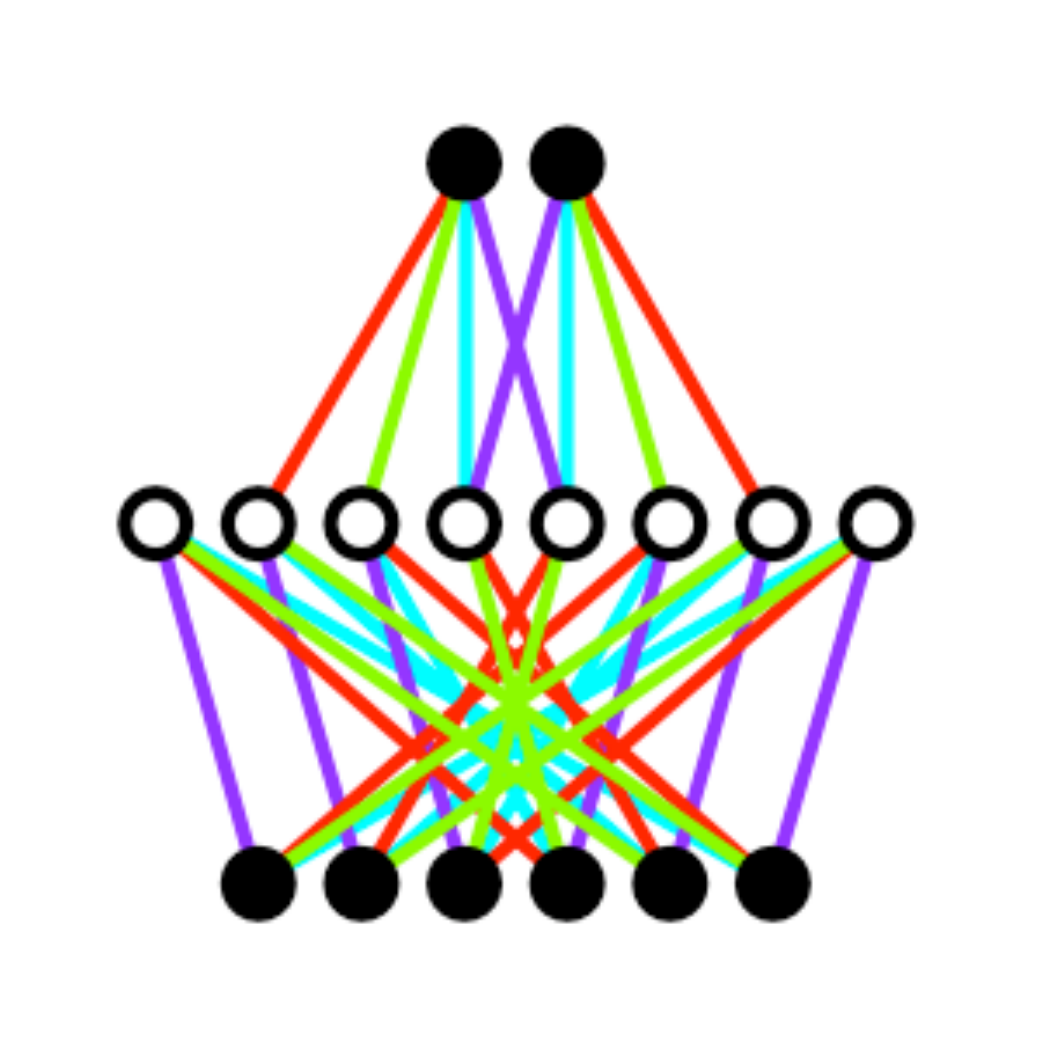}}
      \includegraphics[width=1in]{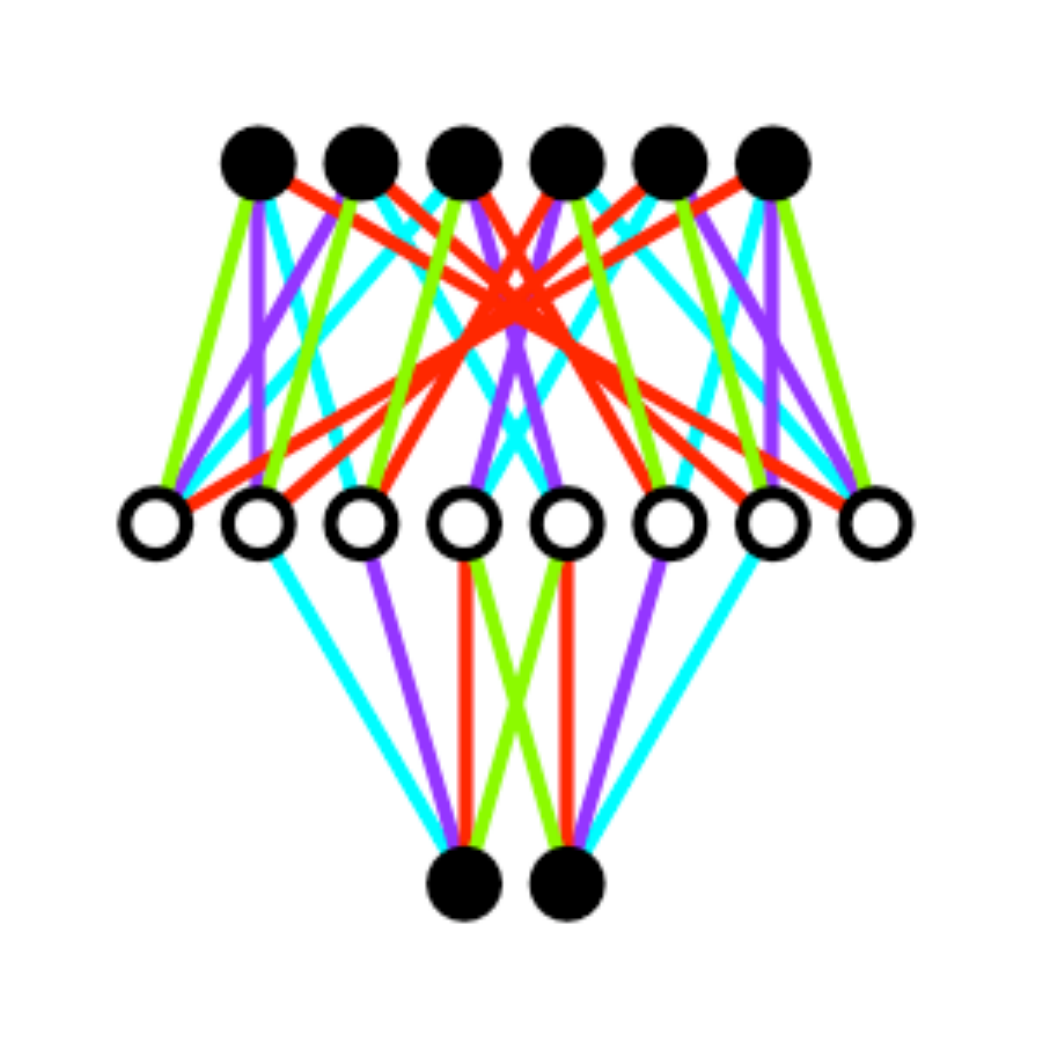}
      \includegraphics[width=1in]{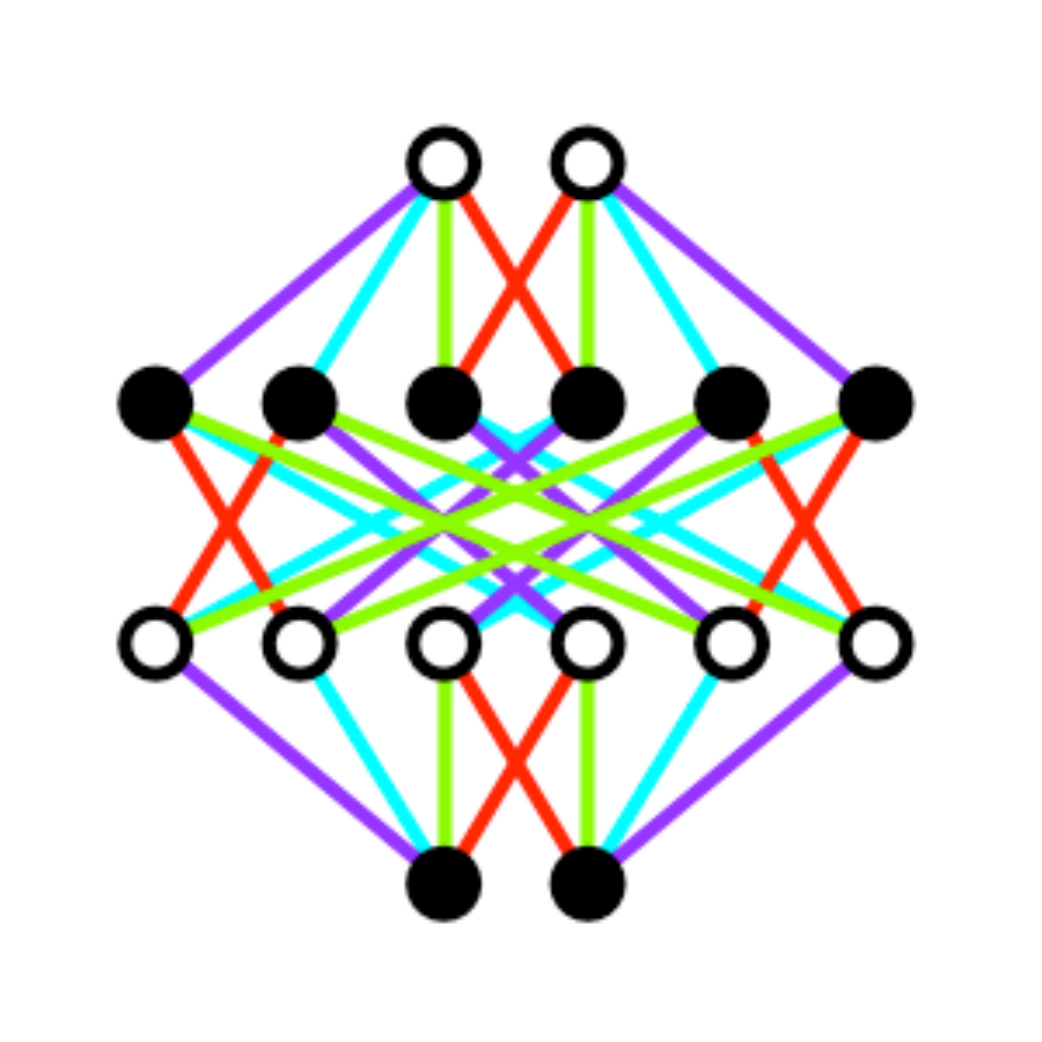} \\
 \cline{2-4}
   & $\Cd{1&1&1&1}$ & $\Cd{1&1&1&1}$
   & \raisebox{-1mm}[1.02\height]{\includegraphics[width=1in]{Pix/A_4_5.pdf}}
      \includegraphics[width=1in]{Pix/A_4_6.pdf} \\
 \cline{2-4}
   & $\Cd{1&1&0&0\\[1pt]1&0&1&0}$ & $\Cd{\text{|}}$
   & \raisebox{-1mm}[1.02\height]{\includegraphics[width=1in]{Pix/A_4_9.pdf}} \\
 \cline{2-4}
   & $\Cd{1&1&0&0\\[1pt]1&1&1&1}$ & $\Cd{1&1&1&1}$
   & \raisebox{-1mm}[1.02\height]{\includegraphics[width=1in]{Pix/A_4_3.pdf}} \\
 \cline{2-4}
   & $\Cd{1&1&0&0\\[1pt]1&0&1&0\\[1pt]1&1&1&1}$ & $\Cd{1&1&1&1}$ 
   & \raisebox{-1mm}[1.02\height]{\includegraphics[width=1in]{Pix/A_4_1.pdf}} \\
 \hline \noalign{\vglue2mm}
\caption{Indecomposable Adinkraic Supermultiplets for $N \leq 4$, their node choice groups (NCG) and doubly-even (DE) maximal subgroups thereof.}
\label{ncstable}
\end{longtable}

\begin{table}[htbp]
  \centering
  \begin{tabular}{>{\centering\baselineskip=10pt\arraybackslash}p{20mm}
                 |>{\centering\baselineskip=10pt\arraybackslash}p{20mm}
                 |>{\centering\baselineskip=10pt\arraybackslash}p{10mm}
                 |>{\centering\baselineskip=10pt\arraybackslash}p{10mm}}
  \bf\boldmath Node Choice Group
 &\bf\boldmath Doubly Even Subgroup
 &\qquad \bf\#
 &\qquad \bf D\,?\\
 \hline\hline
 \rule{0pt}{15pt}%
 $\Cd{\text{|}}$ & $\Cd{\text{|}}$ & 395 & No \\[1mm]
 $\Cd{1&1&0&0&0}$ & $\Cd{\text{|}}$ & 17 & No \\[1mm]
 $\Cd{1&1&1&1&0}$ & $\Cd{\text{|}}$ & 7 & Yes \\[1mm]
 $\Cd{1&1&1&1&0}$ & $\Cd{1&1&1&1&0}$ & 13 & No \\[2mm]
 $\Cd{1&1&0&0&0\\[1pt]1&0&1&0&0}$ & $\Cd{\text{|}}$ & 3 & No \\[2mm]
 $\Cd{1&1&0&0&0\\[1pt]1&1&1&1&0}$ & $\Cd{\text{|}}$ & 3 & Yes \\[2mm]
 $\Cd{1&1&0&0&0\\[1pt]1&1&1&1&0}$ & $\Cd{1 & 1 & 1 & 1 & 0}$ & 3 & No \\[2mm]
 $\Cd{1&1&1&1&0\\[1pt]1&1&1&0 &1}$ &$\Cd{\text{|}}$ & 4 & Yes \\[2mm]
 $\Cd{1&1&1&1&0\\[1pt]1&1&1&0 &1}$ &$\Cd{1&1&1&1&0}$ & 2 & No \\[2.3mm]
    \hline
 \end{tabular}
 \qquad
 \begin{tabular}{>{\centering\baselineskip=10pt\arraybackslash}p{20mm}
                 |>{\centering\baselineskip=10pt\arraybackslash}p{20mm}
                 |>{\centering\baselineskip=10pt\arraybackslash}p{10mm}
                 |>{\centering\baselineskip=10pt\arraybackslash}p{10mm}}
  \bf\boldmath Node Choice Group
 &\bf\boldmath Doubly Even Subgroup
 &\qquad \bf\#
 &\qquad \bf D\,?\\
 \hline\hline
 \rule{0pt}{19pt}%
 $\Cd{1&1&0&0&0\\[1pt]1&0&1&0&0\\[1pt]1&1&1&1&0}$ & $\Cd{\text{|}}$ & 1& Yes \\[4mm]
 $\Cd{1&1&0&0&0\\[1pt]1&0&1&0&0\\[1pt]1&1&1&1&0}$ & $\Cd{1&1&1&1&0}$ &1& No \\[4mm]
 $\Cd{1&1&0&0&0\\[1pt]1&0&1&0&0\\[1pt]0&0&0&1&1}$ & $\Cd{\text{|}}$ &3 & Yes \\[4mm]
 $\Cd{1&1&0&0&0\\[1pt]1&0&1&0&0\\[1pt]0&0&0&1&1}$ & $\Cd{1&1&0&1&1}$ &1& No \\[5mm]
 $\Cd{1&1&0&0&0\\[1pt]0&1&1&0&0\\[1pt]1&1&1&1&0\\[1pt]0&0&0&1&1}$ &$\Cd{\text{|}}$ &1&No
  \\[5mm]
 $\Cd{1&1&0&0&0\\[1pt]0&1&1&0&0\\[1pt]1&1&1&1&0\\[1pt]0&0&0&1&1}$ &$\Cd{1&1&1&1&0}$ &1&Yes
  \\[5mm]
    \hline
\end{tabular}
  \caption{Number (\#) and decomposability (D) of adinkraic supermultiplets for $N = 5$}
\label{N_5_table}
\end{table}
\FloatBarrier

~
\vfill
\begin{flushright}\sl
If you do not expect the unexpected, you will not find it.\\
 -- Aristotle
\end{flushright}
\vfill
\bigskip\paragraph{\bfseries Acknowledgments:}
This research was supported in part by the endowment of the John S.~Toll Professorship, the University of Maryland Center for String \& Particle Theory, National Science Foundation Grant PHY-0354401, and Department of Energy Grant DE-FG02-94ER-40854.  Some Adinkras were drawn with the aid of the {\em Adinkramat\/}~\copyright\,2008 by G.~Landweber.
\vfill

\clearpage
\appendix
\section{Some Linear Algebra Facts About Codes}
\label{s:linalg}
The facts in this appendix are elementary and come from linear algebra, but since linear algebra over $\ZZ_2$ is not as well-known as linear algebra over $\IR$ or $\IC$, it is helpful to point out some properties that are used in this paper.  For the purposes of this paper, all codes are linear binary block codes, even when this is not stated. Codes consist of codewords, which for us are binary $N$-bit numbers, or equivalently, $N$-tuples of 1-bit binary numbers.

First, several results have proofs that are identical to the corresponding ones over $\IR$ or $\IC$, so it will be sufficient to state them, and let the Reader work them out or else consult a linear algebra text\cite{rLinAlg}.
\begin{proposition}
 \label{prop:givens}
Let $C$ be a linear binary code of length $N$.
\begin{itemize}\itemsep=-3pt\vspace{-3mm}
 \item If $g_1,\ldots,g_m$ is a set of linearly independent codewords in $C$, and do not span $C$, then there is another codeword $h\in C$ so that $g_1,\ldots,g_m,h$ is linearly independent.
 \item There exists a linearly independent set that spans $C$.  Such a set is called a {\em basis}, or a {\em generating set}.
 \item Every basis for $C$ has the same number of generators, and this is less than or equal to $N$.  The number of generators is called the {\em dimension} of $C$.
 \item If $g_1,\ldots,g_k$ is a basis for $C$, then every codeword in $C$ can be uniquely written as a linear combination of the $g_1,\ldots,g_k$.
 \item If a subset $S$ of $C$ is linearly independent, the number of elements of $S$ is less than or equal to the dimension of $C$, with equality if and only if $S$ is a basis for $C$.
 \item If $C'$ is a subcode of $C$ and is not equal to $C$, then the dimension of $C'$ is less than the dimension of $C$.
\end{itemize}
\end{proposition}

If $\vec{x}=(x_1,\ldots,x_N)$ and $\vec{y}=(y_1,\ldots,y_N)$ are binary codewords, we can define a $\ZZ_2$-valued inner product $\langle \vec{x}, \vec{y}\rangle$:
\begin{equation}
\langle \vec{x},\vec{y}\rangle = \sum_{i=1}^N x_i y_i \pmod{2}.
\end{equation}
Two codewords are orthogonal if their inner product is zero.  The inner product is bilinear, so if $\vec{x}$ is orthogonal to $\vec{y}_1,\ldots,\vec{y}_m$, then $\vec{x}$ is orthogonal to any linear combination of the $\vec{y}_i$.

By counting 1's, it is easy to see that\cite{rCHVP}:
\begin{proposition}\label{prop:weights}
If $\vec{x}$ and $\vec{y}$ are codewords, then
\begin{eqnarray}
\wt(\vec{x})&\equiv& \langle \vec{x},\vec{x}\rangle \pmod{2},\\
\wt(\vec{x}+\vec{y})&\equiv&\wt(\vec{x})+\wt(\vec{y})-2\langle \vec{x},\vec{y}\rangle\pmod{4}.
\end{eqnarray}
\end{proposition}

From this it follows that the following are equivalent\cite{rCHVP}:
\begin{itemize}\itemsep=-3pt\vspace{-3mm}
\item $C$ is doubly even.
\item Every generating set for $C$ consists of codewords with weight a multiple of 4, and the elements of the generating set are pairwise orthogonal.
\item There exists a generating set for $C$ consisting of codewords with weight a multiple of 4, the elements of which are pairwise orthogonal.
\end{itemize}

\begin{proposition}\label{prop:codeadd}
If $C$ and $D$ are linear binary codes of length $N$, it is possible to choose a generating set $g_1,\ldots,g_k$ for $C$ and a generating set $h_1,\ldots,h_m$ for $D$, so that if $p$ is the dimension of $C\cap D$, then  $g_1=h_1, \ldots, g_p=h_p$, and $g_1,\ldots,g_p$ is a generating set for $C\cap D$.  Also, $g_1,\ldots,g_k,h_{p+1},\ldots, h_m$ is a linearly independent set, and spans a code called $C+D$ of dimension $k+m-p$.
\end{proposition}

\begin{proof}
Take for $g_1,\ldots,g_p$ a basis for $C\cap D$.  Let $h_1=g_1, \ldots, h_p=g_p$.  Use item 1 in Proposition~\ref{prop:givens} to extend $g_1,\ldots,g_p$ to a basis $g_1,\ldots,g_k$ for $C$, and similarly extend $h_1,\ldots,h_p$ to a basis $h_1,\ldots,h_m$ for $D$.

Now to see that $g_1,\ldots,g_k,h_{p+1},\ldots,h_m$ is linearly independent, suppose there existed some $x_1, \ldots, x_k, y_{p+1}, \ldots, y_m$, each in $\{0,1\}$, so that
\begin{equation}
\sum_{i=1}^k x_i g_i  + \sum_{i=p+1}^m y_i h_i = 0.
\end{equation}
We write
\begin{equation}
\sum_{i=1}^k x_i g_i  = \sum_{i=p+1}^m y_i h_i
\end{equation}
(noting that in $\ZZ_2$, $-1=1$).  Define $w$ to be this sum.  By the left hand side, $w$ is in the span of the $g_i$, and so $w\in C$.  By the right hand side, $w$ is in the span of the $h_i$ and so $w\in D$.  Thus, $w\in C\cap D$.  Since this is spanned by $g_1,\ldots,g_p$, and since $g_1,\ldots, g_k$ is a linearly independent set, we see that $x_{p+1},\ldots,x_k$ are all equal to zero.  Likewise, $y_{p+1},\ldots,y_m$ are all zero.  By the right hand side, this means $w=0$.  By the left hand side, and the linear independence of the $g_i$, we have that all the $x_1,\ldots,x_k$ are zero.  Thus, the linear independence of $g_1,\ldots,g_k,h_{p+1},\ldots,h_m$ is established.

The span of it is thus a linear binary code of dimension $k+m-p$ with $g_1,\ldots,g_k,h_{p+1},\ldots,h_m$ as a basis.
\end{proof}

\begin{lemma}
Let $B$ be a linear binary code of length $N$.  Let $C$ be a doubly even subcode of $B$, with generating set $\{c_1,\ldots,c_k\}$.  Suppose $D$ is a doubly even subcode of dimension $m>k$.  Then there is a codeword $d\in D$ so that $\{c_1,\ldots,c_k,d\}$ is linearly independent and spans a doubly even code.
\end{lemma}

\begin{proof}
Use the previous proposition to find a generating set $\hat{c}_1,\ldots,\hat{c}_k$ for $C$ and a generating set $d_1,\ldots,d_m$ for $D$ so that $\hat{c}_1=d_1,\ldots,\hat{c}_p=d_p$ is a generating set for $C\cap D$.

We next note that if $d$ is any codeword in $D$, then $\{c_1,\ldots,c_k,d\}$ is linearly independent if and only if $\{\hat{c}_1,\ldots,\hat{c}_k,d\}$ is; and furthermore, their spans are equal.  This is because each $c_i$ can be written uniquely as a linear combination of the $\hat{c}_1,\ldots,\hat{c}_k$, and vice-versa.  Thus, it suffices to find a codeword $d\in D$ so that $\{\hat{c}_1,\ldots,\hat{c}_k,d\}$ is linearly independent and spans a doubly even code.

\goodbreak
As mentioned above, any generating set for a doubly even code has the following characteristics:
\begin{enumerate}\itemsep=-3pt\vspace{-3mm}
\item Each element has weight a multiple of 4.
\item Each pair of generators is orthogonal.
\end{enumerate}

We look for elements of $D$ of the form
\begin{equation}
d=\sum_{i=p+1}^m x_i d_i,
\end{equation}
where $x_{p+1},\ldots,x_m$ are each in $\{0,1\}$.  Every such linear combination has weight a multiple of 4 since $D$ is doubly even.  Furthermore, if $d$ is any such linear combination, then $\hat{c}_1,\ldots,\hat{c}_k,d$ is linearly independent, by the previous proposition.  So it remains to choose $d$ so that it is orthogonal to each $\hat{c}_i$.

Since codewords in $D$ are orthogonal to each other, and $\hat{c}_1,\ldots,\hat{c}_p$ are in $D$, we have that every codeword in $D$ is orthogonal to $\hat{c}_1,\ldots,\hat{c}_p$.  The condition of being orthogonal to all of the $\hat{c}_i$ for $i>p$ provides $k-p$ homogeneous linear equations in the $m-p$ variables $x_{p+1},\ldots,x_m$.  Since $k-p < m-p$, Gauss--Jordan elimination shows that there must be at least one non-zero solution to this system of equations, and thus, some non-zero
\begin{equation}
d=\sum_{i=p+1}^m x_i d_i
\end{equation}
that is orthogonal to all $\hat{c}_1,\ldots,\hat{c}_k$.
\end{proof}

By applying this inductively, we can prove
\begin{theorem}\label{thm:maxcode}
Given a binary linear code $B$ of length $N$, and $C$ and $D$ doubly even subcodes of $B$, both maximal in the sense that they are not contained in any larger doubly even subcode of $B$, then $C$ and $D$ have the same dimension.
\end{theorem}

\begin{proof}
Let $B$, $C$, and $D$ be given as above.  Suppose the dimension of $C$ and $D$ are different.  Without loss of generality, we may suppose the dimension of $C$ is less than the dimension of $D$.  Then by the previous lemma, $C$ is not maximal.
\end{proof}

Some readers may recognize the formal similarity with matroids\cite{rOxley}.  The set of doubly even subcodes of a given code is not a matroid, but the proof methods are very similar.  In fact, Theorem~\ref{thm:maxcode} corresponds to the well-known statement in matroid theory that every maximal independent set has the same number of elements\cite{rOxley}.\footnote{It may be tempting to define $I$ to be the set of generating sets of doubly even subcodes of $B$, but $(B,I)$ would not then necessarily satisfy the axioms of a matroid.  For instance, take $B$ to be generated by $111100, 001111, 101110$.  There is a doubly even subcode generated by $111100, 001111$.  There is a doubly even subcode generated by $101110$.  But there is no way to adjoin either $111100$ or $001111$ to $101110$ to create a doubly even subcode.  We can, however, take the sum $111100+001111=110011$ and adjoin it to $101110$ and the result generates a doubly even subcode.}

\section{Reducibility {\itshape vs.\/} Decomposability}
 \label{s:red}
Example~\ref{ExCli4d} demonstrates in detail that the supermultiplet $\sM^\diamond_{I^4}$, depicted by the Adinkra\eq{eB14641}, cannot {\em\/decompose\/} into two separate supermultiplets; it is indecomposable. However, it is not {\em\/irreducible\/}: it does contain a sub-supermultiplet, and it is possible to restrict the component fields of $\sM^\diamond_{I^4}$ so that the remaining component fields form a complete $Q$-orbit and so a sub-representation. The complement, however, will turn out not to be definable as a separate supermultiplet. Thus, we show herein that $\sM^\diamond_{I^4}$ may be reduced to a smaller sub-supermultiplet, and so is reducible.

To start, we {\em define\/} a $(1|4|3)$-dimensional supermultiplet of $N=4$ extended supersymmetry:
\begin{subequations}
 \label{e143SS}
\begin{alignat}{5}
&& Q_I\, X_{JK}&= 2\d_{I[J}\dot\c_{K]} + \ve_{IJK}{}^L\,\dot\c_L,\\
 \IX&:\smash{\hbox{$\left\{\vrule width0pt height5.5ex\right.$}}~&
   Q_I\,\c_J&=i\,\d_{IJ}\,\dot x + i \,X_{IJ},
  \smash{\begin{picture}(40,0)
          \put(20,-10){\includegraphics[width=23mm]{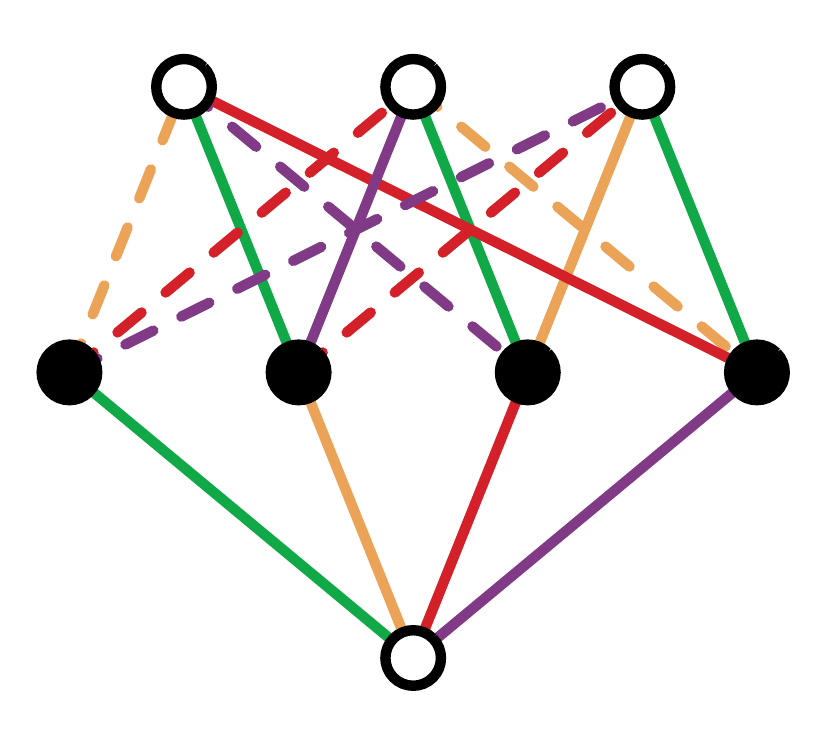}}
          \put(35,-9){\small$x$}
          \put(44,0){\small$\c_I$}
          \put(40,7){\small$X_{IJ}$}
         \end{picture}}\\
&& Q_I\, x&=\c_I
\end{alignat}
\end{subequations}
are the supersymmetry transformation rules, and it is understood that\Ft{It is also possible to use a version of Eq.~(B.2) with a minus on the right-hand side, which would in turn require adjusting a few corresponding signs in Eqs.~(B.1) and~(B.4).}
\begin{equation}
  X_{IJ}=\inv2\ve_{IJ}{}^{KL}X_{KL}.
 \label{e143N=4}
\end{equation}

We also recall the Adinkra\eq{eB14641} and use Eqs.\eq{eN4V->1h} to obtain:
\begin{subequations}
 \label{eM|}
\begin{align}
Q_I\,\cY_{JKLM}&=\d_{I[J}\dot\Y_{KLM]},\\
Q_I\,\Y_{JKL}&=i\,\d_{I[J}\,\dot Y_{KL]} + i\,\cY_{IJKL},\\
 \IY:\smash{\hbox{$\left\{\vrule width0pt height10ex\right.$}}\qquad
Q_I\,Y_{JK}&= \d_{I[J}\dot\h_{K]} + \h_{IJK},
  \smash{\begin{picture}(40,0)
          \put(20,-15){\includegraphics[width=23mm]{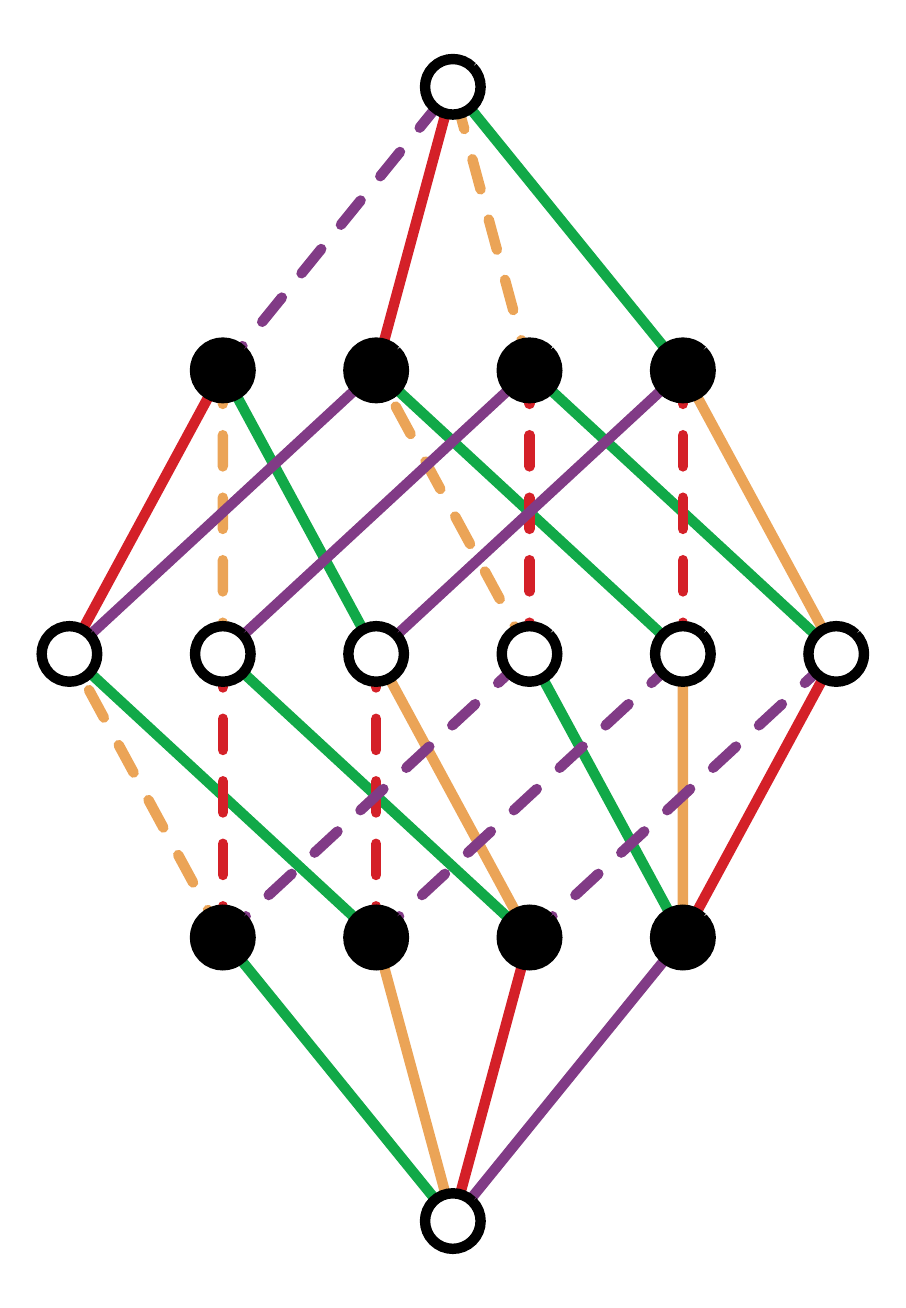}}
          \put(35,-14){\small$y$}
          \put(42,-7){\small$\h_I$}
          \put(44,1){\small$Y_{IJ}$}
          \put(42,8){\small$\Y_{IJK}$}
          \put(35,16){\small$\cY_{IJKL}$}
         \end{picture}}\\
Q_I\,\h_J&=i\,\d_{IJ}\,\dot{y} + i\,Y_{IJ},\\
Q_I\,y&=\h_I,
\end{align}
\end{subequations}
from Eqs.\eq{eM=}. In both Eqs.\eqs{e143N=4}{e143SS} and Eqs\eq{eM|}, we listed the supersymmetry transformation rules for fields in a raising engineering dimension order, so that these would correspond to the nodes at the corresponding heights.

Comparing the two, \eq{e143SS} and\eq{eM|}, we see that it is possible to identify two lowest levels of fields, and that the differences begin at the middle level of\eq{eM|}, as the four $\h_I$ connect to the six $Y_{IJ}$ differently than do the four $\c_I$ connect to the three $X_{IJ}$. To remedy this, we consider imposing the constraint $Y_{IJ}=\inv2\ve_{IJ}{}^{KL}Y_{KL}$ on\eq{eB14641}, and find that the supersymmetry action on\eq{eB14641} then induces the following complete set of constraints:
\begin{subequations}
 \label{eD+}
 \begin{align}
   Y_{IJ}     &=\inv2\ve_{IJ}{}^{KL}\,Y_{KL},\label{eD+2}\\
   \m~:\quad\smash{\hbox{$\left\{\rule{0pt}{5.5ex}\right.$}}\quad
   \Y_{IJK}   &=\ve_{IJK}{}^L\,\dot\h_L,\\
   \cY_{IJKL} &= -\ve_{IJKL}\,\ddot{y}.
 \end{align}
\end{subequations}
On one hand, we note that the so constrained sub-supermultiplet $\m(\IY)$ may be identified with $\IX$. On the other, we note that $\m(\IY)$ is effectively spanned by $(y|\h_I|Y_{IJ})$ with the constraints $Y_{23}=Y_{14}$, $Y_{24}=Y_{13}$ and $Y_{34}=Y_{12}$ imposed to eliminate half of the $Y_{IJ}$. Therefore, $\IX$ may be identified with the sub-supermultiplet $\m(\IY)\subset\IY$, which realizes the embedding $\IX\into\IY$. It is, however, not possible to specify a system of local constraints akin to\eq{eD+}, which would define a sub-supermultiplet effectively spanned by $(Y_{IJ}|\Y_{IJK}|\cY_{IJKL})$, the complement of $\IX\subset\IY$.

Since $\IY$ contains a sub-supermultiplet, it is reducible, albeit indecomposable.

We pause to note that the system of constraints\eq{eD+} has the following graphical equivalent: The six mid-level nodes, representing $Y_{IJ}$, pair up according to the constraint\eq{eD+2}. Each node ``brings along'' its edges from the fermions with the $\inv2$ lower engineering dimension, and so each of the three linearly independent $\m(Y_{IJ})$ connects to each of the four $\h_I$---unlike in Eqs.\eq{eM|}. Since the $\Y_{IJK}$ and $\cY_{IJKL}$ are identified with time-derivatives of $\h_I$ and $y$, respectively, all the edges going upward from the mid-level in $\m(\IY)$ merely duplicate the edges below the mid-level.

Since it is possible to embed $\m:\IX\into\IY$, it is then possible to define a ``gauge transformation'',
\begin{equation}
  \IY \longrightarrow \IY + \x\,\IX,
\end{equation}
and so also the gauge-equivalence quotient,
\begin{equation}
  \IY/\IX \Defl \{\IY\>:~ \IY\simeq\IY+\x\,\IX\},
 \label{eWZ}
\end{equation}
which may be identified with worldline shadow of the ``vector'' supermultiplet of ${\cal N}{=}1$ supersymmetry in 4-dimensional spacetime, in the Wess-Zumino gauge.

The resulting constructions are fairly well known in both superfield and component form. It is gratifying that adinkrammatics---representation in terms of Adinkras---is well capable of describing this. 
 See also Ref.\cite{r6-4.2} for a related but intrinsically 4-dimensional analysis in terms of Adinkras.

Finally, we clarify some of the mathematical details involving sub-supermultiplets, reducibility, and decomposability. First, we need to refine the notion of a map between two supermultiplets. Clearly, such maps must be equivariant with respect to the $N$-extended supersymmetry algebra. However, this condition alone is not sufficient for our purposes. In particular, we could consider the map that takes a each component field of a supermultiplet to its derivative, \ie, $\phi_{i}(\tau) \mapsto \partial_{\tau}\phi_{i}(\tau)$ and $\psi_{i}(\tau) \mapsto \partial_{\tau}\psi_{i}(\tau)$. Such a map would exhibit any supermultiplet as a proper sub-supermultiplet of itself, which is not desirable. In this case, the corresponding quotient supermultiplet would have only finitely many degrees of freedom and no longer be an off-shell supermultiplet. Indeed, the derivative of each component field in the quotient would vanish, violating our off-shell condition that the component fields must not be constrained by any differential equations. To avoid such possibilities, we introduce \emph{strict} maps of off-shell supermultiplets, corresponding to strict homomorphisms of filtered modules via the identification given in Ref.\cite{r6--1}.

\begin{definition}\label{linalg}
A \emph{strict} homomorphism of off-shell supermultiplets is a linear map $\kappa: \mathcal{M}_{1} \to \mathcal{M}_{2}$, which is equivariant with respect to the supersymmetry algebra, and for which the quotient $\mathcal{M}_{2} / \kappa \mathcal{M}_{1}$ is also an off-shell supermultiplet.
\end{definition}

Note that this definition rules out the differentiation map above since the corresponding quotient is not an off-shell supermultiplet. More generally, the strict condition is equivalent to the condition that if a field in $\mathcal{M}_{2}$ is not in the image of $\kappa$, then its derivative cannot be in the image of $\kappa$, either. In order words, a strict homomorphism does not leave orphan degrees of freedom in its cokernel.

\Remk
If a homomorphism $\kappa:\mathcal{M}_{1}\hookrightarrow\mathcal{M}_{2}$ of off-shell supermodules is \emph{injective}, \ie, we have an embedding of $\mathcal{M}_{1}$ in $\mathcal{M}_{2}$, then the strict condition is equivalent to the condition that no component field of $\mathcal{M}_{1}$ can map to the derivative of a field in $\mathcal{M}_{2}$ given by the linear combination of the component fields of $\mathcal{M}_{2}$ and their derivatives. In Adinkra language, this means roughly that vertices of one Adinkra map to linear combinations of vertices of the other Adinkra.

The injective case is of greatest interest to us, as provides us with a suitable notion of sub-supermultiplets.

\begin{definition}
A \emph{strict} sub-supermultiplet of an off-shell supermultiplet $\mathcal{M}$ is an off-shell sub-supermultiplet $\mathcal{N} \subset \mathcal{M}$ such that the quotient $\mathcal{M}/\mathcal{N}$ is also an off-shell supermultiplet, or equivalently, such that the inclusion map $\mathcal{N}\hookrightarrow \mathcal{M}$ is a strict homomorphism.
\end{definition}

Given a strict sub-supermultiplet $\mathcal{N}$ contained in $\mathcal{M}$, we obtain a short exact sequence of off-shell supermultiplets
\begin{equation}\label{eq:exact}
0 \xrightarrow{\phantom{\;\alpha\;}} \mathcal{N} \xrightarrow{\;\alpha\;} \mathcal{M} \xrightarrow{\;\beta\;} \mathcal{M}/\mathcal{N} \xrightarrow{\phantom{\;\alpha\;}}
 0 \,.
\end{equation}
We say that the off-shell supermultiplet $\mathcal{M}$ is \emph{reducible} because it contains a strict sub-supermultiplet. In the representation theory of Lie groups or Lie algebras, reducible implies decomposable, \ie, if a representation $V$ has an invariant subspace or subrepresentation $U$, then there exists a complementary subrepresentation $W\subset V$ such that $V = U \oplus W$. We get such complements for free in Lie theory since these representations admit ad-invariant inner products, with respect to which the orthogonal complement of an invariant subspace is also an invariant subspace. However, this does not hold in general for representations of superalgebras, and we have no such inner product to determine our complements. So, as we saw earlier in this appendix, one can have reducible off-shell supermultiplets which nevertheless are not decomposable. In the decomposable case, we have an additional piece of information: a splitting of the exact sequence \eqref{eq:exact}, which is a map $\gamma:\mathcal{M}/\mathcal{N} \to \mathcal{M}$ such that the composition
$$ \mathcal{M}/\mathcal{N} \xrightarrow{\;\gamma\;} \mathcal{M} \xrightarrow{\;\beta\;} \mathcal{M}/\mathcal{N} $$
is the identity map, or in other words the quotient map $\beta$ has a right inverse $\gamma$. Such a splitting guarantees that we have a decomposition
$$\mathcal{M} = \mathcal{N} \oplus \,\text{im}(\mathcal{M}/\mathcal{N}),$$
or eqivalently that $\mathcal{N}$ has a complement in $\mathcal{M}$. For a general discussion of complemented submodules, see \cite[Chapter 4]{rLinAlg}.

Given the combinatorially growing number of chromotopologies\cite{r6-3}, and the combinatorially growing number of inequivalent ``hangings'' of Adinkras with any of the given chromotopologies, it may seem bewildering that Definition~\ref{linalg} indicates an even further increase in the number of off-shell supermultiplets which can be described in terms of Adinkras. In this sense, the Adinkras and adinkraic supermultiplets, as plentiful as they turn out to be, are merely to be regarded as the simpler building blocks in the toolbox of supersymmetry.

\vfill\clearpage
\bibliographystyle{elsart-numX}
\bibliography{Refs}
\end{document}